%% file: main.tex
\pdfoutput=1
\documentclass[acmsmall,screen]{acmart}
\usepackage[linesnumbered, ruled]{algorithm2e}
\usepackage{balance}
\usepackage{flushend}
\usepackage{xcolor}
\usepackage{tikz}
\usepackage{pgfplots}
\pgfplotsset{compat=1.16}
\usetikzlibrary{patterns,arrows.meta,positioning}
\usepackage{subfig}
\usepackage{xspace}
\usepackage{enumitem}
\usepackage{bbm}
\usepackage{esvect}
\usepackage{makecell}

\newenvironment{customlegend}[1][]{%
    \begingroup
    \csname pgfplots@init@cleared@structures\endcsname
    \pgfplotsset{#1}%
}{%
    \csname pgfplots@createlegend\endcsname
    \endgroup
}%

\def\addlegendimage{\csname pgfplots@addlegendimage\endcsname}

\newcommand{\revise}[1]{#1}

\newcommand{\eat}[1]{}

\definecolor{myblue}{HTML}{0072bd}
\definecolor{my_blue}{HTML}{a2142f}
\definecolor{my_cyan}{HTML}{d95319}
\definecolor{my_teal}{HTML}{77ac30}
\definecolor{my_purple}{HTML}{4dbeee} %{edb120}
\definecolor{my_violet}{HTML}{7e2f8e}
\definecolor{my_green}{HTML}{008080}

\definecolor{my_red}{HTML}{CC79A7}
\definecolor{my_yellow}{HTML}{117733}
\definecolor{my_pp}{HTML}{785ef0}

\def\header{\vspace{1mm} \noindent}

\newcommand{\algo}{$\mathtt{GEER}$\xspace}
\newcommand{\bwd}{$\mathtt{SMM}$\xspace}
\newcommand{\fwd}{$\mathtt{AMC}$\xspace}
\newcommand{\AM}{\mathbf{A}}
\newcommand{\DM}{\mathbf{D}}
\newcommand{\PM}{\mathbf{P}}

\makeatletter
\def\smallunderbrace#1{\mathop{\vtop{\m@th\ialign{##\crcr
   $\hfil\displaystyle{#1}\hfil$\crcr
   \noalign{\kern3\p@\nointerlineskip}%
   \tiny\upbracefill\crcr\noalign{\kern3\p@}}}}\limits}
\makeatother

\let\oldnl\nl% Store \nl in \oldnl
\newcommand{\nonl}{\renewcommand{\nl}{\let\nl\oldnl}}% Remove line number for one line

\AtBeginDocument{%
  \providecommand\BibTeX{{%
    \normalfont B\kern-0.5em{\scshape i\kern-0.25em b}\kern-0.8em\TeX}}}

\setcopyright{acmlicensed}
\acmJournal{PACMMOD}
\acmYear{2023} \acmVolume{1} \acmNumber{1} \acmArticle{16} \acmMonth{5} \acmPrice{15.00}\acmDOI{10.1145/3588696}

\usepackage[utf8]{inputenc} % allow utf-8 input
\usepackage[T1]{fontenc}    % use 8-bit T1 fonts
\usepackage{microtype}      % microtypography
\SetCommentSty{textsf}
\begin{document}

\title{Efficient Estimation of Pairwise Effective Resistance}


\author{Renchi Yang}
\orcid{0000-0002-7284-3096}
\authornote{Work was done when the first author was a Research Fellow at NUS, Singapore.}
\affiliation{%
	\department{Department of Computer Science}
  \institution{Hong Kong Baptist University}
  \country{Hong Kong SAR}
}
\email{renchi@hkbu.edu.hk}

\author{Jing Tang}
\orcid{0000-0002-0785-707X}
\authornote{Corresponding author: Jing Tang.}
\affiliation{%
  \institution{\hbox{The Hong Kong University of Science and Technology (Guangzhou)}}
  \country{China}
}
\affiliation{%
	\institution{The Hong Kong University of Science and Technology}
	\country{Hong Kong SAR}
}
\email{jingtang@ust.hk}

\input{tex/abstract}




\maketitle

\begin{sloppy}
\input{tex/introduction}

\input{tex/preliminary}
\input{tex/proposed-new}

\input{tex/experiment}
\input{tex/relatedwork}

\input{tex/conclusion}

\begin{acks}
This work is partially supported the National Natural Science Foundation of China (NSFC) under Grant No. U22B2060, by Guangzhou Municipal Science and Technology Bureau under Grant No. 2023A03J0667, by HKUST(GZ) under a Startup Grant, and by A*STAR Singapore under Grant A19E3b0099. 
The findings herein reflect the work, and are solely the responsibility, of the authors.
% Any opinions, findings, and conclusions or recommendations expressed in this material are those of the authors and do not necessarily reflect the views of the funding agencies and organizations.
\end{acks}

\appendix
\input{tex/appendix}

% \clearpage
% \newpage
% \pagebreak

\bibliographystyle{ACM-Reference-Format}
\bibliography{main}
\balance
\end{sloppy}

\received{April 2022}
\received[revised]{July 2022}
\received[accepted]{August 2022}

\end{document}

%% file: tex/abstract.tex
\begin{abstract}
Given an undirected graph $G$, the effective resistance $r(s,t)$ measures the dissimilarity of node pair $s,t$ in $G$, which finds numerous applications in real-world scenarios, such as recommender systems, combinatorial optimization, molecular chemistry and electric power networks. 
Existing techniques towards pairwise effective resistance estimation either trade approximation guarantees for practical efficiency, or vice versa. 
In particular, the state-of-the-art solution is based on a multitude of Monte Carlo random walks, rendering it rather inefficient in practice especially on large graphs. 

Motivated by this, this paper first presents an improved Monte Carlo approach, \fwd, which reduces both the length and amount of random walks required without degrading the theoretical accuracy guarantee, through careful theoretical analysis and an adaptive sampling scheme. 
Further, we develop a greedy approach, \algo, which combines \fwd with sparse matrix-vector multiplications in an optimized and non-trivial way. \algo offers significantly improved practical efficiency over \fwd without compromising its asymptotic performance and accuracy guarantees. Extensive experiments on multiple benchmark datasets reveal that \algo is orders of magnitude faster than the state of the art in terms of computational time when achieving the same accuracy.
\end{abstract}

%% file: tex/introduction.tex
\section{Introduction}
Given an undirected graph $G$ and two nodes $s,t\in G$, the {\em effective resistance} (ER)~\cite{tetali1991random} $r(s,t)$ is proportional to the expected number of steps taken by a random walk starting at $s$ visits node $t$ and then goes back to $s$ (i.e., the {\em commute time} \cite{chandra1996electrical} between $s$ and $t$). Intuitively, a small $r(s,t)$ suggests a high similarity between nodes $s$ and $t$.
\revise{Due to the clear physical meaning and unique properties, ER finds a variety of applications in practical problems.
For instance, in the study of electric circuits \cite{nilsson2008electric}, resistive electrical networks are viewed as graphs where each edge corresponds to a resistor.
The definition of ER $r(s,t)$ is basically %the port resistance
the voltage developed between nodes $s$ and $t$ when a unit current is injected at one and extracted at the other \cite{tetali1991random}.
Intuitively, this physical interpretation benefits analyzing cascading failures \cite{guo2017monotonicity,soltan2015analysis} and power network stability \cite{dorfler2010spectral,song2017network,song2019extension} in electric power networks. 
Moreover, \citet{spielman2011graph} theoretically proved that sampling edges in a graph $G$ according to their ER values yields a high-quality sparsified graph with a rigorous accuracy guarantee. Such graph sparsifiers are used as building blocks to accelerate cut approximation \cite{benczur1996approximating}, max-flow algorithms \cite{christiano2011electrical,kelner2014almost}, and solving linear systems \cite{spielman2004nearly}.
In recent years, ER is further applied to many data mining problems such as graph clustering \cite{alev2018graph,qiu2007clustering,yen2005clustering}, recommender systems \cite{kunegis2007collaborative,fouss2007random}, and image segmentation \cite{qiu2005image,behmo2008graph}.
Additionally, the data management community adopts ER in developing efficient graph systems \cite{qiu2021lightne,zhao2013embeddability} and applications \cite{shi2014density,sricharan2014localizing,yin2012challenging}.}

As explained in Section \ref{sec:pd}, the exact computation of ER incurs vast computational costs for large graphs; and thus, the majority of existing work (e.g., \cite{hayashi2016efficient,aybat2017decentralized,jambulapati2018efficient,cohen2014solving,broberg2020prediction,pankdd21}) focuses on computing approximate ER. In the meanwhile, it is prohibitively expensive to materialize the ER values of all possible node pairs in a large graph due to the colossal space, i.e., $O(n^2)$, where $n$ is the number of nodes in $G$. To this end, recent work \cite{cohen2014solving,pankdd21} on ER computation answers the {\em $\epsilon$-approximate pairwise effective resistance} (denoted hereinafter as $\epsilon$-approximate PER) query for each specified node pair individually. More precisely, given an additive error $\epsilon$ and a node pair $(s,t)$ in the input graph $G$, an $\epsilon$-approximate PER query returns an approximate ER $r^{\prime}(s,t)$ with at most $\epsilon$ additive error in it. Even under the approximation definition, $\epsilon$-approximate PER computation remains tenaciously challenging to address owing to its complicated definitions by either the pseudo inverse of matrix or random walks.

%Recently, 
\citet{andoni2018solving}~introduced an algorithm for answering $\epsilon$-approximate PER queries on regular expander graphs using a time complexity of $\textstyle O\left(\tfrac{1}{\epsilon^2} \mathrm{polylog}(\tfrac{1}{\epsilon})\right)$. \citet{spielman2011graph} described a random projection-based algorithm, which allows for the $\epsilon$-approximate PER query of any node pair in $G$ to be answered in $O(\log{n})$ time. However, their algorithm requires $\tilde{O}(m/\epsilon^2)$ time ($m$ is the number of edges in $G$) for constructing a $\textstyle ({24\log{n}}/{\epsilon^2})\times n$ matrix in the preprocessing step, which is prohibitive for large graphs. Very recently, a suite of randomized techniques without any preprocessing are proposed in \cite{pankdd21} for addressing $\epsilon$-approximate PER computation on general graphs with provable performance guarantees. 
Among them, $\mathtt{MC}$ simulates numerous random walks from source node $s$ to estimate $r(s,t)$ according to its {\em commute time}-based interpretation mentioned previously. Nevertheless, this way is highly inefficient in practice as its random walks explore the whole graph, resulting in a large time complexity of $\textstyle O\left(\frac{m\cdot d(s)}{\epsilon^2}\right)$, where $d(s)$ represents the degree of node $s$.  To circumvent this problem, Peng et al. \cite{pankdd21} subsequently propose a new $\epsilon$-approximate PER solution $\mathtt{TP}$ based on another random walk-based interpretation of ER. This interpretation enables the derivation of a maximum length for random walks, and thus $\mathtt{TP}$ only explores a small portion of $G$. In turn, $\mathtt{TP}$ runs in $\textstyle O\left( \tfrac{1}{\epsilon^2}  \log^4{(\tfrac{1}{\epsilon})}\right)$ expected time regardless of the size of $G$.
$\mathtt{TPC}$ subsequently ameliorates $\mathtt{TP}$ by leveraging the idea of stitching short random walks to longer ones, yielding a time cost of $\textstyle O\left( \tfrac{1}{\epsilon^2}  \log^3{(\tfrac{1}{\epsilon})}\right)$ on expander graphs. 
Although these algorithms achieve significant improvements upon previous solutions in terms of asymptotic performance, they still entail an exorbitant amount of time in practice, even on small graphs. This is caused by the sheer number of long random walks required by them, leaving much room for improvement regarding computational costs.

\begin{table}
\centering
\caption{Algorithms for $\epsilon$-approximate PER Queries}
\label{tab:complexity}\vspace{-3mm}
\begin{small}
\begin{tabular}{cl}
\toprule
\bf Algorithm  & \bf Time Complexity \\
\midrule
$\mathtt{TP}$ \cite{pankdd21} &  $\textstyle O\left( \tfrac{1}{\epsilon^2}  \log^4{(\frac{1}{\epsilon})}\right)$ \\
$\mathtt{TPC}$ \cite{pankdd21} &  $\textstyle O\left( \tfrac{1}{\epsilon^2}  \log^3{(\frac{1}{\epsilon})}\right)$ on  expander graphs \\
$\mathtt{MC}$ \cite{pankdd21}  &  $\textstyle O\left(\tfrac{m\cdot d(s)}{\epsilon^2}\right)$  \\
% $\mathtt{MC2}$ &   \\
% $\mathtt{ST}$ \cite{pankdd21}  &   $\textstyle O\left(\frac{1}{\epsilon^5}+\frac{\log^2{n}}{\epsilon^2}\right)$ \\
our \fwd and \algo & $\textstyle O \left(\tfrac{1}{\epsilon^2 d^2} \log^3{\left(\tfrac{1}{\epsilon d}\right)}\right)$ \\
\bottomrule
\end{tabular}
\end{small}
\vspace{-1mm}
\end{table}

\header
{\bf Contributions.}
This paper first presents \fwd, an adaptive Monte Carlo approach that overcomes the deficiencies of $\mathtt{TP}$ and $\mathtt{TPC}$, while retaining the accuracy assurance. More precisely, \fwd derives a refined maximum random walk length through a rigorous theoretical analysis. Moreover, instead of simulating a fixed number of random walks, \fwd conducts random walks in an adaptive fashion by leveraging the empirical Bernstein inequality to facilitate early termination. 
\revise{In particular, our refined maximum random walk length $\ell$ is often less than half of its counterpart in prior work \cite{pankdd21} and the well-thought-out adaptive scheme in \fwd leads to over $20\ell\times$ reduction in random walk samples, compared to the state-of-the-art Monte Carlo approach $\mathtt{TP}$.}
Unfortunately, despite the aforementioned optimizations in \fwd enable a significant reduction in computational cost,  its practical efficiency is still less than satisfactory, as revealed in our empirical studies. For this purpose, we develop our main proposal \algo for improved practical efficiency. \revise{Under the hood, the core of \algo is to integrate sparse matrix-vector multiplications into \fwd in an optimized and non-trivial way}, thereby overcoming the drawbacks of both and achieving significantly enhanced speedup compared to \fwd without degrading the theoretical guarantees.
\revise{Table \ref{tab:complexity} compares the time complexities of our proposed solutions against three approximate algorithms to answer $\epsilon$-approximate PER queries with a high probability. We can see that both \fwd and \algo achieve the best asymptotic performance, i.e., $\textstyle O \left(\tfrac{1}{\epsilon^2 d^2} \log^3{\left(\tfrac{1}{\epsilon d}\right)}\right)$, which theoretically improves the state of the art by a large factor of ${d^2}\log{\left(\tfrac{1}{\epsilon d}\right)}$, where $d$ equals $\min\{d(s),d(t)\}$ and $d(s)$ (resp.\ $d(t)$) represents the degree of node $s$ (resp.\ $t$).} In practice, on the the widely studied Friendster dataset, \algo is up to $38.2\times$ faster than \fwd and achieves more than three orders of magnitude enhancements over the state of the art in query performance. Briefly, our key contributions in this paper can be summarized as follows:
\begin{itemize}[leftmargin=*]
\item We propose \fwd, an adaptive Monte Carlo algorithm runs in  $\textstyle O \left(\tfrac{1}{\epsilon^2 d^2} \log^3{\left(\tfrac{1}{\epsilon d}\right)}\right)$ expected time and returns an $\epsilon$-approximate ER value with a high probability.
\item We further develop \algo, which offers significantly improved empirical efficiency over \fwd without \revise{compromising} the asymptotic guarantees.
\item Extensive experiments using real datasets demonstrate that on a single commodity server \algo consistently outperforms the state-of-the-art solution often by orders of magnitude speedup.
\end{itemize}

The remainder of the paper is organized as follows. In Section \ref{sec:pre}, we provide notations and formally define the problem. We present our first-cut solution \fwd with theoretical analysis in Section \ref{sec:amc}, and further propose an improved solution \algo in Section \ref{sec:main}. Our solutions and existing methods are evaluated in Section \ref{sec:exp}. Related work is reviewed in Section \ref{sec:related}. Finally, Section \ref{sec:conclude} concludes the paper. All proofs of theorems and lemmas appear in Appendix \ref{sec:proof}.

%% file: tex/preliminary.tex
\section{Preliminaries}\label{sec:pre}
\revise{This section first defines common notations, and then provides the formal problem definition and revisits main competitors.}

\begin{table}
	\centering
	\caption{Frequently used notations}
	\label{tab:notation}\vspace{-3mm}
% 	\resizebox{\columnwidth}{!}{%
		\begin{small}
%		\begin{tabular}{p{0.56in}l p{2.4in}}
		\begin{tabular}{rc p{4.2in}}	
			\toprule
			\bf Notation & & \bf Description \\
			\midrule
			$G=(V,E)$ & $\triangleq$ & An undirected and unweighted graph with node set $V$ and edge set $E$. \\
			$n, m$ & $\triangleq$ & The numbers of nodes and edges in graph $G$, respectively.\\
			$\PM$ & $\triangleq$ & The transition matrix of $G$.\\
			$d(v)$ & $\triangleq$ & The degree of node $v$.\\
			$r(s,t)$ & $\triangleq$ & The effective resistance between nodes $s,t$.\\
			$r^{\prime}(s,t)$ & $\triangleq$ & The approximate version of $r(s,t)$.\\
			$r_{\ell}(s,t)$ & $\triangleq$ & The $\ell$-truncated version of $r(s,t)$ (see Eq. \eqref{eq:r-ell}).\\
			% $\vec{\mathbf{e}}_v$ & a one-hot vector\\
			$\lambda_i$ & $\triangleq$ & The $i$-th largest eigenvalue of $\PM$.\\
			% $\lambda$ & \\
			$p_i(u,v)$ & $\triangleq$ & The probability of a length-$i$ random walk starting from node $u$
			ending at node $v$. \\
			$\epsilon$, $\delta$ & $\triangleq$ & The absolute error threshold and the failure probability, respectively.\\
% 			$\delta$ & $\triangleq$ & The failure probability.\\
			$\tau$ & $\triangleq$ & The number of batches of random walks in \fwd. \\
			$\ell$ & $\triangleq$ & The maximum random walk length (see Eq. \eqref{eq:ell}).\\
			$\ell_f$ & $\triangleq$ & The maximum random walk length in \fwd.\\
			$\ell_b$ & $\triangleq$ & The number of iterations in \bwd.\\
			$\min(\vec{\mathbf{x}})$ & $\triangleq$ & The smallest value in vector $\vec{\mathbf{x}}$. \\
			$\max_k(\vec{\mathbf{x}})$ & $\triangleq$ & The $k$-th largest value in vector $\vec{\mathbf{x}}$.\\
			\bottomrule
		\end{tabular}
		\end{small}
% 	}
%\vspace{-1mm}
\end{table}

\subsection{Notations}\label{sec:notation}
We denote vectors $\vec{\mathbf{x}}$ and
matrices $\mathbf{X}$ in bold lower and upper cases, repsectively, and use $\vec{\mathbf{x}}(i)$ to represent the $i$-th element of vector $\vec{\mathbf{x}}$. In addition, ${\vec{\mathbf{e}}}_i$ is used to denote a one-hot vector, which has value 1 at entry $i$ and 0 everywhere else. Row-wise (i.e., $\mathbf{X}(i)$
), column-wise (i.e., $\mathbf{X}(,j)$), and element-wise (i.e., $\mathbf{X}(i,j)$) matrix indexing represent the $i$-th row vector, $j$-th column vector, and the scalar at $i$-th row and $j$-th column of $\mathbf{X}$, respectively. We denote by $\min(\vec{\mathbf{x}})$ and $\max_k(\vec{\mathbf{x}})$ the smallest value and the $k$-th largest value in vector $\vec{\mathbf{x}}$, respectively.

We are given an undirected and unweighted graph $G=(V,E)$ with $|V|=n$ nodes and $|E|=m$ edges. For each node $v\in V$, we denote by $d(v)$ its degree, i.e., the number of neighbors connecting to $v$. Let $\AM$ and $\DM$ be the {\em adjacency matrix} and {\em diagonal degree matrix} of $G$, respectively; i.e., $\AM(u,v)=1$ iff $(u,v)\in E$ and each diagonal entry $\DM(v,v)=d(v)\ \forall{v\in V}$. The random walk matrix (a.k.a. transition matrix) of $G$ is then defined as $\mathbf{P}=\DM^{-1}\AM$; i.e., $\PM(u,v)=\frac{1}{d(v)}$ if $(u,v)\in E$. Accordingly, $p_i(u,v)={\mathbf{P}^i}(u,v)$ signifies the probability of a length-$i$ simple random walk starting from node $u$ ending at node $v$. {For simplicity, we henceforth refer to a simple random walk as a random walk, which at each visited node $v$, moves to a neighbor of $v$ with probability $\frac{1}{d(v)}$ \cite{lovasz1993random}.} Throughout this paper, we assume that $G$ is connected and non-bipartite, 
inducing that $\mathbf{P}$ is ergodic; in other words, $\lim\limits_{i\to \infty} {\vec{\mathbf{e}}}_s\mathbf{P}^i=\vec{\boldsymbol{\pi}}\ \forall{s\in V}$, where $\vec{\boldsymbol{\pi}}(v)=\frac{d(v)}{2m}\ \forall{v\in V}$ \cite{motwani1995randomized}. 
% In addition, we denote by $\mathbf{L}^{\dagger}$ the Moore-Penrose pseudo-inverse of the Laplacian of $G$ (i.e., $\DM-\AM$). 
Let $\lambda_1,\lambda_2,\dotsc,\lambda_n$ be the eigenvalues of $\mathbf{P}$, sorted by algebraic value in descending order, i.e., $1 = \lambda_1 \geq \lambda_2 \geq \dotsb \geq \lambda_n \geq -1$.
% Let .  $\lambda=\max\{|\lambda_2|,|\lambda_n|\}$.
% Let $\tau$ be the mixing time of $\mathbf{P}$.
Let $\mathbbm{1}_{s\neq t}$ be an indicator function, which is equal to $1$ if $s\neq t$ and otherwise $0$. Table \ref{tab:notation} lists the frequently used notations throughout this paper.

\subsection{Problem Definition}\label{sec:pd}

\begin{definition}[Effective Resistance (ER)]\label{def:er}
Given graph $G$ and two nodes $s,t\in G$, the effective resistance between $s$ and $t$ is defined as
\begin{equation}\label{eq:er}
r(s,t)=(\vec{\mathbf{e}}_s-\vec{\mathbf{e}}_t)\cdot \mathbf{L}^{\dagger}\cdot (\vec{\mathbf{e}}_s-\vec{\mathbf{e}}_t)^{\intercal},
\end{equation}
where $\mathbf{L}^{\dagger}$ denotes the Moore-Penrose pseudo-inverse of $\mathbf{D}-\mathbf{A}$.
\end{definition}

Definition \ref{def:er} presents the formal definition of effective resistance (ER). In particular, the ER $r(s,t)$ of node pair $(s,t)$ is $0$ when $s=t$, which indicates that any node $s\in G$ has zero dissimilarity to itself. In the literature \cite{nash1959random,tetali1991random}, ER is closely related to {\em commute time}. More precisely, $r(s,t)=\frac{c(s,t)}{2m}$, where $c(s,t)$ stands for the commute time between nodes $s$ and $t$, i.e., the expected length of a random walk from $s$ to visit $t$ and return back to $s$. 
% According to Corollary 6.7 in \cite{motwani1995randomized}, $c(s,v)<n^3$; and hence $r(s,t)<\frac{n^3}{2m}$.
% According to Lemma 6.5 \cite{motwani1995randomized}, For any edge $(s,t)\in E$, $c(s,t)\le 2m$ and thus $0\le r(s,t)\le 1$.

The exact computation of $r(s,t)$ based on Definition \ref{def:er} is immensely expensive as it requires inverting the Laplacian matrix of $G$ with a time complexity of $O(n^{2.3727})$ \cite{cormen2009introduction,raz2002complexity}. 
% Another possible solution is to compute the commute time $c(s,t)$, which involves sampling random walks of indefinite length, and thus, is not feasible. 
Prior work \cite{pankdd21} pertaining to ER computation proposes to compute $\epsilon$-approximate ER. In particular, we say $r^{\prime}(s,t)$ an $\epsilon$-approximate ER of $r(s,t)$ if it satisfies
\begin{equation}\label{eq:epsilon}
|r(s,t)-r^{\prime}(s,t)|\le \epsilon.
\end{equation}
In this paper, we focus on addressing the $\epsilon$-approximate PER query, as defined in Definition \ref{def:problem}.
\begin{definition}[$\epsilon$-approximate PER Query]\label{def:problem}
Given an undirected graph $G$ and an arbitrarily small additive error $\epsilon$, the $\epsilon$-approximate PER query returns an $\epsilon$-approximate ER value $r^{\prime}(s,t)$ for any specified node pair $(s,t)$.
\end{definition}

\subsection{Main Competitors}\label{sec:mcpt}
In the rest of the section, we briefly overview four recent solutions proposed in \cite{pankdd21} for answering $\epsilon$-approximate PER queries and then discuss their limitations.

\subsubsection{\bf The \texttt{MC} and \texttt{MC2} Algorithms}
% \header
% {\bf $\boldsymbol{\mathtt{MC}}$.}
Using the connection between ER and commute time, $\mathtt{MC}$ \cite{pankdd21} computes approximate $r(s,t)$ by first estimating the commute time $c(s,t)$ of node pair $(s,t)$. Recall that $c(s,t)$ is defined as the expected length of a random walk starting from $s$ to visit $t$ and return back to $s$. $\mathtt{MC}$ conducts Monte Carlo simulations using $\eta$ random walks. That is, it counts the number ${\eta}_r$ of random walks traverses from $s$ to $t$ and back, and then computes an approximate ER $r^{\prime}(s,t)$ as $\tfrac{\eta}{d(s)\cdot {\eta}_r}$. Under the assumption $r(s,t)\le \gamma$, \citet{pankdd21} showed when $\eta=\frac{3\gamma d(s)\cdot\log{(1/\delta)}}{\epsilon^2}$, $r^{\prime}(s,t)$ is an $\epsilon$-approximate ER with a probability at least $1-\delta$. The running time of $\mathtt{MC}$ is $\textstyle O\left(\tfrac{md(s)\gamma^2}{\epsilon^2}\right)$. Notice that $c(s,t)<n^3$, and thus $r(s,t)<\frac{n^3}{2m}$ (Corollary 6.7 in \cite{motwani1995randomized}). Therefore, in the worst case, the time complexity of $\mathtt{MC}$ can be up to $O\left(\tfrac{n^6 d(s)}{m\epsilon^2}\right)$.

Regarding the special case where the query node pair $(s,t)$ is connected by an edge in $G$ (i.e., $(s,t)\in E$), \citet{pankdd21} presented $\mathtt{MC2}$ to support more efficient $\epsilon$-approximate PER queries. Specifically, $\mathtt{MC2}$ relies on the special definition of $r(s,t)$ when $(s,t)\in E$; that is, $r(s,t)$ is the probability that a random walk started at $s$ visits $t$ for the first time using the edge $(s,t)$. Accordingly, $\mathtt{MC2}$ estimates $r(s,t)$ by performing such Monte Carlo random walks. By \cite{pankdd21}, with the assumption $r(s,t)> \gamma$, $\mathtt{MC2}$ returns $\epsilon$-approximate ER $r^{\prime}(s,t)$ with a probability at least $1-\delta$ after $\tfrac{3\log{(1/\delta)}}{\epsilon^2\gamma}$ random walks are conducted. Using the fact that $\tfrac{1}{2m}\le r(s,t)\le 1$, $\forall{(s,t)\in E}$ (Lemma 6.5 in \cite{motwani1995randomized}), the total number of random walks required in $\mathtt{MC2}$ is bounded by $\frac{6m\log{(1/\delta)}}{\epsilon^2}$.

% \subsubsection{\bf Spanning trees-based Approach}
% Given a graph $G$, a {\em spanning tree} $T(G)$ is a subgraph of $G$, which includes all the nodes of $G$ with a minimum possible number of edges \cite{cheriton1976finding}. Corollary 4.2 in \cite{lovasz1993random} proves that $r(s,t)=\frac{|T(G^{\prime})|}{|T(G)|}$, where $G^{\prime}$ be a subgraph of $G$ obtained by identifying $s$ and $t$. $\mathtt{ST}$ \cite{pankdd21} is based on the idea of approximating $r(s,t)$ by estimating $|T(G^{\prime})|$ and $|T(G)|$ via the local algorithm presented in \cite{lyons2018sharp}. 

% In Ref. \cite{pankdd21}, the authors proposed algorithm for answering $\epsilon$-approximate effective resistance. Among them, $\mathtt{TP}$ is shown to be the state-of-the according to their empirical results. 
\subsubsection{\bf The \texttt{TP} and \texttt{TPC} Algorithms}
In Lemma 4.3 of \cite{pankdd21}, the authors proved $\mathbf{L}^{\dagger}=    \sum_{i=0}^{\infty}{\mathbf{P}^i}\mathbf{D}^{-1}\label{eq:L}$, and further converted ER $r(s,t)$ in Eq. \eqref{eq:er} into an equivalent form as follows:
% Combining Eq. \eqref{eq:L} and  yields the following form of effective resistance $r(s,t)$:
% \begin{small}
\begin{equation}\label{eq:er-new}
r(s,t)=\sum_{i=0}^{\infty}\left({\frac{p_i(s,s)}{d(s)}+\frac{p_i(t,t)}{d(t)}- \frac{p_i(s,t)}{d(t)}- \frac{p_i(t,s)}{d(s)}}\right),
\end{equation}
% \end{small}
where $p_i(s,t)$ denotes the probability of a length-$i$ random walk from $s$ visiting $t$. Notice that $r(s,t)$ in Eq. \eqref{eq:er-new} involves summing up an infinite series of random walk probabilities, which is infeasible in practice. To this end, \citet{pankdd21} propose to approximate its truncated version $r_{\ell}(s,t)$ as defined in Eq. \eqref{eq:r-ell}:%\footnote{We rewrite it from the original formula by applying the fact $\textstyle \tfrac{p_i(s,t)}{d(t)}=\frac{p_i(t,s)}{d(s)}$ \cite{pons2005computing}.}
% \begin{small}
\begin{equation}\label{eq:r-ell}
 r_{\ell}(s,t)=\sum_{i=0}^{\ell}\left({\frac{p_i(s,s)}{d(s)}+\frac{p_i(t,t)}{d(t)}-\frac{  p_i(s,t)}{d(t)}}-\frac{ p_i(t,s)}{d(s)}\right),
\end{equation}
% \end{small}
where $\ell$ satisfies
\begin{equation}\label{eq:old-ell}
\ell= \left\lceil {\ln\left(\tfrac{4}{\epsilon-\epsilon\lambda}\right)}\big/{\ln\left(\tfrac{1}{\lambda}\right)}-1\right\rceil, \text{ and } \lambda=\max\{|\lambda_2|,|\lambda_n|\}.
\end{equation}
By setting the maximum random walk length $\ell$ as in Eq. \eqref{eq:old-ell}, we can ensure $|r(s,t)-r_{\ell}(s,t)|\le \tfrac{\epsilon}{2}$.
% \begin{equation}\label{eq:ell-err}
% |r(s,t)-r_{\ell}(s,t)|\le \tfrac{\epsilon}{2}.
% \end{equation}
% The authors further theoretically proved that $r_{\ell}(s,t)$ is an approximation of $r(s,t)$ with at most $\frac{\epsilon}{2}$ absolute error when $\ell$ satisfies
% where $\lambda=\max\{|\lambda_2|,|\lambda_n|\}$. In other words, if $\ell$ is calculated via Eq. \eqref{eq:old-ell}, the following inequality holds

On the basis of the new problem formulation, two randomized algorithms $\mathtt{TP}$ and $\mathtt{TPC}$ are proposed in \cite{pankdd21}, which are the state-of-the-art solutions for $\epsilon$-approximate PER computation. Both of them are to find an approximation $r^{\prime}(s,t)$ of $r_{\ell}(s,t)$, such that $|r_{\ell}(s,t)-r^{\prime}(s,t)|\le \tfrac{\epsilon}{2}$.

\begin{lemma}[Hoeffding’s inequality \cite{hoeffding1963probability}]\label{lem:hoeffding}
Let $Z_1,Z_2,\dotsc,Z_{n_z}$ be independent random variables with $Z_i$ ($\forall{1 \le i\le n_z}$) is strictly bounded by the interval $[a_j, b_j]$. We define the empirical mean of these variables by $Z=\frac{1}{n_z}\sum_{i=1}^{n_z}{Z_i}$. Then, we have
\begin{small}
\begin{equation*}
\mathbb{P}[|Z-\mathbb{E}[Z]|\ge \varepsilon]\le \revise{2}\exp{\left(-\frac{2n_z^2\varepsilon^2}{\sum_{j=1}^{n_z}(b_j-a_j)^2}\right)}.
\end{equation*}
\end{small}
\end{lemma}

Basically, $\mathtt{TP}$ is inspired by the simple and straightforward idea of Monte Carlo approach. More specifically, it simulates random walks of length ranging from $1$ to $\ell$ from nodes $s$ or $t$, records the fraction of them ending at $s$ or $t$, and finally aggregates all these values into $r^{\prime}(s,t)$ by the formula in Eq. \eqref{eq:r-ell}. 
% At its core lies a truncation bound for random walk lengths, stated as follows.
% Hence, the problem of $\epsilon$-approximate PER computation is transformed to find an approximation $r^{\prime}(s,t)$ of $r_{\ell}(s,t)$, such that $|r_{\ell}(s,t)-r^{\prime}(s,t)|\le \frac{\epsilon}{2}$. 
Applying \revise{the} Chernoff-Hoeffding inequality derives that a total of $\frac{40\ell^2\ln{(8\ell/\delta)}}{\epsilon^2}$ random walks are required for each length $i\in [1,\ell]$ so as to attain the desired $\epsilon$-approximation with a success probability at least $1-\delta$. $\mathtt{TP}$ suffers from severe efficiency issues, even on small graphs, by reason of a huge number of random walk samples as well as the large random walk length $\ell$ (up to thousands when $\epsilon$ is small).

Unlike $\mathtt{TP}$, which directly estimates $p_i(s,t)$ via length-$i$ ($i\in [1,\ell]$) random walks, $\mathtt{TPC}$ achieves a better theoretical bound by regarding $p_i(s,t)$ as a collision probability of two random walks of length $i/2$, $\forall{1\le i\le \ell}$. $\mathtt{TPC}$ samples \revise{$\textstyle 40000 \times \left({\ell\sqrt{\ell \beta_i}}/{\epsilon}+{\ell^3\beta^{3/2}_i}/{\epsilon^2}\right)$} length-$i/2$ random walks from nodes $s$ and $t$, respectively, where the key parameter $\beta_i$ is required to satisfy the inequality $\textstyle \beta_i\ge \max\left\{\sum_{v\in V}{\tfrac{p_i(s,v)^2}{d(v)}}, \sum_{v\in V}{\tfrac{p_i(t,v)^2}{d(v)}}\right\}$. This requirement makes $\mathtt{TPC}$ impractical since the parameter $\beta_i$ is unknown and hard to estimate.

%% file: tex/proposed-new.tex
\section{An Adaptive Monte Carlo Approach}\label{sec:amc}
In this section, we present \fwd, an {\em \underline{A}daptive \underline{M}onte \underline{C}arlo} algorithm, for $\epsilon$-approximate PER processing.
\fwd is similar in spirit to $\mathtt{TP}$ in that it first derives a maximum random walk length $\ell$ and then samples truncated random walks to estimate $r(s,t)$. The crucial differences are two-fold. First, contrary to the generic $\ell$ in Eq. \eqref{eq:old-ell}, we offer a refined and individual $\ell$ for each node pair $(s,t)\in V\times V$ through a rigorous theoretical analysis (Section \ref{sec:ell-fwd}). Second, with the help of empirical Bernstein inequality, we tweak the random walk sampling such that it runs in an adaptive fashion; that is to say, it can terminate early under certain conditions without sacrificing any accuracy (Section \ref{sec:alg-fwd}). In addition, a theoretical analysis regarding correctness and complexity is given in Section \ref{sec:als-fwd} to indicate its superiority over $\mathtt{TP}$.

\subsection{Refining Maximum Length $\ell$}\label{sec:ell-fwd}

% \begin{lemma}[\cite{pons2005computing}]\label{lem:dsdt}
% Given an undirected graph $G$, for any two nodes $u,v\in G$, $p_i(u,v)\cdot d(u)=p_i(v,u)\cdot d(v)$.
% \end{lemma}

% By Lemma \ref{lem:dsdt}, Eq. \eqref{eq:r-ell} is equivalent to the following equation:
% \begin{small}
% \begin{equation}\label{eq:r-ell-new}
%  r_{\ell}(s,t)=\sum_{i=0}^{\ell}{\tfrac{p_i(s,s)}{d(s)}+\tfrac{p_i(t,t)}{d(t)}-\tfrac{ 2p_i(s,t)}{d(t)}}.
% \end{equation}
% \end{small}

% Recall that ER $r(s,t)=\frac{c(s,t)}{2m}$, where commute time $c(s,t)$ between $s$ and $t$ is defined as the expected number of steps for a random walk starting at $s$ to return $s$ after at least one visit to $t$. Intuitively, for nodes $s$ and $t$ with large degrees, 
The {\em mixing time} \cite{das2013distributed} $\xi_s$ with respect to a source node $s$ is the number of steps taken by a random walk from $s$ to converge to the stationary distribution $\vec{\boldsymbol{\pi}}$ of the graph. Note that mixing times vary from node to node and are directly associated with their neighborhoods \cite{das2013distributed}. In particular, when we set $\ell=\max\{\xi_s,\xi_t\}$, by Eq. \eqref{eq:r-ell} we can obtain $r(s,t)=r_{\ell}(s,t)$. Intuitively, given the same $\ell$, $r_{\ell}(s,t)$ of a node pair $(s,t)$ with low mixing times is more likely to be close to $r(s,t)$ compared to the case where the node pairs have high mixing times. To put it from another way, the maximum random walk length $\ell$ required to ensure $|r(s,t)-r_{\ell}(s,t)|\le \tfrac{\epsilon}{2}$ is small if the mixing times pertaining to nodes $s,t$ are low. However, in $\mathtt{TP}$ and $\mathtt{TPC}$, they provide a large and generic $\ell$ (i.e., Eq. \eqref{eq:old-ell}) for all node pairs, regardless of the structural property surrounding each node, adversely impacting their performance. 
To remedy this problem, we offer an improved and personalized $\ell$ for different node pairs in Theorem \ref{lem:ell} by  capitalizing on the properties of eigenvalues and eigenvectors of $\PM$.
% in Lemma \ref{lem:pist}.

\begin{theorem}\label{lem:ell}
Given a graph $G$, $|r(s,t)-r_{\ell}(s,t)|\le \frac{\epsilon}{2}$ holds for any two nodes $s,t$ in $G$ when $\ell$ satisfies
\begin{small}
\begin{equation}\label{eq:ell}
\textstyle
\ell= \left\lceil\log{\left(\tfrac{\tfrac{2}{{d(s)}}+\tfrac{2}{{d(t)}}}{\epsilon(1-\lambda)}\right)}\bigg/\log{\left(\tfrac{1}{\lambda}\right)}-1\right\rceil,\ \textrm{where } \lambda=\max\{|\lambda_2|,|\lambda_n|\}.
\end{equation}
\end{small}
% \tj{$\lambda_2$ can be improved to $\lambda$ in \eqref{eq:ell}.}
\end{theorem}
\revise{
Compared to Peng et al.'s $\ell$ in Eq. \eqref{eq:old-ell}, our refined $\ell$ in Eq. \eqref{eq:ell} incorporates node degrees $d(s),d(t)$ in it, which is favorable particularly for graphs with high average degrees. The empirical results in Section \ref{sec:exp-param} reveal the great superiority of our refined $\ell$ over Peng et al.'s $\ell$ in the practical efficiency of PER computation.
}
% An additional superiority of our $\ell$ over Peng et al.'s $\ell$ is that we only need to calculate the second largest eigenvalue $\lambda_2$ of $\PM$, whereas Eq. \eqref{eq:old-ell} requires computing both $\lambda_2$ and $\lambda_n$ to find $\lambda$ by $\lambda=\max\{|\lambda_2|,|\lambda_n|\}$, making Eq. \eqref{eq:old-ell} further inferior to Eq. \eqref{eq:ell}. 
\revise{Given that $\PM$ is a sparse matrix with $m$ non-zero entries, iterative methods \cite{demmel1997applied} such as the Implicitly Restarted Arnoldi Method \cite{lehoucq1998arpack} for computing $\lambda$ merely require performing sparse matrix-vector multiplications between $\PM$ and two length-$n$ vectors, resulting in linear time complexity of $O(m)$. This suggests that 
% both the computation of $\lambda_2$ and $\lambda_n$ can be efficiently done using .
% ,combes2019computationally,demmel1997applied}; 
% As such, 
the computation of $\lambda$ can be efficiently done through a preprocessing step. For example, on the Orkut graph with $117$ million edges, computing $\lambda$ consumes less than five minutes.
Note that this preprocessing step only needs to be conducted once for a graph, and the obtained $\lambda$ can be reused for computing the maximum length $\ell$ by Eq. \eqref{eq:ell} for any node pairs. 
}

\begin{figure}[!t]
\centering
\includegraphics[width=0.66\columnwidth]{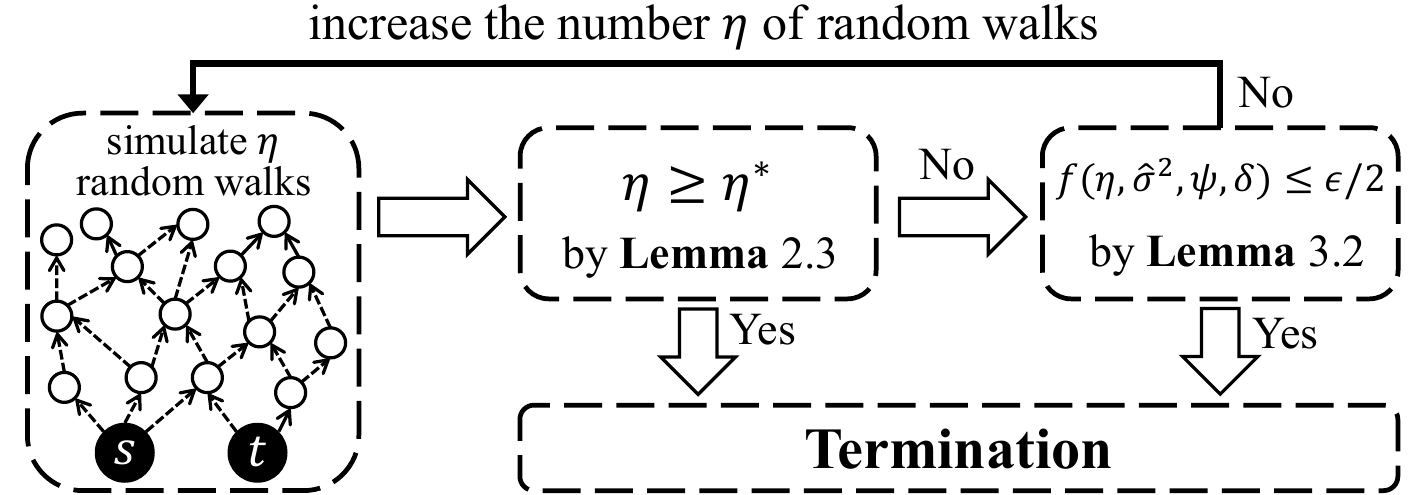}
\vspace{-1mm}
\caption{Overview of \fwd}\label{fig:amc}
%\vspace{-2mm}
\end{figure}

% \subsection{Bidirectional Estimation}

% https://www.normalesup.org/~rpeyre/pro/research/carne.pdf
% https://math.dartmouth.edu/~pw/math100w13/komarov.pdf
% \begin{theorem}[Varopoulos-Carne bound \cite{carne1985transmutation,peyre2008probabilistic}]
% Let $\PM$ be a reversible markov chain, let $i\ge 1$; let $u\neq v\in V$; then
% \begin{equation*}
% p_i(u,v)\le \sqrt{e}\sqrt{\frac{d(v)}{d(u)}}\exp{\left(-\frac{dist(s,t)^2}{2t}\right)}.
% \end{equation*}
% \end{theorem}
% \input{tex/proposed-fix}

\subsection{\bf Sampling Random Walks Adaptively}\label{sec:alg-fwd}
\revise{Before delving into the algorithmic details, we first elaborate on the high-level idea of \fwd.}

% incremental computation of standard deviation/ variance
% https://math.stackexchange.com/questions/102978/incremental-computation-of-standard-deviation
\subsubsection{High-level Idea}
% In the first place, we elaborate on the high-level idea of \fwd.
Recall from Section \ref{sec:mcpt} that $\mathtt{TP}$ and $\mathtt{TPC}$ utilize the Hoeffding’s inequality (Lemma \ref{lem:hoeffding}) to derive the maximum number of random walks needed in advance before conducting them in a single batch. Practically, most of the random walks are excessive due to the looseness of the Hoeffding bound, wherein the actual variance of random variables is much lower than the one it assumes. Ideally, an efficient method should run with the minimum amount of random walks derived according to the actual variance, which is unknown beforehand. 

Motivated by this, \fwd's idea is to alter the sampling phase such that it runs adaptively. That is, instead of performing all random walks in a single batch, we conduct multiple batches of independent Monte Carlo random walks incrementally to estimate ER $r(s,t)$, and determine the termination according to the empirical error computed from the observed samples. Note that the empirical error can be calculated via the empirical Bernstein inequality in Lemma \ref{lem:bernstein}.
% at the end of each round we compute its . Once the empirical error is , we terminate the procedure immediately and return. Note that the empirical error can be determinted via empirical Bernstein inequality (Lemma \ref{lem:bernstein}). 
Fig. \ref{fig:amc} overviews this adaptive sampling scheme adopted in \fwd. More specifically, \fwd starts with a small-size batch of random walks ($\eta$ samples in total) to compute an approximate ER $r^{\prime}(s,t)$. Let $\eta^\ast$ be the maximum number of random walks required by Lemma \ref{lem:hoeffding}. If either $\eta$ exceeds $\eta^\ast$ or the empirical error $f(\eta,\hat{\sigma}^2,\psi,\delta)$ computed by Lemma \ref{lem:bernstein} is not greater than the desired error threshold $\epsilon/2$, \fwd terminates and returns $r^{\prime}(s,t)$ as the answer. Otherwise, we discard all results and samples in this batch and repeat the above procedure with an increased $\eta$. This adaptive scheme 
% makes a number of attempts in simulations to find an early termination point and it
ensures its cost never exceeds that of simple Monte Carlo by imposing a restriction on $\eta$, i.e., $\eta \le \eta^\ast$. 
% Towards this end, we leverage the empirical Bernstein inequality (Lemma \ref{lem:bernstein}), which allows us for using the empirical variance to determine the termination of sampling. Fig. \ref{fig:amc} an overview of \fwd.
On top of that, it largely curtails the amount of random walks in \fwd compared with prior Monte Carlo approaches (e.g., $\mathtt{TP}$) as the experiments in Section \ref{sec:exp} demonstrate, owing to that the empirical Bernstein inequality quickly gets to be much tighter than the Hoeffding's inequality when the variances of the random variables are small \cite{mnih2008empirical}.

\begin{lemma}[empirical Bernstein
inequality \cite{audibert2007tuning}]\label{lem:bernstein}
Let $Z_1,Z_2,\dotsc,Z_{n_z}$ be real-valued i.i.d. random variables, such that $0\le Z_i \le \psi$. We denote by $Z=\frac{1}{n_z}\sum_{i=1}^{n_z}{Z_i}$ the empirical mean of these variables and $\hat{\sigma}^2=\frac{1}{n_z}\sum_{i=1}^{n_z}{(Z_i-Z)^2}$ their empirical variance. Then, we have
\begin{equation*}
\mathbb{P}\left[|Z-\mathbb{E}[Z]|\ge f(n_z,\hat{\sigma}^2,\psi,\delta)\right]\le \delta,
\end{equation*}
where
% \begin{small}
\begin{equation}\label{eq:fx}
f(n_z,\hat{\sigma}^2,\psi,\delta)=\sqrt{\frac{2{\hat{\sigma}^2}\log{(3/\delta)}}{n_z}}+\frac{3\psi\log{(3/\delta)}}{n_z}.
\end{equation}
% \end{small}
\end{lemma}
Building on the idea explained above, we delineate the algorithmic details of \fwd in what follows.

\begin{algorithm}[tb]
\caption{\fwd}\label{alg:fwd}
\KwIn{Graph $G$, two nodes $s,t$, two vectors $\vec{\mathbf{s}}$, $\vec{\mathbf{t}}$, error threshold $\epsilon$, maximum random walk length $\ell_f$, the  number $\tau$ of batches, and failure probability $\delta$}
\KwOut{$r_f(s,t)$}
Compute $\eta^\ast$ by Eq. \eqref{eq:etas-star} and $\psi$ by Eq.~\eqref{eq:Delta-st}\;\label{alg:fwd-eta-ast}
$\eta\gets \lceil\eta^\ast/2^{\tau-1}\rceil$\;\label{alg:fwd-eta}
\For{$i\gets 1$ \KwTo $\tau$ \label{alg:fwd-iter}}{
$Z\gets 0;\ \hat{\sigma}^2\gets 0$\; 
\For{$k \gets 1$ \KwTo $\eta$}{
Perform two random walks ${S}_{k}$ and ${T}_{k}$ with length-$\ell_f$ from $s$ and $t$, respectively\;
$Z_{k} \gets \sum_{u\in S_k}\left(\frac{{\vec{\mathbf{s}}(u)}}{d(s)}-\frac{{\vec{\mathbf{t}}(u)}}{d(t)}\right)+\sum_{u\in T_k}\left(\frac{{\vec{\mathbf{t}}(u)}}{d(t)}-\frac{{\vec{\mathbf{s}}(u)}}{d(s)}\right)$\;
Update $Z\gets Z+Z_k$\;
Update $\hat{\sigma}^2\gets\hat{\sigma}^2+Z_k^2$\;\label{alg:fwd-variance-sum}
}
Compute $Z\gets \frac{Z}{\eta}$\;
Compute $\hat{\sigma}^2\gets\frac{\hat{\sigma}^2}{\eta}-Z^2$\;\label{alg:fwd-variance}
\lIf{$f\left(\eta,\hat{\sigma}^2,\psi,{\delta}/{\tau}\right) \le \frac{\epsilon}{2}$}{\textbf{break}\label{alg:fwd-stop}}
$\eta\gets 2\eta$\;\label{alg:fwd-iter-end}
}
\Return{$r_f(s,t)\gets Z$\;\label{alg:fwd-return}}
\end{algorithm}

\subsubsection{Algorithm}
Algorithm \ref{alg:fwd} illustrates the pseudo-code of \fwd. To begin with, \fwd takes as input a graph $G$, two nodes $s,t$, maximum random walk length $\ell_f$, the error threshold $\epsilon$, failure probability $\delta$, the number of batches $\tau$, and two length-$n$ non-negative vectors $\vec{\mathbf{s}}$ and $\vec{\mathbf{t}}$. For the case of $\epsilon$-approximate PER computation, $\vec{\mathbf{s}}$ (resp.\ $\vec{\mathbf{t}}$) is fixed to be $\vec{\mathbf{e}}_s$ (resp.\ $\vec{\mathbf{e}}_t$), as explained later in Section \ref{sec:als-fwd}. Initially, \fwd computes the maximum number $\eta^\ast$ of random walks required from both $s$ and $t$ and a parameter $\psi$ at Line~\ref{alg:fwd-eta-ast} by the following equations:
% \begin{small}
\begin{equation}
\eta^\ast= \frac{2\psi^2\log{({2\tau}/{\delta})}}{\epsilon^2},\label{eq:etas-star}
\end{equation}
% \end{small}
where 
\begin{equation}\label{eq:Delta-st}
	\psi=2 \left\lceil \tfrac{\ell_f}{2} \right\rceil\left(\tfrac{\max\nolimits_1(\vec{\mathbf{s}})}{d(s)}+\tfrac{\max\nolimits_1(\vec{\mathbf{t}})}{d(t)}\right)+2\left\lfloor \tfrac{\ell_f}{2} \right\rfloor\left(\tfrac{\max\nolimits_2(\vec{\mathbf{s}})}{d(s)}+\tfrac{\max\nolimits_2(\vec{\mathbf{t}})}{d(t)}\right).
\end{equation}
\revise{Intuitively, the parameter $\psi/2$ can be regarded as an upper bound for the absolute value of the random variable sampled by each random walk in \fwd (detailed in Section \ref{sec:als-fwd}).} \fwd then initializes the number $\eta$ of random walks to be performed from both $s$ and $t$ as $\lceil\eta^\ast/2^{\tau-1}\rceil$ (Line~\ref{alg:fwd-eta}).
Afterwards, \fwd starts $\tau$ batches of sampling (Lines~\ref{alg:fwd-iter}--\ref{alg:fwd-iter-end}), in each $i$-th of which it calculates the empirical mean $Z$ and empirical variance $\hat{\sigma}^2$, respectively. Observing that the empirical variance $\hat{\sigma}^2$ can be rewritten as
\begin{equation*}
	\hat{\sigma}^2=\frac{1}{\eta}\sum_{k=1}^{\eta}(Z_k-Z)^2=\frac{\sum_{k=1}^{\eta}Z_k^2}{\eta}-Z^2,
\end{equation*}
we can compute $\hat{\sigma}^2$ efficiently (Lines~\ref{alg:fwd-variance-sum} and \ref{alg:fwd-variance}). We repeat the process until they satisfy any of the two termination conditions: (i) the maximum number of iterations is reached (i.e., the maximum number of required random walks are conducted) (Lines 3) and (ii) the empirical error is less than the desired error threshold $\tfrac{\epsilon}{2}$ (Line~\ref{alg:fwd-stop}). Finally, $r_f(s,t)=Z$ will be returned to derive an accurate estimate of $r(s,t)$ (Line~\ref{alg:fwd-return}). 

\subsection{Analysis}\label{sec:als-fwd}
\revise{In what follows, we theoretically analyze the correctness and complexity of \fwd.}

\subsubsection{Correctness}
In the following, we will show that $r_f(s,t)+\mathbbm{1}_{s\neq t}\cdot(\tfrac{1}{d(s)}+\tfrac{1}{d(t)})$ is an $\epsilon$-approximation of $r(s,t)$ (Eq. \eqref{eq:epsilon}) w.h.p. when the input $\ell_f=\ell$ (Eq. \eqref{eq:ell}), $\vec{\mathbf{s}}=\vec{\mathbf{e}}_s$ and $\vec{\mathbf{t}}=\vec{\mathbf{e}}_t$.

\begin{lemma}\label{lem:sum-bound}
Given a length-$n$ non-negative vector $\vec{\mathbf{x}}$ and any length-$\ell_f$ random walk $W$ originating from $u$ on graph $G$, we have
% \begin{small}
\begin{equation}
\textstyle
\ell_f \cdot \min(\vec{\mathbf{x}})\le \sum_{w\in W}{\vec{\mathbf{x}}(w)}\le  \left\lceil \tfrac{\ell_f}{2} \right\rceil \cdot \max_1(\vec{\mathbf{x}}) + \left\lfloor \tfrac{\ell_f}{2} \right\rfloor \cdot \max_2(\vec{\mathbf{x}}).
\end{equation}
% \end{small}
\end{lemma}
For random walks $S_k$ and $T_k$ with length-$\ell_f$ (Line 6), consider the random variable $Z_k$ \revise{arising from the random walk in \fwd (Line 7) such that}
\begin{equation}
	Z_{k} = \sum_{u\in S_k}\left(\frac{{\vec{\mathbf{s}}(u)}}{d(s)}-\frac{{\vec{\mathbf{t}}(u)}}{d(t)}\right)+\sum_{u\in T_k}\left(\frac{{\vec{\mathbf{t}}(u)}}{d(t)}-\frac{{\vec{\mathbf{s}}(u)}}{d(s)}\right).
\end{equation}
By Lemma~\ref{lem:sum-bound}, given $\min(\vec{\mathbf{s}})\geq 0$ and $\min(\vec{\mathbf{t}})\geq 0$, we have
\begin{equation*}
	|Z_{k}|\leq \left\lceil \tfrac{\ell_f}{2} \right\rceil\left(\tfrac{\max\nolimits_1(\vec{\mathbf{s}})}{d(s)}+\tfrac{\max\nolimits_1(\vec{\mathbf{t}})}{d(t)}\right)+\left\lfloor \tfrac{\ell_f}{2} \right\rfloor\left(\tfrac{\max\nolimits_2(\vec{\mathbf{s}})}{d(s)}+\tfrac{\max\nolimits_2(\vec{\mathbf{t}})}{d(t)}\right)=\tfrac{\psi}{2},
\end{equation*}
where $\psi$ is given in Eq. \eqref{eq:Delta-st}. Further, let
\begin{equation}\label{eq:qst}
	q(s,t):=\sum_{i=1}^{\ell_f}{\sum_{v\in V}{\left(\big(p_{i}(s,v)-p_{i}(t,v)\big)\cdot \left(\frac{\vec{\mathbf{s}}(v)}{d(s)}-\frac{\vec{\mathbf{t}}(v)}{d(t)}\right)\right)}},
\end{equation}
\revise{where $p_i(s,v)$ (resp. $p_i(t,v)$) signifies the probability that a length-$i$ random walk starting from node $s$ (resp. $t$) would end at node $v$. Then, by the definition of $Z_k$ (Line 7), it is trivial to derive the expectation of $Z_k$ as follows:}
\begin{equation}\label{eq:x-exp}
	\mathbb{E}[Z_k]= q(s,t),
\end{equation}
which indicates that $r_f(s,t)=Z$ (Line 16) is an unbiased estimator of $q(s,t)$.

Note that Algorithm \ref{alg:fwd} terminates under two conditions: (i) $\epsilon_f \le \frac{\epsilon}{2}$, where $\epsilon_f=f\left(\eta,\hat{\sigma}^2,\psi,{\delta}/{\tau}\right)$ (Line 13), or (ii) $\eta \ge \eta^\ast$ (Line 3). For the first case that $\epsilon_f \le \frac{\epsilon}{2}$, we can bound its failure probability by applying the empirical Bernstein inequality (Lemma \ref{lem:bernstein}). As for another termination case that $\eta \ge \eta^\ast$, its failure probability can be bounded by using Hoeffding's inequality (Lemma \ref{lem:hoeffding}). Further, using union bound over these two cases, we can prove $r_f(s,t)$ returned by Algorithm~\ref{alg:fwd} is an $\frac{\epsilon}{2}$-approximate of $q(s,t)$ with high probability. In addition, note that $r_{\ell}(s,t)=q(s,t)+\mathbbm{1}_{s\neq t}\cdot(\tfrac{1}{d(s)}+\tfrac{1}{d(t)})$ when the input of Algorithm \ref{alg:fwd} satisfies $\ell_f=\ell$ as in Eq.~\eqref{eq:ell}, $\vec{\mathbf{s}}=\vec{\mathbf{e}}_s$ and $\vec{\mathbf{t}}=\vec{\mathbf{e}}_t$. Therefore, in Theorem \ref{lem:fwd}, we have $\left\lvert r_f(s,t)+\mathbbm{1}_{s\neq t}\cdot \left(\frac{1}{d(s)}+\frac{1}{d(t)}\right)-r_{\ell}(s,t)\right\rvert \le \frac{\epsilon}{2}$ with high probability; in other words, $r_f(s,t)+\mathbbm{1}_{s\neq t}\cdot \left(\tfrac{1}{d(s)}+\tfrac{1}{d(t)}\right)$ is an $\epsilon$-approximate ER of node pair $(s,t)$.
\begin{theorem}\label{lem:fwd}
Algorithm \ref{alg:fwd} ensures $\left|r_f(s,t)- q(s,t)\right| \le \frac{\epsilon}{2}$ with a probability of at least $1-\delta$. Moreover, setting $\ell_f=\ell$ as in Eq.~\eqref{eq:ell}, $\vec{\mathbf{s}}=\vec{\mathbf{e}}_s$ and $\vec{\mathbf{t}}=\vec{\mathbf{e}}_t$, Algorithm \ref{alg:fwd} returns $r_f(s,t)$ satisfying $\left\lvert r_f(s,t)+\mathbbm{1}_{s\neq t}\cdot(\tfrac{1}{d(s)}+\tfrac{1}{d(t)})-r(s,t)\right\rvert\le {\epsilon}$ with a probability of at least $1-\delta$.
\end{theorem}

% \subsubsection
\subsubsection{Complexity}
% \begin{theorem}\label{lem:fwd-time}
% The expected time complexity of \fwd is bounded by $O(h(\ell_f)\cdot \ell_f^2)$, where
% \begin{small}
% \begin{align}
%  h(\ell_f)=& \frac{16{\ell_f}^2\log{(\frac{3\ell_f}{\delta})}}{\epsilon^2}\cdot \max\Bigg\{ \frac{2\psi^2_t\cdot \left(\sqrt{2}+\sqrt{2+\frac{3\epsilon d(t)}{2\ell_f \psi_t}}\right)^2}{d(t)^2}, \notag\\
%  &\ \frac{\psi^2_s \cdot \left(\sqrt{2}+\sqrt{2+\frac{3\epsilon d(s)}{2\ell_f \psi_s}}\right)}{d(s)^2}+\frac{\psi^2_t\cdot \left(\sqrt{2}+\sqrt{2+\frac{3\epsilon d(t)}{2\ell_f \psi_t}}\right)^2}{d(t)^2} \Bigg\} \label{eq:time-fwd}.
% \end{align} 
% \end{small}
% \end{theorem}
According to Lines~\ref{alg:fwd-eta} and \ref{alg:fwd-iter-end} in Algorithm \ref{alg:fwd}, the total number of random walks conducted from nodes $s,t$ in the course of $\tau$ iterations is bounded by
\begin{equation*}
h(\ell_f)=\sum_{i=1}^{\tau}{2^{i-1}\eta}=(2^\tau-1)\eta< 2\eta^\ast. %+2^\tau.
\end{equation*}
The running time of \fwd is bounded by $O(h(\ell_f) \cdot \ell_f)$ as the length of each random walk in \fwd is $\ell_f$. When $\ell_f=\ell$ (Eq. \eqref{eq:ell}), $\vec{\mathbf{s}}=\vec{\mathbf{e}}_s$,
$\vec{\mathbf{t}}=\vec{\mathbf{e}}_t$, Eq. \eqref{eq:Delta-st} implies $\psi=2\left\lceil \tfrac{\ell}{2} \right\rceil\left(\frac{1}{d(s)}+\frac{1}{d(t)}\right)$. Recall that $\eta^{\ast}$ is defined in Eq. \eqref{eq:etas-star}.
Therefore, the total time complexity of employing \fwd for answering $\epsilon$-approximate PER query of node pair $(s,t)$ is
\begin{small}
\begin{equation}\label{eq:time-amc}
%  O \left( \frac{\ell^4 \log{(\frac{3\ell}{\delta})}}{\epsilon^2}\cdot \left(\frac{1}{d(s)^2}+\frac{1}{d(t)^2}\right) \right).
 O \left( \Big(\tfrac{1}{\epsilon d(s)}+\tfrac{1}{\epsilon d(t)}\Big)^2\log^3{\Big(\tfrac{1}{\epsilon d(s)}+\tfrac{1}{\epsilon d(t)}\Big)}\right)=O \left(\tfrac{1}{\epsilon^2 d^2} \log^3{\left(\tfrac{1}{\epsilon d}\right)}\right),
\end{equation}
\end{small}
where $d=\min\{d(s),d(t)\}$ and $\delta$ is regarded as a constant.

\revise{
\header
{\bf Remark.}
The total number of random walks conducted in \fwd for answering an $\epsilon$-approximate PER query is roughly $\tfrac{2\ell^2\log{(2\tau/\delta)}}{\epsilon^2}\cdot \left(\tfrac{1}{d(s)}+\tfrac{1}{d(t)}\right)^2$. By contrast, the state of the art PER solution $\mathtt{TP}$ requires $\tfrac{40\ell^3\ln{(8\ell/\delta)}}{\epsilon^2}$ random walks \cite{pankdd21}, which is larger than \fwd's amount by at least a significant factor of $\textstyle 20\ell/{\big(\tfrac{1}{d(s)}+\frac{1}{d(t)}\big)^2}$.
}

% Furthermore, in scale-free graphs, it can be verified that the average node degree $m/n = O(\log{n})$ \cite{barabasi1999emergence,faloutsos2011power}. Suppose the nodes $s$ and $t$ are picked uniformly at random. According to Eq. \eqref{eq:ell}, the amortized time complexity of Algorithm \ref{alg:fwd} can be simplified as $ O\left(\frac{\textrm{poly} \log{(\frac{1}{\epsilon})}}{\epsilon^2\log^2{n}}\right)$.

\begin{table}[!t]
\hspace{0mm}
\begin{minipage}{0.32\linewidth}
% \begin{figure}[!t]
% \centering
\begin{tikzpicture}[
  mycircle/.style={
     circle,
     draw=black,
     fill=gray,
     fill opacity = 0.3,
     text opacity=1,
     inner sep=0pt,
     minimum size=12pt,
     font=\small},
  myarrow/.style={-,>=stealth},
  node distance=0.2cm and 0.4cm
  ]
  \node[mycircle] (c1) {$t$};
  \node[mycircle,above right=of c1] (c4) {$v_1$};
%   \node[mycircle,right=of c1] (c7) {$v_1$};
  \node[mycircle,below right=of c1] (c2) {$v_2$};
  \node[mycircle,above=of c1] (c9) {$v_3$};
  \node[mycircle,below=of c1] (c10) {$v_4$};
  \node[mycircle,left=of c1] (c11) {$v_5$};
  \node[mycircle,above left=of c1] (c12) {$v_6$};
  \node[mycircle,below left=of c1] (c13) {$v_7$};
  \node[mycircle,right=of c2] (c3) {$v_8$};
  \node[mycircle,right=of c4] (c5) {$v_9$};
  \node[mycircle,right=of c1] (c6) {$s$};
%   \node[mycircle,left=of c6] (c8) {$v_11$};

\foreach \i/\j/\txt/\p in {% start node/end node/text/position
%   c1/c7//right,
  c1/c9//above,
  c1/c10//below,
  c1/c11//left,
  c1/c12//left,
  c1/c13//left,
%   c6/c8//right,
%   c7/c8//right,
%   c7/c5//right,
%   c7/c3//right,
  c1/c2//below,
  c1/c4//above,
  c2/c3//below,
  c3/c6//below,
  c4/c5//above,
  c5/c6//above}
%   c5/c2/4/below,
%   c3/c5/7/below,
  %c2.70/c4.290/1/below}
  \draw [myarrow] (\i) -- node[sloped,font=\small,\p] {\txt} (\j);

 % draw this outside loop to get proper orientation of 10
%  \draw [myarrow] (c4.250) -- node[sloped,font=\small,above,rotate=180] {10} (c2.110);
\end{tikzpicture}
% \vspace{-2mm}
% \caption{An example graph}\label{fig:toy}
% \vspace{-2mm}
% \end{figure}
\end{minipage}
\hspace{-10mm}
\begin{minipage}{0.62\linewidth}
\centering
\begin{small}
\renewcommand{\arraystretch}{1.0}
% \resizebox{0.95\columnwidth}{!}{%
\begin{tabular}{|c|c|c|c|c|c|c|c|c|}\hline
$\ell_f$ &  1 & 2 & 3 & 4 & 5 & 6 & 7 & 8\\  \hline   
% \#path($s$) & 2 & 6 & 12 & 38 & 80 & 264 \\  \hline
% ${\eta_s}$  & 80 & 320 & 720 & 1280 & 2000 & 2879 \\  \hline
% \#path($t$)  & 7 & 16 & 69 & 140 & 537 & 1076 \\  \hline
% ${\eta_t}$  & 8 & 27 & 59 & 105 & 164 & 235 \\  \hline
\#path($s$) & 2 & 4 & 8 & 26 & 42 & 184 & 268 & 1346 \\  
\#path($t$) & 7 & 9 & 53 & 71 & 397 & 539 & 2963 & 4041 \\  \hline
\makecell{\#path($s$)\\ + \#path($t$)} & 9 & 13 & 61 & 97 & 439 & 723 & \bf 3231 & \bf 5387 \\  \hline
${\eta}^\ast$  & \bf 31 & \bf 122 & \bf 275 & \bf 488 & \bf 762 & \bf 1097 & 1493 & 1949 \\  \hline
\end{tabular}
% }
\end{small}
\end{minipage}
% \hfill
\vspace{-1mm}
\captionof{figure}{A running example.}
\label{fig:toy}
%\vspace{-2mm}
\end{table}

\section{The Greedy Approach}\label{sec:main}
% \subsection{Limitations of AMC}
Although \fwd significantly improves over existing approaches in terms of running time, we observe in our experiments (Section \ref{sec:exp}) that the performance of \fwd is still intolerable especially for large graphs, as an aftermath of the enormous number of excessive random walks needed in it.
% owing to the twofold reasons as follows. 
% First, \fwd requires a huge number of random access to nodes in $G$, leading to intensive memory-access patterns and in turn a substantial overhead. Apart from this, another 
% major issue that severely thwarts the performance of \fwd is attributed to the sheer number of excessive random walks. 
To explain, we illustrate an example in Fig. \ref{fig:toy}, which contains a toy graph with eleven nodes including $v_1$--$v_9$ and $s,t$.
% To identify the inherent drawbacks of \fwd, we first introduce a deterministic method \bwd for PER computation and make a comparison between it and \fwd.

\header
{\bf A Running Example.}
The r.h.s.\ table in Fig.\ \ref{fig:toy} lists the numbers $\#path(s),\#path(t)$ of distinct paths from node $s$ and $t$ within varied length $\ell_f$ (ranging from 1 to 8). Note that these numbers can be obtained by performing deterministic graph traversals from nodes $s$ and $t$, and used to compute $r_{f}(s,t)$. The table also reports $\eta^\ast$, which is the number of random walks required by \fwd from nodes $s$ and $t$ to estimate $r_f(s,t)$ when varying length $\ell_f$ from $1$ to $8$ with additive error $\epsilon=0.5$ and failure probability $\delta=0.1$. As we can see from the table, when length $\ell_f$ is increased from $1$ to $8$, $\eta^\ast$ constantly outnumbers $\#path(s)+\#path(t)$ by up to $10$-fold, whereas $\#path(t)+\#path(t)$ exceeds $\eta^\ast$ notably when $\ell_f\ge 7$.
The former is by virtue of the insignificant amount of nodes  in the near vicinity of nodes $s$ and $t$, and thus a simple graph traversal can efficiently explore all of them. This suggests that \fwd tends to be even less efficient than the deterministic graph traversal, when $\ell_f$ is small and source nodes have scant connections to the rest of the graph.
However, as $\ell_f$ continues to increase, the number $\#path(t)$ of paths from $t$ grows at an astonishing rate (from $7$ to $4041$) as it has $7$ adjacent neighbors, rendering sampling random walks (i.e., \fwd) a preferred choice over the graph traversal. This happens especially on large graphs with high average degrees.

%The former is due to the insignificant amount of reachable nodes from node $s$, which has only 2 adjacent nodes $v_8$ and $v_9$. In such cases, all these nodes can be probed efficiently by the graph traversal and \fwd becomes rather inefficient. In contrast, in the latter case where $t$ has $7$ neighbors, the number $\#path(t)$ of paths grows at an astonishing rate (from $7$ to $4041$) as $\ell_f$ increases, rendering sampling random walks (i.e., \fwd) a preferred choice over the graph traversal. This happens especially on large graphs with high average degrees.

%To recapitulate, \fwd tends to be even less efficient than the deterministic graph traversal, when $\ell_f$ is small and source nodes have scant connections to the rest of the graph.

% especially in the cases where the nodes have less connections to the rest of the graph.
% few connections between nodes, which can be efficiently by graph traversal, where $\ell_f$ is small, graph size is small, and node degrees are small. 

% A simple and straightforward way to 

Based upon the above-said observations, it is natural to investigate how to combine \fwd and deterministic graph traversal together 
% so as to overcome the limitations of both 
and create an even more efficient approach for answering $\epsilon$-approximate PER queries. 
To achieve this goal, this section proposes \algo (\underline{G}reedy \underline{E}stimation of \underline{E}ffective \underline{R}esistance), which significantly advances the practical efficiency of \fwd without degrading its theoretical guarantees, by harnessing the virtues of sparse matrix-vector multiplications and random walks in a judicious manner. The subsequent section elucidates the rationale behind \algo; after that the detailed algorithm and rigorous theoretical analysis will be presented.

% still suffers from severe computational efficiency issues

% \begin{algorithm}[tb]
% \caption{\bwd}
% \label{alg:bwd}
% \begin{flushleft}
% \textbf{Input}: Graph $G$, two nodes $s,t$, and $\ell_b$\\
% \textbf{Output}: $\rho(s,s),\rho(s,t),\rho(t,t),\vec{\mathbf{t}}$ and $\vec{\mathbf{s}}$
% \end{flushleft}
% \begin{algorithmic}[1] %[1] enables line numbers
% % \STATE Compute $\ell$ according to Eq. \eqref{eq:ell};
% \STATE $\vec{\mathbf{t}}\gets \vec{\mathbf{e}}_t; \vec{\mathbf{s}}\gets \vec{\mathbf{e}}_s$;
% \STATE $\rho(s,s)\gets 0; \rho(s,t)\gets 0; \rho(t,t)\gets 0$;
% \FOR{$i\gets 0$ to $\ell_b$}
% \STATE $\rho(s,s)\gets \rho(s,s)+\vec{\mathbf{s}}[s]$;
% \STATE $\rho(s,t)\gets \rho(s,t)+\vec{\mathbf{t}}[t]$;
% \STATE $\rho(t,t)\gets \rho(t,t)+\vec{\mathbf{t}}[t]$;
% \STATE $\vec{\mathbf{t}}\gets \mathbf{P}\vec{\mathbf{t}}; \vec{\mathbf{s}}\gets \mathbf{P}\vec{\mathbf{s}}$;
% \ENDFOR
% \STATE \textbf{return} $\rho(s,s),\rho(s,t),\rho(t,t),\vec{\mathbf{t}},\vec{\mathbf{s}}$
% \end{algorithmic}
% \end{algorithm}

% \subsection{Sparse Matrix-Vector Multiplications}

\subsection{Overview of \texttt{GEER}}\label{sec:algo-over}
% The random walk-based approach presented thus far still involves 
% Intuition and rationale. Observe that a few hops away from s and t. no many nodes.. For instance, the first step random walk only has $d(s)$ and $d(t)$ nodes, meaning that we can calculate the probabilities via matrix-vector iterations, only $d(s)$ and $d(t)$ operations are needed. But if we employ Monte Carlo random walks, XXXX samples are required, which is rather costly.

To begin with, we introduce a secondary algorithm \bwd, which implements the deterministic graph traversal over $G$ based on sparse matrix-vector multiplications. 
% Next, we show how to utilize it to significantly reduce the number of random walks needed in \fwd.
% for enhanced efficiency.
% with a significant reduction in random walks needed. 

\subsubsection{A Deterministic Method \texttt{SMM}}\label{sec:smm}
Recall from Eq. \eqref{eq:r-ell} that the computation of $r_\ell(s,t)$ involves calculating $p_i(s,t)$, $p_i(s,s)$ and $p_i(t,t)$ for $0\le i\le \ell$.
By definition, for any two nodes $u,v\in V$
%\begin{small}
\begin{equation*}
%\textstyle
p_i(v,u)=\PM^{i}(v,u)=(\PM^i\vec{\mathbf{e}}_u)(v)=(\smallunderbrace{\PM\cdot\PM\dotsb\PM}\cdot\vec{\mathbf{e}}_u)(v),
\end{equation*}
%\end{small}
% Note that the above can be implemented by multiplying sparse matrix $\PM$ with a vector with initial value as $\vec{\mathbf{e}}_u$ using $i$ iterations.
\revise{and $p_0(u,v)=\vec{\mathbf{e}}_u(v)$}. Therefore, $r_{\ell}(s,t)$ can be obtained via an iterative process of sparse matrix-vector multiplications as displayed in Algorithm \ref{alg:smm} with initial vectors $\vec{\mathbf{s}}^\ast=\vec{\mathbf{e}}_s$, $\vec{\mathbf{t}}^\ast=\vec{\mathbf{e}}_t$, and 
\begin{equation*}
	r_b(s,t)= \frac{\vec{\mathbf{s}}^\ast(s)}{d(s)}+\frac{\vec{\mathbf{t}}^\ast(t)}{d(t)}-\frac{\vec{\mathbf{s}}^\ast(t)}{d(s)}-\frac{\vec{\mathbf{t}}^\ast(s)}{d(t)},
\end{equation*}
\revise{where $r_b(s,t)$ is equal to $\frac{p_0(s,s)}{d(s)}+\frac{p_0(t,t)}{d(t)}-\frac{p_0(s,t)}{d(s)}-\frac{p_0(t,s)}{d(t)}$, i.e., the zeroth iteration part of $r_{\ell}(s,t)$}. Each of the $\ell_b$ iterations in \bwd updates $\vec{\mathbf{s}}^\ast$ (resp. $\vec{\mathbf{t}}^\ast$) as $\PM\vec{\mathbf{s}}^\ast$ (resp. $\PM\vec{\mathbf{t}}^\ast$), and then increases $r_b(s,t)$ by $\frac{\vec{\mathbf{s}}^\ast(s)}{d(s)}+\frac{\vec{\mathbf{t}}^\ast(t)}{d(t)}-\frac{\vec{\mathbf{s}}^\ast(t)}{d(s)}-\frac{\vec{\mathbf{t}}^\ast(s)}{d(t)}$ (Lines \ref{alg:smm-update}--\ref{alg:smm-update-r}). Particularly, at the end of the $i$-th iteration, $\forall{v\in V}$,
% \begin{small}
\begin{equation}\label{eq:stvec-star}
\vec{\mathbf{s}}^\ast(v)=p_i(v,s),\text{ and } \vec{\mathbf{t}}^\ast(v)=p_i(v,t).
\end{equation}
% \end{small}
In turn, when $\ell_b=\ell$ (as defined in Eq. \eqref{eq:ell}), the $r_b(s,t)$ returned by \bwd is an approximation of $r(s,t)$ satisfying $|r(s,t)-r_{\ell}(s,t)|\le \tfrac{\epsilon}{2}$. 

\begin{algorithm}[tb]
	\caption{\bwd}\label{alg:smm}
	\KwIn{Graph $G$, nodes $s,t$, $\ell_b$.}
	\KwOut{$r_b(s,t)$}
	$\vec{\mathbf{s}}^\ast\gets \vec{\mathbf{e}}_s$; $\vec{\mathbf{t}}^\ast\gets \vec{\mathbf{e}}_t$\;
	%$r_b(s,t) \gets \frac{\vec{\mathbf{s}}^\ast(s)}{d(s)}+\frac{\vec{\mathbf{t}}^\ast(t)}{d(t)}-2\frac{\vec{\mathbf{s}}^\ast(t)}{d(t)}$\;
	$r_b(s,t) \gets \frac{\vec{\mathbf{s}}^\ast(s)}{d(s)}+\frac{\vec{\mathbf{t}}^\ast(t)}{d(t)}-\frac{\vec{\mathbf{s}}^\ast(t)}{d(s)}-\frac{\vec{\mathbf{t}}^\ast(s)}{d(t)}$\;
	\For{$i\gets 1$ to $\ell_b$}{
		$\vec{\mathbf{s}}^\ast\gets \mathbf{P}\vec{\mathbf{s}}^\ast$; $\vec{\mathbf{t}}^\ast\gets \mathbf{P}\vec{\mathbf{t}}^\ast$\;\label{alg:smm-update}
		%$r_b(s,t)\gets r_b(s,t)+\frac{\vec{\mathbf{s}}^\ast(s)}{d(s)}+\frac{\vec{\mathbf{t}}^\ast(t)}{d(t)}-2\frac{\vec{\mathbf{s}}^\ast(t)}{d(t)}$\;
		$r_b(s,t)\gets r_b(s,t)+\frac{\vec{\mathbf{s}}^\ast(s)}{d(s)}+\frac{\vec{\mathbf{t}}^\ast(t)}{d(t)}-\frac{\vec{\mathbf{s}}^\ast(t)}{d(s)}-\frac{\vec{\mathbf{t}}^\ast(s)}{d(t)}$\;\label{alg:smm-update-r}
	}
\end{algorithm}

\bwd is essentially a deterministic graph traversal over $G$. Compared to na\"ive graph traversal, \bwd attains cache-friendly memory access patterns as it sequentially visits all reachable nodes from node $s,t$ via $\PM\vec{\mathbf{s}}^\ast$ and $\PM\vec{\mathbf{t}}^\ast$. Recall that $G$ is connected, meaning every two nodes can reach each other via one or multiple paths. Thus, when $\ell_b$ is sufficiently large and source nodes $s,t$ have many neighbors, the vectors $\vec{\mathbf{s}}^\ast$ and $\vec{\mathbf{t}}^\ast$ in \bwd turns to be dense rapidly and during each iteration in Algorithm \ref{alg:smm}, the matrix-vector multiplications at Line~\ref{alg:smm-update} will consume an exorbitant time of up to $O(m)$.
% , cache-friendly memory access patterns. 
% That said, \bwd entails an exorbitant overhead  rendering itself solely not a XXX for PER computation.

% That said, as remarked earlier, the deterministic graph traversal entails an exorbitant overhead when sources nodes connected other nodes, especially on large graphs with high average degrees (see Section \ref{sec:exp}). Particularly, each multiplication up to $O(m)$ time in the worse case.
% this process expensive as $\ell$ becoming larger on massive graphs with large degrees. 

With the strengths and weaknesses of \fwd and \bwd in mind, next we investigate how to integrate them into \algo, as well as orchestrate and optimize the entire \algo algorithm for enhanced empirical efficiency.

\subsubsection{Greedy Integration of \texttt{SMM} and \texttt{AMC}}\label{sec:greedy-combi}
% In what follows, we show how to utilize \bwd to significantly reduce the number of random walks needed in \fwd.
% orchestrate and optimize entire algorithm for enhanced empirical efficiency.

Fig. \ref{fig:overview} depicts an overview of our proposed \algo. At a high level, \algo is to strike a good trade-off between $\ell_b$ ($\ell_b\le \ell$) iterations of sparse matrix-vector multiplications in \bwd and random walks in \fwd. More precisely, \algo aims to compute $r^\ast_b(s,t)$ and $r^\ast_f(s,t)$ as follows:
\begin{equation}\label{eq:2-parts}
\begin{split}
{r}^\ast_b(s,t)&= \sum_{i=0}^{\ell_b}{\frac{p_i(s,s)}{d(s)}+\frac{p_i(t,t)}{d(t)}-\frac{ p_i(s,t)}{d(t)}-\frac{ p_i(t,s)}{d(s)}},\\
{r}^\ast_f(s,t)&= \sum_{i=\ell_b+1}^{\ell}{\frac{p_i(s,s)}{d(s)}+\frac{p_i(t,t)}{d(t)}-\frac{ p_i(s,t)}{d(t)}-\frac{ p_i(t,s)}{d(s)}},
\end{split}
\end{equation}
and finally combine them as the final ER via $r_{\ell}(s,t)=r^\ast_b(s,t)+r^\ast_f(s,t)$. Note that $r^\ast_b(s,t)$ aggregates the probabilities of visiting all nodes within the vicinity of $\ell_b$ hops away from nodes $s,t$, which can be efficiently computed using \bwd when $\ell_b$ is small, as remarked earlier. As regards $r^\ast_f(s,t)$, given that  it considers far-reaching nodes beyond $\ell_b$ hops from $s,t$, we resort to conducting random walks described in \fwd to estimate it. Instead of simulating random walks of lengths ranging from $\ell_b+1$ to $\ell$, we leverage the byproduct in \bwd so as to reduce the length as well as the amount of random walks needed in \fwd, as stated in what follows.

\begin{figure}[!t]
	\centering
	\resizebox{0.7\columnwidth}{!}{%
		% \includegraphics{overview.pdf}
		% \tikzset{every picture/.style={line width=0.75pt}} %set default line width to 0.75pt        
		\tikzset{every picture/.style={line width=0.75pt}} %set default line width to 0.75pt        
		
		\begin{tikzpicture}[x=0.75pt,y=0.75pt,yscale=-1,xscale=1,box/.style={rectangle,draw=white,thin, minimum size=0.5cm},]
			%uncomment if require: \path (0,189); %set diagram left start at 0, and has height of 189
			
			\node[box,fill=gray,fill opacity=0.3] at (505.5,90){}; 
			\node[box,fill=gray,fill opacity=0.3] at (404.5,89){}; 
			\node[box,fill=gray,fill opacity=0.3] at (464.5,89){};
			\node[box,fill=gray,fill opacity=0.3] at (303.5,89){};
			\node[box,fill=gray,fill opacity=0.3] at (243.5,89){};
			
			\node[box,fill=gray,fill opacity=0.3] at (464.5,109){};
			\node[box,fill=gray,fill opacity=0.3] at (424.5,109){};
			\node[box,fill=gray,fill opacity=0.3] at (384.5,109){};
			\node[box,fill=gray,fill opacity=0.3] at (303.5,109){};
			\node[box,fill=gray,fill opacity=0.3] at (263.5,109){};
			\node[box,fill=gray,fill opacity=0.3] at (223.5,109){};
			
			\node[box,fill=gray,fill opacity=0.3] at (464.5,129){};
			\node[box,fill=gray,fill opacity=0.3] at (444.5,129){};
			\node[box,fill=gray,fill opacity=0.3] at (424.5,129){};
			\node[box,fill=gray,fill opacity=0.3] at (404.5,129){};
			\node[box,fill=gray,fill opacity=0.3] at (384.5,129){};
			\node[box,fill=gray,fill opacity=0.3] at (303.5,129){};
			\node[box,fill=gray,fill opacity=0.3] at (283.5,129){};
			\node[box,fill=gray,fill opacity=0.3] at (263.5,129){};
			\node[box,fill=gray,fill opacity=0.3] at (243.5,129){};
			\node[box,fill=gray,fill opacity=0.3] at (223.5,129){};
			
			\node[box,fill=gray,fill opacity=0.3] at (464.5,69){};
			\node[box,fill=gray,fill opacity=0.3] at (444.5,69){};
			\node[box,fill=gray,fill opacity=0.3] at (384.5,69){};
			\node[box,fill=gray,fill opacity=0.3] at (303.5,69){};
			\node[box,fill=gray,fill opacity=0.3] at (283.5,69){};
			\node[box,fill=gray,fill opacity=0.3] at (223.5,69){};
			
			\node[box,fill=gray,fill opacity=0.3] at (464.5,49){};
			\node[box,fill=gray,fill opacity=0.3] at (424.5,49){};
			\node[box,fill=gray,fill opacity=0.3] at (303.5,49){};
			\node[box,fill=gray,fill opacity=0.3] at (263.5,49){};
			
			%Shape: Circle [id:dp32088817281894544] 
			\draw  [line width=1.5,fill=gray,fill opacity=0.3]  (8,91) .. controls (8,83.82) and (13.82,78) .. (21,78) .. controls (28.18,78) and (34,83.82) .. (34,91) .. controls (34,98.18) and (28.18,104) .. (21,104) .. controls (13.82,104) and (8,98.18) .. (8,91) -- cycle ;
			%Shape: Grid [id:dp02521835021622354] 
			\draw  [draw opacity=0][line width=1.5]  (214,39) -- (314.5,39) -- (314.5,139.5) -- (214,139.5) -- cycle ; \draw  [line width=1.5]  (214,39) -- (214,139.5)(234,39) -- (234,139.5)(254,39) -- (254,139.5)(274,39) -- (274,139.5)(294,39) -- (294,139.5)(314,39) -- (314,139.5) ; \draw  [line width=1.5]  (214,39) -- (314.5,39)(214,59) -- (314.5,59)(214,79) -- (314.5,79)(214,99) -- (314.5,99)(214,119) -- (314.5,119)(214,139) -- (314.5,139) ; \draw  [line width=1.5]   ;
			%Shape: Grid [id:dp8450314855262502] 
			\draw  [draw opacity=0,fill=gray,fill opacity=0.3][line width=1.5]  (137,39) -- (157.5,39) -- (157.5,139.5) -- (137,139.5) -- cycle ; \draw  [line width=1.5]  (137,39) -- (137,139.5)(157,39) -- (157,139.5) ; \draw  [line width=1.5]  (137,39) -- (157.5,39)(137,59) -- (157.5,59)(137,79) -- (157.5,79)(137,99) -- (157.5,99)(137,119) -- (157.5,119)(137,139) -- (157.5,139) ; \draw  [line width=1.5]   ;
			%Straight Lines [id:da7169477559127133] 
			\draw [line width=0.75]    (34,91) -- (57.6,57.39) ;
			\draw [shift={(58.75,55.75)}, rotate = 125.07] [fill={rgb, 255:red, 0; green, 0; blue, 0 }  ][line width=0.08]  [draw opacity=0] (9.6,-2.4) -- (0,0) -- (9.6,2.4) -- cycle    ;
			%Straight Lines [id:da7384749684805199] 
			\draw [line width=0.75]    (57.75,56.25) -- (92.9,42) ;
			\draw [shift={(94.75,41.25)}, rotate = 157.93] [fill={rgb, 255:red, 0; green, 0; blue, 0 }  ][line width=0.08]  [draw opacity=0] (9.6,-2.4) -- (0,0) -- (9.6,2.4) -- cycle    ;
			%Straight Lines [id:da5103328071364268] 
			\draw [line width=0.75]    (94.75,41.75) -- (133.79,49.84) ;
			\draw [shift={(135.75,50.25)}, rotate = 191.71] [fill={rgb, 255:red, 0; green, 0; blue, 0 }  ][line width=0.08]  [draw opacity=0] (9.6,-2.4) -- (0,0) -- (9.6,2.4) -- cycle    ;
			%Straight Lines [id:da9743548989406379] 
			\draw [line width=0.75]    (34.75,90.75) -- (63.09,71.86) ;
			\draw [shift={(64.75,70.75)}, rotate = 146.31] [fill={rgb, 255:red, 0; green, 0; blue, 0 }  ][line width=0.08]  [draw opacity=0] (9.6,-2.4) -- (0,0) -- (9.6,2.4) -- cycle    ;
			%Straight Lines [id:da3151465493211625] 
			\draw [line width=0.75]    (64.75,71.25) -- (96.82,62.76) ;
			\draw [shift={(98.75,62.25)}, rotate = 165.17] [fill={rgb, 255:red, 0; green, 0; blue, 0 }  ][line width=0.08]  [draw opacity=0] (9.6,-2.4) -- (0,0) -- (9.6,2.4) -- cycle    ;
			%Straight Lines [id:da1680744073674152] 
			\draw [line width=0.75]    (98.75,62.25) -- (134.79,69.36) ;
			\draw [shift={(136.75,69.75)}, rotate = 191.16] [fill={rgb, 255:red, 0; green, 0; blue, 0 }  ][line width=0.08]  [draw opacity=0] (9.6,-2.4) -- (0,0) -- (9.6,2.4) -- cycle    ;
			%Straight Lines [id:da8919092112907709] 
			\draw [line width=0.75]    (34,91) -- (72.25,92.66) ;
			\draw [shift={(74.25,92.75)}, rotate = 182.49] [fill={rgb, 255:red, 0; green, 0; blue, 0 }  ][line width=0.08]  [draw opacity=0] (9.6,-2.4) -- (0,0) -- (9.6,2.4) -- cycle    ;
			%Straight Lines [id:da09096573221327509] 
			\draw [line width=0.75]    (72.75,92.75) -- (104.28,87.57) ;
			\draw [shift={(106.25,87.25)}, rotate = 170.68] [fill={rgb, 255:red, 0; green, 0; blue, 0 }  ][line width=0.08]  [draw opacity=0] (9.6,-2.4) -- (0,0) -- (9.6,2.4) -- cycle    ;
			%Straight Lines [id:da5185040131830372] 
			\draw [line width=0.75]    (104.75,87.25) -- (135.25,87.72) ;
			\draw [shift={(137.25,87.75)}, rotate = 180.88] [fill={rgb, 255:red, 0; green, 0; blue, 0 }  ][line width=0.08]  [draw opacity=0] (9.6,-2.4) -- (0,0) -- (9.6,2.4) -- cycle    ;
			%Straight Lines [id:da38839492857760405] 
			\draw    (34,91) -- (63.62,134.1) ;
			\draw [shift={(64.75,135.75)}, rotate = 235.51] [fill={rgb, 255:red, 0; green, 0; blue, 0 }  ][line width=0.08]  [draw opacity=0] (9.6,-2.4) -- (0,0) -- (9.6,2.4) -- cycle    ;
			%Straight Lines [id:da6196834137816927] 
			\draw    (64.5,134.5) -- (96.77,138.5) ;
			\draw [shift={(98.75,138.75)}, rotate = 187.07] [fill={rgb, 255:red, 0; green, 0; blue, 0 }  ][line width=0.08]  [draw opacity=0] (9.6,-2.4) -- (0,0) -- (9.6,2.4) -- cycle    ;
			%Straight Lines [id:da04228474704147245] 
			\draw    (98.25,138.25) -- (134.3,130.19) ;
			\draw [shift={(136.25,129.75)}, rotate = 167.39] [fill={rgb, 255:red, 0; green, 0; blue, 0 }  ][line width=0.08]  [draw opacity=0] (9.6,-2.4) -- (0,0) -- (9.6,2.4) -- cycle    ;
			%Straight Lines [id:da7939636772577605] 
			\draw    (34,91) -- (68.12,115.1) ;
			\draw [shift={(69.75,116.25)}, rotate = 215.23] [fill={rgb, 255:red, 0; green, 0; blue, 0 }  ][line width=0.08]  [draw opacity=0] (9.6,-2.4) -- (0,0) -- (9.6,2.4) -- cycle    ;
			%Straight Lines [id:da6549323163443623] 
			\draw    (69.25,116.25) -- (103.25,115.31) ;
			\draw [shift={(105.25,115.25)}, rotate = 178.41] [fill={rgb, 255:red, 0; green, 0; blue, 0 }  ][line width=0.08]  [draw opacity=0] (9.6,-2.4) -- (0,0) -- (9.6,2.4) -- cycle    ;
			%Straight Lines [id:da818964575205186] 
			\draw    (104.25,115.25) -- (133.79,109.15) ;
			\draw [shift={(135.75,108.75)}, rotate = 168.34] [fill={rgb, 255:red, 0; green, 0; blue, 0 }  ][line width=0.08]  [draw opacity=0] (9.6,-2.4) -- (0,0) -- (9.6,2.4) -- cycle    ;
			%Shape: Brace [id:dp7860631675137759] 
			\draw  [line width=1.5]  (40,152) .. controls (40.06,156.67) and (42.42,158.97) .. (47.09,158.91) -- (77.31,158.51) .. controls (83.98,158.42) and (87.34,160.71) .. (87.4,165.38) .. controls (87.34,160.71) and (90.64,158.34) .. (97.31,158.25)(94.31,158.29) -- (128.84,157.84) .. controls (133.51,157.78) and (135.81,155.42) .. (135.75,150.75) ;
			%Shape: Brace [id:dp21177206620260525] 
			\draw  [line width=1.5]  (208.5,151) .. controls (208.51,155.67) and (210.84,158) .. (215.51,157.99) -- (334.26,157.76) .. controls (340.93,157.75) and (344.27,160.07) .. (344.28,164.74) .. controls (344.27,160.07) and (347.59,157.73) .. (354.26,157.72)(351.26,157.73) -- (473.01,157.49) .. controls (477.68,157.48) and (480.01,155.15) .. (480,150.48) ;
			%Shape: Grid [id:dp1284577725366658] 
			\draw  [draw opacity=0][line width=1.5]  (374.5,39) -- (475,39) -- (475,139.5) -- (374.5,139.5) -- cycle ; \draw  [line width=1.5]  (374.5,39) -- (374.5,139.5)(394.5,39) -- (394.5,139.5)(414.5,39) -- (414.5,139.5)(434.5,39) -- (434.5,139.5)(454.5,39) -- (454.5,139.5)(474.5,39) -- (474.5,139.5) ; \draw  [line width=1.5]  (374.5,39) -- (475,39)(374.5,59) -- (475,59)(374.5,79) -- (475,79)(374.5,99) -- (475,99)(374.5,119) -- (475,119)(374.5,139) -- (475,139) ; \draw  [line width=1.5]   ;
			%Shape: Grid [id:dp6077324674585625] 
			\draw  [draw opacity=0.3][line width=1.5]  (495,39) -- (515.5,39) -- (515.5,139.5) -- (495,139.5) -- cycle ; \draw  [line width=1.5]  (495,39) -- (495,139.5)(515,39) -- (515,139.5) ; \draw  [line width=1.5]  (495,39) -- (515.5,39)(495,59) -- (515.5,59)(495,79) -- (515.5,79)(495,99) -- (515.5,99)(495,119) -- (515.5,119)(495,139) -- (515.5,139) ; \draw  [line width=1.5]   ;
			%Rounded Rect [id:dp8233320201570238] 
			\draw  [dash pattern={on 4.5pt off 4.5pt}] (4.25,54.05) .. controls (4.25,41.18) and (14.68,30.75) .. (27.55,30.75) -- (148.45,30.75) .. controls (161.32,30.75) and (171.75,41.18) .. (171.75,54.05) -- (171.75,123.95) .. controls (171.75,136.82) and (161.32,147.25) .. (148.45,147.25) -- (27.55,147.25) .. controls (14.68,147.25) and (4.25,136.82) .. (4.25,123.95) -- cycle ;
			%Rounded Rect [id:dp04931006278433525] 
			\draw  [dash pattern={on 4.5pt off 4.5pt}] (198.25,54.55) .. controls (198.25,41.68) and (208.68,31.25) .. (221.55,31.25) -- (514.95,31.25) .. controls (527.82,31.25) and (538.25,41.68) .. (538.25,54.55) -- (538.25,124.45) .. controls (538.25,137.32) and (527.82,147.75) .. (514.95,147.75) -- (221.55,147.75) .. controls (208.68,147.75) and (198.25,137.32) .. (198.25,124.45) -- cycle ;
			
			% Text Node
			\draw (55,168.4) node [anchor=north west][inner sep=0.75pt]  [font=\Large]  {$1,2,\cdots, \ell_f$};
			% Text Node
			\draw (336,166.9) node [anchor=north west][inner sep=0.75pt]  [font=\Large]  {$\ell_b$};
			% Text Node
			\draw (322,82.4) node [anchor=north west][inner sep=0.75pt]  [font=\Large]  {$\times \cdots \times$};
			% Text Node
			\draw (178.5,81.4) node [anchor=north west][inner sep=0.75pt]  [font=\Large]  {$\times$};
			% Text Node
			\draw (11,83) node [anchor=north west][inner sep=0.75pt]  [font=\Large]  {$s/t$};
			% Text Node
			\draw (139,44.4) node [anchor=north west][inner sep=0.75pt]  [font=\footnotesize]  {$0.3$};
			% Text Node
			\draw (139,65.4) node [anchor=north west][inner sep=0.75pt]  [font=\footnotesize]  {$0.1$};
			% Text Node
			\draw (139,85.4) node [anchor=north west][inner sep=0.75pt]  [font=\footnotesize]  {$0.1$};
			% Text Node
			\draw (138.5,104.9) node [anchor=north west][inner sep=0.75pt]  [font=\footnotesize]  {$0.3$};
			% Text Node
			\draw (139,125.4) node [anchor=north west][inner sep=0.75pt]  [font=\footnotesize]  {$0.2$};
			% Text Node
			\draw (216.5,123.9) node [anchor=north west][inner sep=0.75pt]  [font=\footnotesize]  {$0.2$};
			% Text Node
			\draw (236,124.4) node [anchor=north west][inner sep=0.75pt]  [font=\footnotesize]  {$0.2$};
			% Text Node
			\draw (256,124.4) node [anchor=north west][inner sep=0.75pt]  [font=\footnotesize]  {$0.2$};
			% Text Node
			\draw (276,124.4) node [anchor=north west][inner sep=0.75pt]  [font=\footnotesize]  {$0.2$};
			% Text Node
			\draw (295.5,124.4) node [anchor=north west][inner sep=0.75pt]  [font=\footnotesize]  {$0.2$};
			% Text Node
			\draw (296,84.4) node [anchor=north west][inner sep=0.75pt]  [font=\footnotesize]  {$0.5$};
			% Text Node
			\draw (296,43.4) node [anchor=north west][inner sep=0.75pt]  [font=\footnotesize]  {$0.5$};
			% Text Node
			\draw (236,84.4) node [anchor=north west][inner sep=0.75pt]  [font=\footnotesize]  {$0.5$};
			% Text Node
			\draw (256,43.4) node [anchor=north west][inner sep=0.75pt]  [font=\footnotesize]  {$0.5$};
			% Text Node
			\draw (296,103.4) node [anchor=north west][inner sep=0.75pt]  [font=\footnotesize]  {$0.3$};
			% Text Node
			\draw (256,103.4) node [anchor=north west][inner sep=0.75pt]  [font=\footnotesize]  {$0.3$};
			% Text Node
			\draw (216,103.4) node [anchor=north west][inner sep=0.75pt]  [font=\footnotesize]  {$0.3$};
			% Text Node
			\draw (296,63.4) node [anchor=north west][inner sep=0.75pt]  [font=\footnotesize]  {$0.3$};
			% Text Node
			\draw (216,63.4) node [anchor=north west][inner sep=0.75pt]  [font=\footnotesize]  {$0.3$};
			% Text Node
			\draw (276,63.4) node [anchor=north west][inner sep=0.75pt]  [font=\footnotesize]  {$0.3$};
			% Text Node
			\draw (377,123.9) node [anchor=north west][inner sep=0.75pt]  [font=\footnotesize]  {$0.2$};
			% Text Node
			\draw (396.5,124.4) node [anchor=north west][inner sep=0.75pt]  [font=\footnotesize]  {$0.2$};
			% Text Node
			\draw (416.5,124.4) node [anchor=north west][inner sep=0.75pt]  [font=\footnotesize]  {$0.2$};
			% Text Node
			\draw (436.5,124.4) node [anchor=north west][inner sep=0.75pt]  [font=\footnotesize]  {$0.2$};
			% Text Node
			\draw (456,124.4) node [anchor=north west][inner sep=0.75pt]  [font=\footnotesize]  {$0.2$};
			% Text Node
			\draw (456.5,84.4) node [anchor=north west][inner sep=0.75pt]  [font=\footnotesize]  {$0.5$};
			% Text Node
			\draw (456.5,43.4) node [anchor=north west][inner sep=0.75pt]  [font=\footnotesize]  {$0.5$};
			% Text Node
			\draw (396.5,84.4) node [anchor=north west][inner sep=0.75pt]  [font=\footnotesize]  {$0.5$};
			% Text Node
			\draw (416.5,43.4) node [anchor=north west][inner sep=0.75pt]  [font=\footnotesize]  {$0.5$};
			% Text Node
			\draw (456.5,103.4) node [anchor=north west][inner sep=0.75pt]  [font=\footnotesize]  {$0.3$};
			% Text Node
			\draw (416.5,103.4) node [anchor=north west][inner sep=0.75pt]  [font=\footnotesize]  {$0.3$};
			% Text Node
			\draw (376.5,103.4) node [anchor=north west][inner sep=0.75pt]  [font=\footnotesize]  {$0.3$};
			% Text Node
			\draw (456.5,63.4) node [anchor=north west][inner sep=0.75pt]  [font=\footnotesize]  {$0.3$};
			% Text Node
			\draw (376.5,63.4) node [anchor=north west][inner sep=0.75pt]  [font=\footnotesize]  {$0.3$};
			% Text Node
			\draw (436.5,63.4) node [anchor=north west][inner sep=0.75pt]  [font=\footnotesize]  {$0.3$};
			% Text Node
			\draw (478,80.9) node [anchor=north west][inner sep=0.75pt]  [font=\Large]  {$\times$};
			% Text Node
			\draw (500.5,83.4) node [anchor=north west][inner sep=0.75pt]  [font=\footnotesize]  {$\mathbf{1}$};
			% Text Node
			\draw (517,80) node [anchor=north west][inner sep=0.75pt]  [font=\Large]  {$s/t$};
			% Text Node
			\draw (40.5,6) node [anchor=north west][inner sep=0.75pt] [font=\LARGE]  [align=left] {random walks};
			% Text Node
			\draw (247.5,7.5) node [anchor=north west][inner sep=0.75pt] [font=\LARGE]  [align=left] {sparse matrix-vector multiplications};
		\end{tikzpicture}
	}
	\vspace{-3mm}
	\caption{Overview of \algo}
	\label{fig:overview}
%	\vspace{-2mm}
\end{figure}
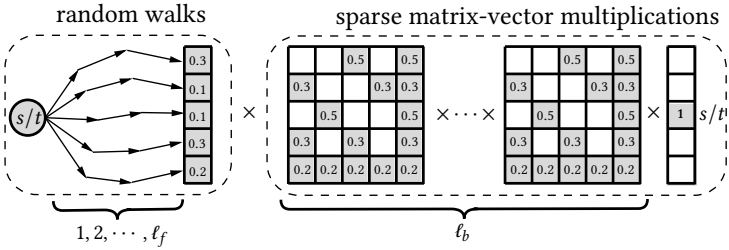

Let $\vec{\mathbf{s}}^\ast$ and $\vec{\mathbf{t}}^\ast$ be the vectors after $\ell_b$ iterations in \bwd. By Eq. \eqref{eq:stvec-star}, we have $\vec{\mathbf{s}}^\ast(v)=p_{\ell_b}(v,s)$ and $\vec{\mathbf{t}}^\ast(v)=p_{\ell_b}(v,t)\ \forall{v\in V}$. As a consequence, $r^\ast_f(s,t)$ in Eq. \eqref{eq:2-parts} can be rewritten as follows:
\begin{align*}
r^\ast_f(s,t)&=\sum_{i=1}^{\ell-\ell_b}\frac{\sum_{v\in V}p_{i}(s,v)\cdot p_{\ell_b}(v,s)}{d(s)}+ \frac{\sum_{v\in V} p_{i}(t,v)\cdot p_{\ell_b}(v,t)}{d(t)} \\
&\mathrel{\phantom{=\sum_{i=1}^{\ell-\ell_b}}}\mathop{-} \frac{\sum_{v\in V} p_{i}(s,v)\cdot p_{\ell_b}(v,t)}{d(t)}-\frac{\sum_{v\in V} p_{i}(t,v)\cdot p_{\ell_b}(v,s)}{d(s)} \\
& = \sum_{i=1}^{\ell-\ell_b}{\sum_{v\in V}{\left(\big(p_{i}(s,v)-p_{i}(t,v)\big)\cdot\Big( \frac{\vec{\mathbf{s}}^\ast(v)}{d(s)}- \frac{\vec{\mathbf{t}}^\ast(v)}{d(t)}\Big)\right)}}.
\end{align*}
Observe that the above equation is identical to $q(s,t)$ in Eq. \eqref{eq:qst} by letting $\ell_f=\ell-\ell_b$, $\vec{\mathbf{s}}=\vec{\mathbf{s}}^\ast$, and $\vec{\mathbf{t}}=\vec{\mathbf{t}}^\ast$. Namely, we can directly estimate $r^\ast_f(s,t)$ by invoking \fwd with $\ell_f=\ell-\ell_b$, $\vec{\mathbf{s}}=\vec{\mathbf{s}}^\ast$, and $\vec{\mathbf{t}}=\vec{\mathbf{t}}^\ast$. By doing so, the length of random walks varies from $1$ to $\ell-\ell_b$, rather than from $\ell_b+1$ to $\ell$.

Furthermore, a vital benefit from the intermediate result $\vec{\mathbf{s}}^\ast$ and $\vec{\mathbf{t}}^\ast$ is that we can significantly prune the random walks needed in \fwd.
Specifically, recall that in Eq. \eqref{eq:etas-star}, the total number of random walks simulated in \fwd is proportional to $\psi^2$ defined in Eq. \eqref{eq:Delta-st}. In Section \ref{sec:amc}, we input $\vec{\mathbf{s}}=\vec{\mathbf{e}}_s$ and $\vec{\mathbf{t}}=\vec{\mathbf{e}}_t$ to \fwd; and thus $\psi=2\left\lceil \tfrac{\ell}{2} \right\rceil\left(\frac{1}{d(s)}+\frac{1}{d(t)}\right)$.
By contrast, \algo invokes \fwd with $\vec{\mathbf{s}}=\vec{\mathbf{s}}^\ast$ and $\vec{\mathbf{t}}=\vec{\mathbf{t}}^\ast$, whose maximum values will dwindle to small ones when increasing $\ell_b$ since they are the probabilities of length-$\ell_b$ random walks from a node visiting $s$ and $t$. Empirically, when $\ell_b$ is sufficiently large, the maximum element in $\vec{\mathbf{s}}^\ast$ or $\vec{\mathbf{t}}^\ast$ is typically less than $0.2$, and hence $\psi\approx\frac{2}{5}\left\lceil \tfrac{\ell}{2} \right\rceil\left(\frac{1}{d(s)}+\frac{1}{d(t)}\right)$, yielding at least a $96\%$ reduction in the amount of random walks in \fwd. Additionally, with $\vec{\mathbf{s}}=\vec{\mathbf{s}}^\ast$ and $\vec{\mathbf{t}}=\vec{\mathbf{t}}^\ast$ as input to \fwd, the variance $\hat{\sigma}^2$ in \fwd can be considerably narrowed, thereby facilitating the early termination of \fwd (at Line~\ref{alg:fwd-stop} in Algorithm \ref{alg:fwd}). To see this, we first recall that the random variable $Z_k$ in \fwd is the weighted sums of $\ell_f$ entries in $\vec{\mathbf{s}}$ or $\vec{\mathbf{t}}$. In \algo, entries in $\vec{\mathbf{s}}^\ast$ (resp.\ $\vec{\mathbf{t}}^\ast$) are more evenly distributed compared to $\vec{\mathbf{e}}_s$ (resp.\ $\vec{\mathbf{e}}_t$), resulting in smaller variance $\hat{\sigma}^2$.
In sum, we can tremendously expedite the computational efficiency of \fwd by picking a large $\ell_b$, which ensures the maximum values in $\vec{\mathbf{s}}^\ast$ and $\vec{\mathbf{t}}^\ast$ are small and values in them are evenly distributed.

% as random variables $X^{(k)}_{ss},X^{(k)}_{st}$,and $X^{(k)}_{tt}$ are not Bernoulli variables, facilitating the early termination of \fwd (Line 28 in Algorithm \ref{alg:fwd}) with the power of the empirical Bernstein inequality
% \fwd can not fully harness the power of the empirical Bernstein inequality ($\vec{\mathbf{s}}=\vec{\mathbf{e}}_s$ and $\vec{\mathbf{t}}=\vec{\mathbf{e}}_t$) as the random variable in \fwd is Bernoulli variable and the variance is large. In \algo, , and hence \fwd can terminate fast .

% As mentioned in Section \ref{sec:smm}, a large $\ell_b$ leads to an exorbitant time cost in \bwd.
% The problem turns to be choosing the value of $\ell_b$. 
To minimize the total cost incurred by \algo, we adopt a greedy strategy to determine $\ell_b$, i.e., the switch point between \bwd and \fwd. In a nutshell, we cease the iterative sparse matrix-vector multiplications (\bwd) and start to simulate random walks (\fwd) once the overhead of the former exceeds that of the latter. To be specific, consider $\ell_b$-th iteration in \bwd. Let $V_s$ and $V_t$ be the sets of nodes having non-zero entries in $\vec{\mathbf{s}}$ and $\vec{\mathbf{t}}$, respectively. Intuitively, the total number of operations when performing $\PM\vec{\mathbf{s}}$ and $\PM\vec{\mathbf{t}}$ in the following iteration is bounded by $\sum_{v\in V_s}{d(v)}+\sum_{v\in V_t}{d(v)}$.
On the other hand, the total number of random samples required in \fwd is currently $h(\ell-\ell_b)$. Thus, \algo greedily performs sparse matrix-vector multiplications as in \bwd until the following inequality holds.
\begin{equation}\label{eq:smm-amc}
\sum_{v\in V_s}{d(v)}+\sum_{v\in V_t}{d(v)} > h(\ell-\ell_b).
\end{equation}

\begin{algorithm}[tb]
\caption{\algo}\label{alg:bere}
\KwIn{Graph $G$, nodes $s,t$, error threshold $\epsilon$, the maximum number $\tau$ of batches in \fwd, failure probability $\delta$}
\KwOut{$r^{\prime}(s,t)$}
Calculate $\ell$ according to Eq. \eqref{eq:ell}\;
$\ell_b \gets 0$\;
% $\rho(s,s)\gets 0; \rho(s,t)\gets 0; \rho(t,t)\gets 0$\;
% $r_b(s,t) \gets 0$\;
% $\vec{\mathbf{t}}\gets \vec{\mathbf{e}}_t; \vec{\mathbf{s}}\gets \vec{\mathbf{e}}_s$\;
%{\nonl{Lines 3--4 are the same as Lines 1--2 in Algorithm \ref{alg:smm}\;}\setcounter{AlgoLine}{4}}
%\tcc{Lines 3--4 are the same as Lines 1--2 in Algorithm \ref{alg:smm}}
$\vec{\mathbf{s}}^\ast\gets \vec{\mathbf{e}}_s$; $\vec{\mathbf{t}}^\ast\gets \vec{\mathbf{e}}_t$\;
%$r_b(s,t) \gets \frac{\vec{\mathbf{s}}^\ast(s)}{d(s)}+\frac{\vec{\mathbf{t}}^\ast(t)}{d(t)}-2\frac{\vec{\mathbf{s}}^\ast(t)}{d(t)}$\;
$r_b(s,t) \gets \frac{\vec{\mathbf{s}}^\ast(s)}{d(s)}+\frac{\vec{\mathbf{t}}^\ast(t)}{d(t)}-\frac{\vec{\mathbf{s}}^\ast(t)}{d(s)}-\frac{\vec{\mathbf{t}}^\ast(s)}{d(t)}$\;
\Repeat{Eq. \eqref{eq:smm-amc} holds or $\ell_b\ge \ell$}{
% $r_b(s,t)\gets r_b(s,t)+\frac{\vec{\mathbf{s}}(s)}{d(s)}+\frac{\vec{\mathbf{t}}(t)}{d(t)}-2\frac{\vec{\mathbf{s}}(t)}{d(t)}$\;
% $\vec{\mathbf{t}}\gets \mathbf{P}\vec{\mathbf{t}}; \vec{\mathbf{s}}\gets \mathbf{P}\vec{\mathbf{s}}$\;
% \For{$i\gets 1$ to $\ell_b$}{
%{\nonl{Lines 6--7 are the same as Lines 4--5 in Algorithm \ref{alg:smm}\;}
%\setcounter{AlgoLine}{7}}
%\tcc{Lines 6--7 are the same as Lines 4--5 in Algorithm \ref{alg:smm}}
$\vec{\mathbf{s}}^\ast\gets \mathbf{P}\vec{\mathbf{s}}^\ast$; $\vec{\mathbf{t}}^\ast\gets \mathbf{P}\vec{\mathbf{t}}^\ast$\;
%$r_b(s,t)\gets r_b(s,t)+\frac{\vec{\mathbf{s}}^\ast(s)}{d(s)}+\frac{\vec{\mathbf{t}}^\ast(t)}{d(t)}-2\frac{\vec{\mathbf{s}}^\ast(t)}{d(t)}$\;
$r_b(s,t)\gets r_b(s,t)+\frac{\vec{\mathbf{s}}^\ast(s)}{d(s)}+\frac{\vec{\mathbf{t}}^\ast(t)}{d(t)}-\frac{\vec{\mathbf{s}}^\ast(t)}{d(s)}-\frac{\vec{\mathbf{t}}^\ast(s)}{d(t)}$\;
$\ell_b\gets \ell_b+1$\;
% Calculate $\psi_s$ and $\psi_t$ according to Eq. \eqref{eq:Delta-st}\;
% Let $V_s$ and $V_t$ be the set of nodes that have non-zero entries in $\vec{\mathbf{s}}$ and $\vec{\mathbf{t}}$, respectively\;
% }
% \lIf{Eq. \eqref{eq:smm-amc} holds}{\textbf{break}}
}
% \lIf{$d(t)>d(s)$}{
$r_f(s,t)\gets\mathtt{AMC}(G,s,t,\vec{\mathbf{s}}^\ast,\vec{\mathbf{t}}^\ast,\epsilon,\ell-\ell_b,\tau,\delta)$\;
% }\lElse{
% $r_f(s,t)\gets$\fwd{}$(G,t,s,\vec{\mathbf{t}},\vec{\mathbf{s}},\epsilon,\ell-\ell_b,\tau,\delta)$
% }
$r^{\prime}(s,t)\gets r_f(s,t)+r_b(s,t)$\;
\end{algorithm}

\subsection{Complete \texttt{GEER} Algorithm and Analysis}\label{sec:algo-als}

The complete pseudo-code of \algo is presented in Algorithm \ref{alg:bere}. Algorithm \ref{alg:bere} starts by taking as input a graph $G$, two nodes $s,t$, the error threshold $\epsilon$, the number of batches $\tau$, and the failure probability $\delta$. First, \algo calculates $\ell$ according to Eq. \eqref{eq:ell} (Line 1). We assume that $\lambda$ in Eq. \eqref{eq:ell} is obtained in a pre-processing step. Then, \algo initializes $\ell_b=0$ (Line 2) and starts to invoke \bwd (Lines 3--9). Distinct from \bwd, the iterative procedure is terminated when Eq. \eqref{eq:smm-amc} holds (Line 9) or the number $\ell_b$ of iterations performed thus far exceeds $\ell$ (i.e., $\ell_b\ge \ell$). 
% After that, we iteratively updates   and the number $\ell_b$ of performed iterations is raised by 1 (Lines 9-10). 
% Based on the updated $\vec{\mathbf{s}}$ and $\vec{\mathbf{t}}$, \algo calculates $\psi_s$ and $\psi_t$ according to Eq. \eqref{eq:time-fwd} (Line 11).  
After that, Algorithm \ref{alg:fwd} is invoked with parameters $G$, $s,t$, $\epsilon$, $\delta$, $\vec{\mathbf{s}}^\ast$, $\vec{\mathbf{t}}^\ast$, and $\ell_f=\ell-\ell_b$ (Line 10). Let $r_f(s,t)$ be the output of \fwd. 
Eventually, \algo computes $r^{\prime}(s,t)$ by $r^{\prime}(s,t)= r_f(s,t)+r_b(s,t)$ and return it as an $\epsilon$-approximate ER of node pair $(s,t)$ (Lines 10--11). According to Theorem \ref{lem:fwd},
$|r_f(s,t)-r^\ast_f(s,t)|\le \frac{\epsilon}{2}$ holds with probability at least $1-\delta$. With $r_b(s,t)=r^\ast_b(s,t)$ as defined in Eq. \eqref{eq:2-parts}, we derive
\begin{equation*}
|r^{\prime}(s,t)-r_{\ell}(s,t)|=|r_f(s,t)+r_b(s,t)-r^\ast_f(s,t)-r^\ast_b(s,t)|\le \frac{\epsilon}{2}.
\end{equation*}
Plugging it into $|r(s,t)-r_{\ell}(s,t)|\le \tfrac{\epsilon}{2}$ leads to Ineq. \eqref{eq:epsilon}, indicating that $r^{\prime}(s,t)$ returned by Algorithm \ref{alg:bere} is an $\epsilon$-approximation of $r(s,t)$.
% The following theorem establishes \algo's correctness.
% \subsection{Analysis}\label{sec:als}
% \subsubsection{\bf Correctness}
% Theorem \ref{lem:main} establishes \algo's correctness.
% \begin{theorem}\label{lem:main}
% $r^{\prime}(s,t)$ returned by Algorithm \ref{alg:bere} satisfies Eq. \eqref{eq:epsilon}.
% \end{theorem}

% \subsubsection{\bf Complexity}
% In Lines 6-7 of Algorithm \ref{alg:bere}, $\PM\vec{\mathbf{s}}$ and $\PM\vec{\mathbf{t}}$ takes at most $O(m)$ time. Accordingly, the total cost incurred by \bwd is $O(m\ell_b)$. Moreover, 
The invocation of Algorithm \ref{alg:fwd} at Line 10 requires $O(h(\ell-\ell_b)\cdot (\ell-\ell_b))$ time. 
Notice that we ensure the cost of each iteration in Lines 6--7 not exceed that of using \fwd by Eq. \eqref{eq:smm-amc} (Line 9). Consequently, the overall time complexity of \algo is bounded by Eq. \eqref{eq:time-amc} as well.

%% file: tex/experiment.tex
\section{Experiments}\label{sec:exp}
\revise{This section experimentally evaluates the efficiency and accuracy of
our proposed algorithms for answering $\epsilon$-approximate PER queries.}

\subsection{Experimental Setup}
\vspace{-1mm}
\header
{\bf Datasets, Query sets, and Groundtruth.}
To evaluate both the efficiency and accuracy of our proposed \fwd and \algo, we conduct experiments on 6 real datasets of various sizes, as described in Table \ref{tab:dataset}. All datasets are collected from SNAP\footnote{\url{https://snap.stanford.edu}} and used as benchmark datasets in prior work \cite{pankdd21}. For each dataset, we pick 100 node pairs uniformly at random as the random query set and randomly select 100 edges out of edge set $E$ as the edge query set. The ground-truth ER values for these query node pairs are obtained by applying \bwd with $1000$ iterations in parallel ($\epsilon$ is roughly $10^{-8}$--$10^{-6}$).

\header
{\bf Implementation Details.}
All experiments are conducted on a Linux machine with an Intel Xeon(R) Gold 6240@2.60GHz 32-core processor and 377GB of RAM. None of the experiments need anywhere near all the memory, nor any of the parallelism. Each experiment loaded the graph used into memory before beginning
any timings. The eigenvalues $\lambda_2$ and $\lambda_n$ of each tested dataset are approximated via \texttt{ARPACK}\footnote{\url{http://www.caam.rice.edu/software/ARPACK}} \cite{lehoucq1998arpack}.
For a fair comparison, all tested algorithms are implemented in C++ and compiled by g++ 7.5 with $\mathtt{-O3}$ optimization. For reproducibility, the source code is available at: {\url{https://github.com/AnryYang/GEER}}.

\begin{table}
\centering
\caption{Statistics of Datasets}
\label{tab:dataset}
\vspace{-3mm}
%\begin{footnotesize}
\begin{tabular}{lrrc}
\toprule
\bf Name  & \#nodes ($n$) & \#edges ($m$) & avg. degree\\
\midrule
Facebook \cite{mcauley2012learning} &  4,039  & 88,234  & 43.69 \\
DBLP \cite{yang2012defining} & 317,080 & 1,049,866 & 6.62 \\
YouTube \cite{yang2012defining}  & 1,134,890 & 2,987,624 & 5.27 \\ 
Orkut \cite{yang2012defining} & 3,072,441 & 117,185,082 & 76.28 \\
LiveJournal \cite{yang2012defining} & 3,997,962 & 34,681,189 & 17.35 \\
Friendster \cite{yang2012defining} & 65,608,366 & 1,806,067,135 & 55.06 \\
\bottomrule
\end{tabular}
%\end{footnotesize}
\vspace{3mm}
\end{table}

\header
{\bf Competitors and Parameters.} 
For random queries, we compare \fwd and \algo with \bwd, the state-of-the-art solutions $\mathtt{TP}$ and $\mathtt{TPC}$ discussed in Section \ref{sec:mcpt}, \revise{as well as the random projection method $\mathtt{RP}$ \cite{spielman2011graph}. We also include the $\mathtt{EXACT}$ method for comparison that requires computing the pseudo-inverse of matrix $\mathbf{D}-\mathbf{A}$ in Eq. \eqref{eq:er}.}
For edge queries, i.e., $(s,t)\in E$, we evaluate \fwd and \algo against \bwd, and \revise{two baseline solutions $\mathtt{HAY}$ \cite{hayashi2016efficient} and $\mathtt{MC2}$ \cite{pankdd21}, which are specially designed for answering edge queries}. We set the input parameter $\ell_b$ of \bwd according to Eq.~\eqref{eq:ell}.

For all randomized algorithms (i.e., \fwd, \algo, $\mathtt{TP}$, $\mathtt{TPC}$, and $\mathtt{MC2}$), we set failure probability $\delta=0.01$. \revise{Since the input parameters $\beta_i$ of $\mathtt{TPC}$ are unknown, we adopt the heuristic settings suggested in \cite{pankdd21}. Note that these settings do not ensure the returned value of $\mathtt{TPC}$ is an $\epsilon$-approximate PER.} Unless otherwise specified, the number of batches $\tau$ in \fwd and \algo is set to $5$. We report the average query times (measured in wall-clock time) and the actual average absolute error of each algorithm on each dataset with various $\epsilon$ in $\{0.01,0.02,0.05,0.1,0.2,0.5\}$. We exclude a method if it fails to report the result for each query within one day.

\begin{figure}[!t]
\centering
\begin{small}
\begin{tikzpicture}[every mark/.append style={mark size=3pt}]
    \begin{customlegend}[legend columns=7,
        legend entries={\algo,\fwd,\bwd,$\mathtt{TP}$,$\mathtt{TPC}$,$\mathtt{RP}$,$\mathtt{EXACT}$},
        legend style={at={(0.5,1.05)},anchor=north,draw=none,font=\small,column sep=0.2cm}]
    \addlegendimage{line width=0.2mm,mark=square,color=my_blue}
    \addlegendimage{line width=0.2mm,mark=o,color=my_cyan}
    \addlegendimage{line width=0.2mm,mark=triangle,color=my_violet}
    \addlegendimage{line width=0.2mm,mark=pentagon,color=myblue}
        \addlegendimage{line width=0.2mm,mark=x,color=my_purple}
        \addlegendimage{line width=0.2mm,mark=10-pointed star,color=my_yellow}
        \addlegendimage{line width=0.2mm,mark=diamond,color=my_pp}
    \end{customlegend}
\end{tikzpicture}
\\[-\lineskip]

\vspace{-3mm}
\hspace{-1mm}
\subfloat[Facebook]{
\begin{tikzpicture}[scale=1,every mark/.append style={mark size=2.5pt}]
    \begin{axis}[
        height=\columnwidth/3.4,
        width=\columnwidth/2.6,
        ylabel={\em running time} (ms),
        xmin=0.5, xmax=6.5,
        ymin=1, ymax=100000000,
        xtick={1,2,3,4,5,6},
        xticklabel style = {font=\footnotesize},
        yticklabel style = {font=\footnotesize},
        xticklabels={0.5,0.2,0.1,0.05,0.02,0.01},
        ymode=log,
        ytick={1,10,100,1000,10000,100000,1000000,10000000,100000000},
        log basis y={10},
        % every axis y label/.style={font=\small,at={(current axis.north)},right=0mm,above=0mm},
        % ylabel near ticks,
        every axis y label/.style={font=\small,at={{(0.28,1.0)}},right=10mm,above=0mm},
        y label style = {font=\footnotesize},
        xlabel near ticks,
        x label style = {font=\tiny},
        legend style={fill=none,font=\tiny,at={(0.02,0.99)},anchor=north west,draw=none},
        x tick label style={rotate=45,anchor=east},
    ]
    \addplot[line width=0.2mm,mark=square,color=my_blue] %GEER 
        plot coordinates {
(1,	0.658922	)
(2,	5.87016	)
(3,	30.379	)
(4,	60.1643	)
(5,	152.78	)
(6, 280.519   )
    };

    \addplot[line width=0.2mm,mark=o,color=my_cyan]  %AMC
        plot coordinates {
(1,	4.56631	)
(2,	35.753	)
(3,	150.99	)
(4,	569.064	)
(5,	4531.25	)
(6, 19213.5 )
    };

    \addplot[line width=0.2mm,mark=pentagon,color=myblue]  %TP
        plot coordinates {
(1,	233026	)
(2,	3.21754e+06	)
(3,	2.23066e+07	)
    };
    
    \addplot[line width=0.2mm,mark=triangle,color=my_violet]  %SMM
        plot coordinates {
(1,	392.83	)
(2,	636.362	)
(3,	837.567	)
(4,	1055.03	)
(5,	1311.88	)
(6, 1521.94   )
    };

    \addplot[line width=0.2mm,mark=x,color=my_purple]  %TPC
        plot coordinates {
(1,	50041.2	)
(2,	545628	)
(3,	3.1553e+06	)
(4,	1.67161e+07	)
(5,	0	)
(6, 0   )
    };

    \addplot[line width=0.2mm,mark=diamond,color=my_pp]  %MPI
        plot coordinates {
(1,	829291.119194	)
(2,	829291.119194	)
(3,	829291.119194	)
(4,	829291.119194	)
(5,	829291.119194	)
(6, 829291.119194   )
    };

    \addplot[line width=0.2mm,mark=10-pointed star,color=my_yellow]  %RP
        plot coordinates {
(1,	6523.119354	)
(2,	33789.212799	)
(3,	149878.105545	)
(4,	639960.897636	)
(5,	0	)
(6, 0   )
    };
    % \legend{\algo,$\mathtt{TP}$,\fwd}
    \end{axis}
\end{tikzpicture}\hspace{2mm}\label{fig:time-fb}%
}%
\hspace{0mm}\subfloat[DBLP]{
\begin{tikzpicture}[scale=1,every mark/.append style={mark size=2.5pt}]
    \begin{axis}[
        height=\columnwidth/3.4,
        width=\columnwidth/2.6,
        ylabel={\em running time} (ms),
        xmin=0.5, xmax=6.5,
        ymin=10, ymax=100000000,
        xtick={1,2,3,4,5,6},
        xticklabel style = {font=\footnotesize},
        yticklabel style = {font=\footnotesize},
        xticklabels={0.5,0.2,0.1,0.05,0.02,0.01},
        ytick={10,100,1000,10000,100000,1000000,10000000,100000000},
        ymode=log,
        log basis y={10},
        % every axis y label/.style={font=\small,at={(current axis.north)},right=0mm,above=0mm},
        % ylabel near ticks,
        every axis y label/.style={font=\small,at={{(0.28,1.0)}},right=10mm,above=0mm},
        y label style = {font=\footnotesize},
        xlabel near ticks,
        x label style = {font=\tiny},
        legend style={fill=none,font=\tiny,at={(0.02,0.99)},anchor=north west,draw=none},
        x tick label style={rotate=45,anchor=east},
    ]
    \addplot[line width=0.2mm,mark=square,color=my_blue]  
        plot coordinates {
(1,	15.3166	)
(2,	67.8221	)
(3,	195.061	)
(4,	1155.26	)
(5,	4666.67	)
(6, 8014.29   )
    };

    \addplot[line width=0.2mm,mark=o,color=my_cyan]  
        plot coordinates {
(1,	39.2685	)
(2,	434.14	)
(3,	1438.96	)
(4,	8965.35	)
(5,	65537.1	)
(6, 347505   )
    };
    
    \addplot[line width=0.2mm,mark=pentagon,color=myblue]  
        plot coordinates {
(1,	630383	)
(2,	1.00033e+07	)
    };
    
    \addplot[line width=0.2mm,mark=triangle,color=my_violet]  %SMM
        plot coordinates {
(1,	40157.3	)
(2,	55433.4	)
(3,	63089	)
(4,	81107.8	)
(5,	88736	)
(6, 102343   )
    };

    \addplot[line width=0.2mm,mark=x,color=my_purple]  %TPC
        plot coordinates {
(1,	44898.3	)
(2,	504279	)
(3,	2.90708e+06	)
(4,	1.55338e+07	)
(5,	0	)
(6, 0   )
    };
    
    \addplot[line width=0.2mm,mark=diamond,color=my_pp]  %MPI
        plot coordinates {
(1,	0	)
(2,	0	)
(3,	0	)
(4,	0	)
(5,	0	)
(6, 0   )
    };

    \addplot[line width=0.2mm,mark=10-pointed star,color=my_yellow]  %RP
        plot coordinates {
(1,	289102.994347	)
(2,	1400361.632156	)
(3,	0	)
(4,	0	)
(5,	0	)
(6, 0   )
    };

    % \legend{\algo,$\mathtt{TP}$,\fwd}
    \end{axis}
\end{tikzpicture}\hspace{1mm}\label{fig:time-dblp}%
}%
\vspace{-3mm}
\hspace{-1mm}\subfloat[Youtube]{
\begin{tikzpicture}[scale=1,every mark/.append style={mark size=2.5pt}]
    \begin{axis}[
        height=\columnwidth/3.4,
        width=\columnwidth/2.6,
        ylabel={\em running time} (ms),
        xmin=0.5, xmax=6.5,
        ymin=1, ymax=100000000,
        xtick={1,2,3,4,5,6},
        xticklabel style = {font=\footnotesize},
        yticklabel style = {font=\footnotesize},
        xticklabels={0.5,0.2,0.1,0.05,0.02,0.01},
        ymode=log,
        ytick={1,10,100,1000,10000,100000,1000000,10000000,100000000},
        log basis y={10},
        % every axis y label/.style={font=\small,at={(current axis.north)},right=0mm,above=0mm},
        % ylabel near ticks,
        every axis y label/.style={font=\small,at={{(0.28,1.0)}},right=10mm,above=0mm},
        y label style = {font=\footnotesize},
        xlabel near ticks,
        x label style = {font=\tiny},
        legend style={fill=none,font=\tiny,at={(0.02,0.99)},anchor=north west,draw=none},
        x tick label style={rotate=45,anchor=east},
    ]
    \addplot[line width=0.2mm,mark=square,color=my_blue]  
        plot coordinates {
(1,	4.0431	)
(2,	21.0494	)
(3,	53.1979	)
(4,	288.417	)
(5,	1609.34	)
(6, 3549.26   )
    };

    \addplot[line width=0.2mm,mark=o,color=my_cyan]  
        plot coordinates {
(1,	16.2291	)
(2,	70.8329	)
(3,	396.569	)
(4,	2391.36	)
(5,	24052.7	)
(6, 133144   )
    };

    \addplot[line width=0.2mm,mark=pentagon,color=myblue]  
        plot coordinates {
(1,	240293	)
(2,	3.62576e+06	)
(3,	2.36424e+07	)
    };
    
    \addplot[line width=0.2mm,mark=triangle,color=my_violet]  %SMM
        plot coordinates {
(1,	26122.3	)
(2,	34288.7	)
(3,	40878.7	)
(4,	47314.1	)
(5,	56301.2	)
(6, 62654.7   )
    };

    \addplot[line width=0.2mm,mark=x,color=my_purple]  %TPC
        plot coordinates {
(1,	49517.8	)
(2,	532873	)
(3,	3.06365e+06	)
(4,	1.64146e+07	)
(5,	0	)
(6, 0   )
    };

    \addplot[line width=0.2mm,mark=diamond,color=my_pp]  %MPI
        plot coordinates {
(1,	0	)
(2,	0	)
(3,	0	)
(4,	0	)
(5,	0	)
(6, 0   )
    };

    \addplot[line width=0.2mm,mark=10-pointed star,color=my_yellow]  %RP
        plot coordinates {
(1,	1778836.291885	)
(2,	23418889.015961	)
(3,	0	)
(4,	0	)
(5,	0	)
(6, 0   )
    };

    % \legend{\algo,$\mathtt{TP}$,\fwd}
    \end{axis}
\end{tikzpicture}\hspace{2mm}\label{fig:time-ytb}%
}%

\subfloat[Orkut]{
\begin{tikzpicture}[scale=1,every mark/.append style={mark size=2.5pt}]
    \begin{axis}[
        height=\columnwidth/3.4,
        width=\columnwidth/2.6,
        ylabel={\em running time} (ms),
        xmin=0.5, xmax=6.5,
        ymin=1, ymax=100000000,
        xtick={1,2,3,4,5,6},
        xticklabel style = {font=\footnotesize},
        yticklabel style = {font=\footnotesize},
        xticklabels={0.5,0.2,0.1,0.05,0.02,0.01},
        ytick={1,10,100,1000,10000,100000,1000000,10000000,100000000},
        ymode=log,
        log basis y={10},
        % every axis y label/.style={font=\small,at={(current axis.north)},right=0mm,above=0mm},
        % ylabel near ticks,
        every axis y label/.style={font=\small,at={{(0.28,1.0)}},right=10mm,above=0mm},
        y label style = {font=\footnotesize},
        xlabel near ticks,
        x label style = {font=\tiny},
        legend style={fill=none,font=\tiny,at={(1.07, 0.7)},anchor=north west,draw=none},
        x tick label style={rotate=45,anchor=east},
    ]
    \addplot[line width=0.2mm,mark=square,color=my_blue]  
        plot coordinates {
(1,	1.34556	)
(2,	6.07452	)
(3,	47.9349	)
(4,	304.265	)
(5,	3158.24	)
(6, 8991.46   )
    };

    \addplot[line width=0.2mm,mark=o,color=my_cyan]  
        plot coordinates {
(1,	4.19	)
(2,	60.2425	)
(3,	279.972	)
(4,	1625.1	)
(5,	12179.5	)
(6, 54218.6   )
    };

    \addplot[line width=0.2mm,mark=pentagon,color=myblue]  
        plot coordinates {
(1,	2.8299e+06	)
(2,	3.97027e+07	)
    };

    \addplot[line width=0.2mm,mark=triangle,color=my_violet]  %SMM
        plot coordinates {
(1,	2.02927e+06	)
(2,	4.5298e+06	)
(3,	7.27517e+06	)
(4,	9.94211e+06	)
(5,	1.36905e+07	)
(6, 1.73915e+07   )
    };

    \addplot[line width=0.2mm,mark=x,color=my_purple]  %TPC
        plot coordinates {
(1,	40032.5	)
(2,	453136	)
(3,	2.67492e+06	)
(4,	1.38667e+07	)
(5,	0	)
(6, 0   )
    };

    \addplot[line width=0.2mm,mark=diamond,color=my_pp]  %MPI
        plot coordinates {
(1,	0	)
(2,	0	)
(3,	0	)
(4,	0	)
(5,	0	)
(6, 0   )
    };

    \addplot[line width=0.2mm,mark=10-pointed star,color=my_yellow]  %RP
        plot coordinates {
(1,	0	)
(2,	0	)
(3,	0	)
(4,	0	)
(5,	0	)
(6, 0   )
    };

    % \legend{\algo,$\mathtt{TP}$,\fwd}
    \end{axis}
\end{tikzpicture}\hspace{1mm}\label{fig:time-ok}%
}%
\vspace{-3mm}
\hspace{-1mm}\subfloat[LiveJournal]{
\begin{tikzpicture}[scale=1,every mark/.append style={mark size=2.5pt}]
    \begin{axis}[
        height=\columnwidth/3.4,
        width=\columnwidth/2.6,
        ylabel={\em running time} (ms),
        xmin=0.5, xmax=6.5,
        ymin=10, ymax=100000000,
        xtick={1,2,3,4,5,6},
        xticklabel style = {font=\footnotesize},
        yticklabel style = {font=\footnotesize},
        xticklabels={0.5,0.2,0.1,0.05,0.02,0.01},
        ytick={10,100,1000,10000,100000,1000000,10000000,100000000},
        ymode=log,
        log basis y={10},
        % every axis y label/.style={font=\small,at={(current axis.north)},right=0mm,above=0mm},
        % ylabel near ticks,
        every axis y label/.style={font=\small,at={{(0.28,1.0)}},right=10mm,above=0mm},
        y label style = {font=\footnotesize},
        xlabel near ticks,
        x label style = {font=\tiny},
        legend style={fill=none,font=\tiny,at={(0.02,0.99)},anchor=north west,draw=none},
        x tick label style={rotate=45,anchor=east},
    ]
    \addplot[line width=0.2mm,mark=square,color=my_blue]  
        plot coordinates {
(1,	26.6665	)
(2,	95.958	)
(3,	441.08	)
(4,	1201.82	)
(5,	10033.9	)
(6, 27367.6   )
    };

    \addplot[line width=0.2mm,mark=o,color=my_cyan]  
        plot coordinates {
(1,	91.1405	)
(2,	907.701	)
(3,	3018.06	)
(4,	12860.5	)
(5,	114951	)
(6, 531865   )
    };

    \addplot[line width=0.2mm,mark=pentagon,color=myblue]  
        plot coordinates {
(1,	1.71965e+06	)
(2,	2.32758e+07	)
    };

    \addplot[line width=0.2mm,mark=triangle,color=my_violet]  %SMM
        plot coordinates {
(1,	1.0636e+06	)
(2,	1.45992e+06	)
(3,	1.8501e+06	)
(4,	2.21778e+06	)
(5,	2.65892e+06	)
(6, 3.46289e+06   )
    };

    \addplot[line width=0.2mm,mark=x,color=my_purple]  %TPC
        plot coordinates {
(1,	55453.3	)
(2,	588374	)
(3,	3.27629e+06	)
(4,	1.74856e+07	)
(5,	0	)
(6, 0   )
    };

    \addplot[line width=0.2mm,mark=diamond,color=my_pp]  %MPI
        plot coordinates {
(1,	0	)
(2,	0	)
(3,	0	)
(4,	0	)
(5,	0	)
(6, 0   )
    };

    \addplot[line width=0.2mm,mark=10-pointed star,color=my_yellow]  %RP
        plot coordinates {
(1,	0	)
(2,	0	)
(3,	0	)
(4,	0	)
(5,	0	)
(6, 0   )
    };

    % \legend{\algo,$\mathtt{TP}$,\fwd}
    \end{axis}
\end{tikzpicture}\hspace{2mm}\label{fig:time-lj}%
}%
\hspace{0mm}\subfloat[Friendster]{
\begin{tikzpicture}[scale=1,every mark/.append style={mark size=2.5pt}]
    \begin{axis}[
        height=\columnwidth/3.4,
        width=\columnwidth/2.6,
        ylabel={\em running time} (ms),
        xmin=0.5, xmax=6.5,
        ymin=100, ymax=100000000,
        xtick={1,2,3,4,5,6},
        xticklabel style = {font=\footnotesize},
        yticklabel style = {font=\footnotesize},
        xticklabels={0.5,0.2,0.1,0.05,0.02,0.01},
        ytick={100,1000,10000,100000,1000000,10000000,100000000},
        ymode=log,
        log basis y={10},
        % every axis y label/.style={font=\small,at={(current axis.north)},right=0mm,above=0mm},
        % ylabel near ticks,
        every axis y label/.style={font=\small,at={{(0.28,1.0)}},right=10mm,above=0mm},
        y label style = {font=\footnotesize},
        xlabel near ticks,
        x label style = {font=\tiny},
        legend style={fill=none,font=\tiny,at={(1.07, 0.7)},anchor=north west,draw=none},
        x tick label style={rotate=45,anchor=east},
    ]
    \addplot[line width=0.2mm,mark=square,color=my_blue]  % GEER
        plot coordinates {
(1,	141.406	)
(2,	147.172	)
(3,	179.464	)
(4,	416.672	)
(5,	1324.31	)
(6, 3243.29   )
    };

    \addplot[line width=0.2mm,mark=o,color=my_cyan]  %AMC
        plot coordinates {
(1,	195.971	)
(2,	534.244	)
(3,	1602.96	)
(4,	3239.98	)
(5,	23036.6	)
(6, 123857   )
    };

    \addplot[line width=0.2mm,mark=pentagon,color=myblue]  % TP
        plot coordinates {
(1,	4.12814e+07	)
(2,	0	)
(3,	0	)
(4,	0	)
(5,	0	)
(6, 0   )
    };
    
    \addplot[line width=0.2mm,mark=triangle,color=my_violet]  %SMM
        plot coordinates {
(1,	1	)
(2,	0	)
(3,	0	)
(4,	0	)
(5,	0	)
(6, 0   )
    };

    \addplot[line width=0.2mm,mark=x,color=my_purple]  %TPC
        plot coordinates {
(1,	107423.2	)
(2,	900665	)
(3,	4.75019e+06	)
(4,	2.67835e+07	)
(5,	0	)
(6, 0   )
    };

    \addplot[line width=0.2mm,mark=diamond,color=my_pp]  %MPI
        plot coordinates {
(1,	0	)
(2,	0	)
(3,	0	)
(4,	0	)
(5,	0	)
(6, 0   )
    };

    \addplot[line width=0.2mm,mark=10-pointed star,color=my_yellow]  %RP
        plot coordinates {
(1,	0	)
(2,	0	)
(3,	0	)
(4,	0	)
(5,	0	)
(6, 0   )
    };

    % \legend{\algo,$\mathtt{TP}$,\fwd}
    \end{axis}
\end{tikzpicture}\label{fig:time-fri}%
}%
\end{small}
\vspace{0mm}
\caption{Running time vs. $\epsilon$ for random queries.} \label{fig:time}
%\vspace{-2mm}
\end{figure}
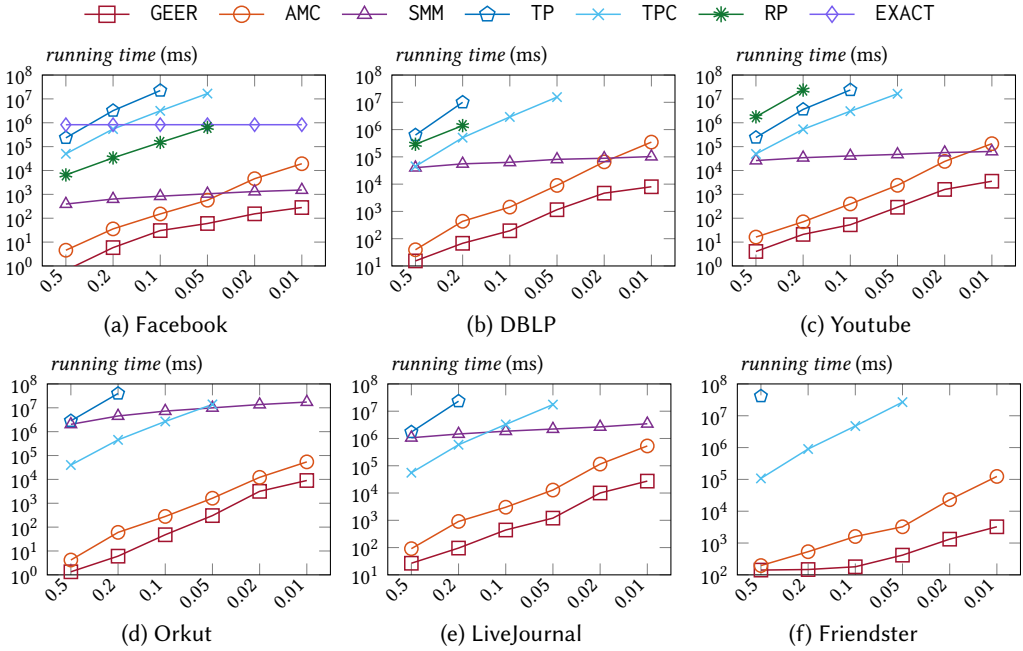

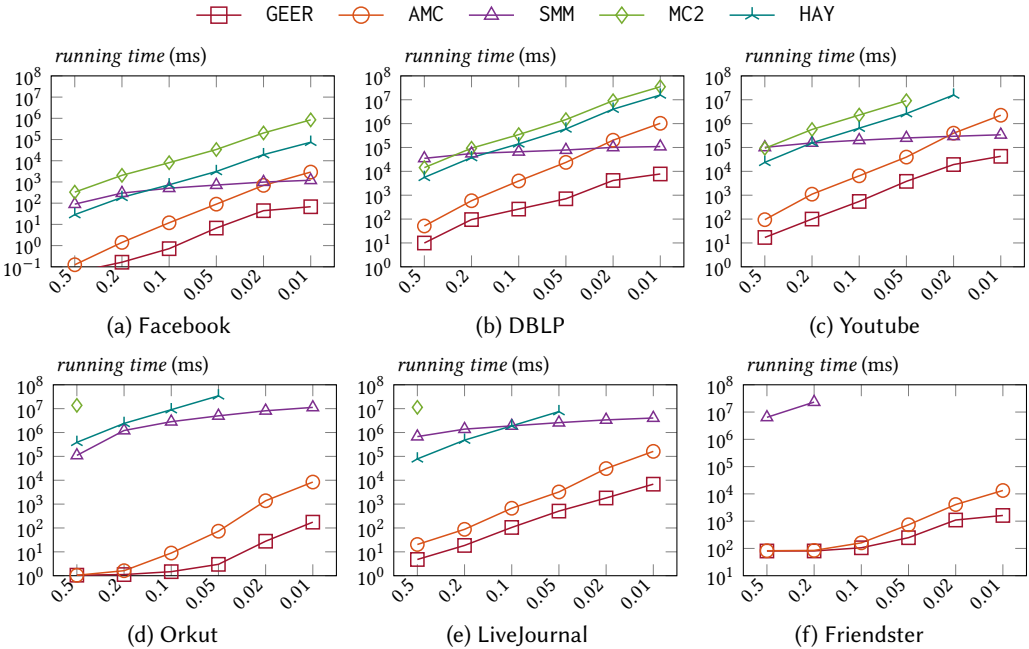
\begin{figure}[!t]
\centering
\begin{small}
\begin{tikzpicture}[every mark/.append style={mark size=3pt}]
    \begin{customlegend}[legend columns=5,
        legend entries={\algo,\fwd,\bwd,$\mathtt{MC2}$,$\mathtt{HAY}$},
        legend style={at={(0.5,1.05)},anchor=north,draw=none,font=\small,column sep=0.25cm}]
    \addlegendimage{line width=0.2mm,mark=square,color=my_blue}
    \addlegendimage{line width=0.2mm,mark=o,color=my_cyan}
    \addlegendimage{line width=0.2mm,mark=triangle,color=my_violet}
    \addlegendimage{line width=0.2mm,mark=diamond,color=my_teal}
    \addlegendimage{line width=0.2mm,mark=Mercedes star,color=my_green}
    % \addlegendimage{line width=0.2mm,mark=pentagon,color=myblue}
    % \addlegendimage{line width=0.2mm,mark=x,color=my_purple}
    \end{customlegend}
\end{tikzpicture}
\\[-\lineskip]
\vspace{-3mm}

\hspace{-1mm}\subfloat[Facebook]{
\begin{tikzpicture}[scale=1,every mark/.append style={mark size=2.5pt}]
    \begin{axis}[
        height=\columnwidth/3.4,
        width=\columnwidth/2.6,
        ylabel={\em running time} (ms),
        xmin=0.5, xmax=6.5,
        ymin=0.1, ymax=100000000,
        xtick={1,2,3,4,5,6},
        xticklabel style = {font=\footnotesize},
        yticklabel style = {font=\footnotesize},
        xticklabels={0.5,0.2,0.1,0.05,0.02,0.01},
        ymode=log,
        ytick={0.1,1,10,100,1000,10000,100000,1000000,10000000,100000000},
        log basis y={10},
        % every axis y label/.style={font=\small,at={(current axis.north)},right=0mm,above=0mm},
        % ylabel near ticks,
        every axis y label/.style={font=\small,at={{(0.28,1.0)}},right=10mm,above=0mm},
        y label style = {font=\footnotesize},
        xlabel near ticks,
        x label style = {font=\tiny},
        legend style={fill=none,font=\tiny,at={(0.02,0.99)},anchor=north west,draw=none},
        x tick label style={rotate=45,anchor=east},
    ]
    \addplot[line width=0.2mm,mark=square,color=my_blue] %GEER 
        plot coordinates {
(1,	0.0549168	)
(2,	0.164552	)
(3,	0.716495	)
(4,	6.70665	)
(5,	44.2437	)
(6, 67.9332   )
    };

    \addplot[line width=0.2mm,mark=o,color=my_cyan]  %AMC
        plot coordinates {
(1,	0.125097	)
(2,	1.41088	)
(3,	11.9051	)
(4,	89.8564	)
(5,	689.921	)
(6, 2923.91   )
    };
    
    \addplot[line width=0.2mm,mark=triangle,color=my_violet]  %SMM
        plot coordinates {
(1,	90.0702	)
(2,	298.504	)
(3,	507.836	)
(4,	712.399	)
(5,	987.068	)
(6, 1202.26   )
    };

    \addplot[line width=0.2mm,mark=diamond,color=my_teal]  %MC2
        plot coordinates {
(1,	330.884	)
(2,	2051.13	)
(3,	8102.01	)
(4,	33812.3	)
(5,	203738	)
(6, 846086   )
    };

    \addplot[line width=0.2mm,mark=Mercedes star,color=my_green]  %HAY
        plot coordinates {
(1,	28.8672	)
(2,	187.628	)
(3,	736.282	)
(4,	3113.86	)
(5,	19394.6	)
(6, 74463.9   )
    };
    % \legend{\algo,$\mathtt{TP}$,\fwd}
    \end{axis}
\end{tikzpicture}\hspace{2mm}\label{fig:time-fb-edge}%
}%
\hspace{0mm}\subfloat[DBLP]{
\begin{tikzpicture}[scale=1,every mark/.append style={mark size=2.5pt}]
    \begin{axis}[
        height=\columnwidth/3.4,
        width=\columnwidth/2.6,
        ylabel={\em running time} (ms),
        xmin=0.5, xmax=6.5,
        ymin=1, ymax=100000000,
        xtick={1,2,3,4,5,6},
        xticklabel style = {font=\footnotesize},
        yticklabel style = {font=\footnotesize},
        xticklabels={0.5,0.2,0.1,0.05,0.02,0.01},
        ytick={1,10,100,1000,10000,100000,1000000,10000000,100000000},
        ymode=log,
        log basis y={10},
        % every axis y label/.style={font=\small,at={(current axis.north)},right=0mm,above=0mm},
        % ylabel near ticks,
        every axis y label/.style={font=\small,at={{(0.28,1.0)}},right=10mm,above=0mm},
        y label style = {font=\footnotesize},
        xlabel near ticks,
        x label style = {font=\tiny},
        legend style={fill=none,font=\tiny,at={(0.02,0.99)},anchor=north west,draw=none},
        x tick label style={rotate=45,anchor=east},
    ]
    \addplot[line width=0.2mm,mark=square,color=my_blue] %GEER 
        plot coordinates {
(1,	9.93152	)
(2,	95.374	)
(3,	262.461	)
(4,	708.825	)
(5,	4143.66	)
(6, 7757.58   )
    };

    \addplot[line width=0.2mm,mark=o,color=my_cyan]  %AMC
        plot coordinates {
(1,	51.2569	)
(2,	587.517 	)
(3,	3957.14	)
(4,	23705.9	)
(5,	201658	)
(6, 1.02908e+06   )
    };
    
    \addplot[line width=0.2mm,mark=triangle,color=my_violet]  %SMM
        plot coordinates {
(1,	35048.5	)
(2,	55701.1	)
(3,	66085.4	)
(4,	78274.2	)
(5,	100945	)
(6, 109396   )
    };

    \addplot[line width=0.2mm,mark=diamond,color=my_teal]  %MC2
        plot coordinates {
(1,	14339.5	)
(2,	93356.2	)
(3,	346110	)
(4,	1.47176e+06	)
(5,	9.09159e+06	)
(6, 3.50087e+07   )
    };

    \addplot[line width=0.2mm,mark=Mercedes star,color=my_green]  %HAY
        plot coordinates {
(1,	5590.22	)
(2,	37695	)
(3,	141670	)
(4,	597815	)
(5,	4.02573e+06	)
(6, 1.57987e+07   )
    };
    % \legend{\algo,$\mathtt{TP}$,\fwd}
    \end{axis}
\end{tikzpicture}\hspace{1mm}\label{fig:time-dblp-edge}%
}%
\vspace{-3mm}
\hspace{-1mm}\subfloat[Youtube]{
\begin{tikzpicture}[scale=1,every mark/.append style={mark size=2.5pt}]
    \begin{axis}[
        height=\columnwidth/3.4,
        width=\columnwidth/2.6,
        ylabel={\em running time} (ms),
        xmin=0.5, xmax=6.5,
        ymin=1, ymax=100000000,
        xtick={1,2,3,4,5,6},
        xticklabel style = {font=\footnotesize},
        yticklabel style = {font=\footnotesize},
        xticklabels={0.5,0.2,0.1,0.05,0.02,0.01},
        ymode=log,
        ytick={1,10,100,1000,10000,100000,1000000,10000000,100000000},
        log basis y={10},
        % every axis y label/.style={font=\small,at={(current axis.north)},right=0mm,above=0mm},
        % ylabel near ticks,
        every axis y label/.style={font=\small,at={{(0.28,1.0)}},right=10mm,above=0mm},
        y label style = {font=\footnotesize},
        xlabel near ticks,
        x label style = {font=\tiny},
        legend style={fill=none,font=\tiny,at={(0.02,0.99)},anchor=north west,draw=none},
        x tick label style={rotate=45,anchor=east},
    ]
    \addplot[line width=0.2mm,mark=square,color=my_blue] %GEER 
        plot coordinates {
(1,	16.9777	)
(2,	100.214	)
(3,	547.62	)
(4,	3806.13	)
(5,	19282.6	)
(6, 42661.6   )
    };

    \addplot[line width=0.2mm,mark=o,color=my_cyan]  %AMC
        plot coordinates {
(1,	94.4665	)
(2,	1104.99	)
(3,	6495.3	)
(4,	39546.1	)
(5,	397252	)
(6, 2.24132e+06   )
    };
    
    \addplot[line width=0.2mm,mark=triangle,color=my_violet]  %SMM
        plot coordinates {
(1,	101053	)
(2,	155823	)
(3,	200897	)
(4,	250165	)
(5,	295071	)
(6, 340352   )
    };

    \addplot[line width=0.2mm,mark=diamond,color=my_teal]  %MC2
        plot coordinates {
(1,	92372.3	)
(2,	567357	)
(3,	2.27937e+06	)
(4,	9.1224e+06	)
(5,	0	)
(6, 0   )
    };

    \addplot[line width=0.2mm,mark=Mercedes star,color=my_green]  %HAY
        plot coordinates {
(1,	23847.7	)
(2,	156469	)
(3,	644243	)
(4,	2.62309e+06	)
(5,	1.5776e+07	)
(6, 0   )
    };
    % \legend{\algo,$\mathtt{TP}$,\fwd}
    \end{axis}
\end{tikzpicture}\hspace{2mm}\label{fig:time-ytb-edge}%
}%

\subfloat[Orkut]{
\begin{tikzpicture}[scale=1,every mark/.append style={mark size=2.5pt}]
    \begin{axis}[
        height=\columnwidth/3.4,
        width=\columnwidth/2.6,
        ylabel={\em running time} (ms),
        xmin=0.5, xmax=6.5,
        ymin=1, ymax=100000000,
        xtick={1,2,3,4,5,6},
        xticklabel style = {font=\footnotesize},
        yticklabel style = {font=\footnotesize},
        xticklabels={0.5,0.2,0.1,0.05,0.02,0.01},
        ytick={1,10,100,1000,10000,100000,1000000,10000000,100000000},
        ymode=log,
        log basis y={10},
        % every axis y label/.style={font=\small,at={(current axis.north)},right=0mm,above=0mm},
        % ylabel near ticks,
        every axis y label/.style={font=\small,at={{(0.28,1.0)}},right=10mm,above=0mm},
        y label style = {font=\footnotesize},
        xlabel near ticks,
        x label style = {font=\tiny},
        legend style={fill=none,font=\tiny,at={(1.07, 0.7)},anchor=north west,draw=none},
        x tick label style={rotate=45,anchor=east},
    ]
    \addplot[line width=0.2mm,mark=square,color=my_blue] %GEER 
        plot coordinates {
(1,	1.05344	)
(2,	1.1158 	)
(3,	1.47813	)
(4,	2.99757	)
(5,	27.8744	)
(6, 175.525   )
    };

    \addplot[line width=0.2mm,mark=o,color=my_cyan]  %AMC
        plot coordinates {
(1,	1.04184	)
(2,	1.62088	)
(3,	8.90405	)
(4,	73.1761	)
(5,	1382.86	)
(6, 8401.02   )
    };
    
    \addplot[line width=0.2mm,mark=triangle,color=my_violet]  %SMM
        plot coordinates {
(1,	109902	)
(2,	1.21411e+06	)
(3,	2.84363e+06	)
(4,	4.99321e+06	)
(5,	8.14573e+06	)
(6, 1.1285e+07   )
    };

    \addplot[line width=0.2mm,mark=diamond,color=my_teal]  %MC2
        plot coordinates {
(1,	1.37417e+07	)
(2,	0	)
(3,	0	)
(4,	0	)
(5,	0	)
(6, 0   )
    };

    \addplot[line width=0.2mm,mark=10-pointed star,color=my_red]  %naive SMM
        plot coordinates {
(1,	0	)
(2,	0	)
(3,	0	)
(4,	0	)
(5,	0	)
(6, 0   )
    };

    \addplot[line width=0.2mm,mark=Mercedes star,color=my_green]  %HAY
        plot coordinates {
(1,	387054	)
(2,	2.40943e+06	)
(3,	9.03337e+06	)
(4,	3.43541e+07	)
    };
    % \legend{\algo,$\mathtt{TP}$,\fwd}
    \end{axis}
\end{tikzpicture}\hspace{1mm}\label{fig:time-ok-edge}%
}%
\vspace{-3mm}
\hspace{-1mm}\subfloat[LiveJournal]{
\begin{tikzpicture}[scale=1,every mark/.append style={mark size=2.5pt}]
    \begin{axis}[
        height=\columnwidth/3.4,
        width=\columnwidth/2.6,
        ylabel={\em running time} (ms),
        xmin=0.5, xmax=6.5,
        ymin=1, ymax=100000000,
        xtick={1,2,3,4,5,6},
        xticklabel style = {font=\footnotesize},
        yticklabel style = {font=\footnotesize},
        xticklabels={0.5,0.2,0.1,0.05,0.02,0.01},
        ytick={1,10,100,1000,10000,100000,1000000,10000000,100000000},
        ymode=log,
        log basis y={10},
        % every axis y label/.style={font=\small,at={(current axis.north)},right=0mm,above=0mm},
        % ylabel near ticks,
        every axis y label/.style={font=\small,at={{(0.28,1.0)}},right=10mm,above=0mm},
        y label style = {font=\footnotesize},
        xlabel near ticks,
        x label style = {font=\tiny},
        legend style={fill=none,font=\tiny,at={(0.02,0.99)},anchor=north west,draw=none},
        x tick label style={rotate=45,anchor=east},
    ]
    \addplot[line width=0.2mm,mark=square,color=my_blue] %GEER 
        plot coordinates {
(1,	4.78905	)
(2,	18.5469	)
(3,	105.172	)
(4,	509.644	)
(5,	1806.82	)
(6, 6845.89   )
    };

    \addplot[line width=0.2mm,mark=o,color=my_cyan]  %AMC
        plot coordinates {
(1,	20.3413	)
(2,	87.6857	)
(3,	670.292	)
(4,	3231.07	)
(5,	30827.2	)
(6, 162007   )
    };
    
    \addplot[line width=0.2mm,mark=triangle,color=my_violet]  %SMM
        plot coordinates {
(1,	675063	)
(2,	1.38579e+06	)
(3,	1.8959e+06	)
(4,	2.59014e+06	)
(5,	3.37852e+06	)
(6, 4.05086e+06   )
    };

    \addplot[line width=0.2mm,mark=diamond,color=my_teal]  %MC2
        plot coordinates {
(1,	1.13577e+07	)
(2,	0	)
(3,	0	)
(4,	0	)
(5,	0	)
(6, 0   )
    };

    \addplot[line width=0.2mm,mark=Mercedes star,color=my_green]  %HAY
        plot coordinates {
(1,	78473.7	)
(2,	478927	)
(3,	1.89333e+06	)
(4,	7.52794e+06	)
(5,	0	)
(6, 0   )
    };
    % \legend{\algo,$\mathtt{TP}$,\fwd}
    \end{axis}
\end{tikzpicture}\hspace{2mm}\label{fig:time-lj-edge}%
}%
\hspace{0mm}\subfloat[Friendster]{
\begin{tikzpicture}[scale=1,every mark/.append style={mark size=2.5pt}]
    \begin{axis}[
        height=\columnwidth/3.4,
        width=\columnwidth/2.6,
        ylabel={\em running time} (ms),
        xmin=0.5, xmax=6.5,
        ymin=10, ymax=100000000,
        xtick={1,2,3,4,5,6},
        xticklabel style = {font=\footnotesize},
        yticklabel style = {font=\footnotesize},
        xticklabels={0.5,0.2,0.1,0.05,0.02,0.01},
        ytick={10,100,1000,10000,100000,1000000,10000000,100000000},
        ymode=log,
        log basis y={10},
        % every axis y label/.style={font=\small,at={(current axis.north)},right=0mm,above=0mm},
        % ylabel near ticks,
        every axis y label/.style={font=\small,at={{(0.28,1.0)}},right=10mm,above=0mm},
        y label style = {font=\footnotesize},
        xlabel near ticks,
        x label style = {font=\tiny},
        legend style={fill=none,font=\tiny,at={(1.07, 0.7)},anchor=north west,draw=none},
        x tick label style={rotate=45,anchor=east},
    ]
    \addplot[line width=0.2mm,mark=square,color=my_blue] %GEER 
        plot coordinates {
(1,	80.4485	)
(2,	80.7797	)
(3,	105.656	)
(4,	247.6767	)
(5,	1100.174	)
(6, 1617.357   )
    };

    \addplot[line width=0.2mm,mark=o,color=my_cyan]  %AMC
        plot coordinates {
(1,	82.2459	)
(2,	84.4861	)
(3,	160.217	)
(4,	736.445	)
(5,	4063.67	)
(6, 13316.8   )
    };
    
    \addplot[line width=0.2mm,mark=triangle,color=my_violet]  %SMM
        plot coordinates {
(1,	6.40729e+06	)
(2,	2.31961e+07	)
    };

    \addplot[line width=0.2mm,mark=diamond,color=my_teal]  %MC2
        plot coordinates {
(1,	0	)
(2,	0	)
(3,	0	)
(4,	0	)
(5,	0	)
(6, 0   )
    };

    \addplot[line width=0.2mm,mark=Mercedes star,color=my_green]  %HAY
        plot coordinates {
(1,	0	)
(2,	0	)
(3,	0	)
(4,	0	)
(5,	0	)
(6, 0   )
    };
    % \legend{\algo,$\mathtt{TP}$,\fwd}
    \end{axis}
\end{tikzpicture}\label{fig:time-fri-edge}%
}%
\end{small}
\vspace{0mm}
\caption{Running time vs. $\epsilon$ for edge queries.} \label{fig:time-edge}
%\vspace{-2mm}
\end{figure}

\subsection{Query Efficiency}\label{sec:exp-time}
% \subsubsection{\bf Random Queries}
Our first set of experiments compares \fwd and \algo against \bwd , $\mathtt{TP}$, and $\mathtt{TPC}$ on random queries. 
Fig. \ref{fig:time} reports the evaluation results on query efficiency (i.e., average running time) of each method on each dataset when $\epsilon$ is varied from $0.01$ to $0.5$ ($x$-axis). Note that the $y$-axis is in log-scale and the measurement unit for running time is millisecond (ms).
$\mathtt{TP}$, $\mathtt{TPC}$, and \bwd cannot terminate within one day in some cases, and, thus their results are not reported. \revise{Note that $\mathtt{EXACT}$ can only handle the smallest dataset Facebook as it incurs out-of-memory errors on larger datasets due to the colossal space needed for materializing the $n\times n$ matrix pseudo-inverse. Akin to $\mathtt{EXACT}$, $\mathtt{RP}$ runs out of memory on Orkut, LiveJournal, and Friendster datasets, where it requires constructing large dense random matrices.}
On small datasets including Facebook, DBLP, and Youtube, \fwd offers remarkable speedup over $\mathtt{TP}$ and $\mathtt{TPC}$, but is comparable or inferior to \bwd when $\epsilon\le 0.02$, whereas \algo consistently outperforms all competitors often by orders of magnitude. For instance, on Youtube, \algo is up to $37.5\times$ and $6461\times$ faster than \fwd and \bwd, respectively, and consistently achieves more than three orders of magnitude of speedup over \revise{$\mathtt{RP}$}, $\mathtt{TP}$, and $\mathtt{TPC}$.

Another observation we can make from Fig. \ref{fig:time} is that both \fwd and \algo significantly outperform \bwd, $\mathtt{TP}$, and $\mathtt{TPC}$ on all large datasets including Orkut, LiveJournal, and Friendster. 
Notably, on such graphs, the efficiency of \bwd is largely degraded on account of the highly expensive matrix-vector multiplications.
Although \fwd remarkably speeds up $\mathtt{TP}$ and $\mathtt{TPC}$ with the refined $\ell$ and its adaptive sampling scheme, it is still rather costly as it involves numerous random walks and intensive memory patterns on large graphs when $\epsilon$ is small. In comparison, \algo dominates all competitors including \fwd with considerable speedups, i.e., up to $38.2\times$ improvement than \fwd, over $126.5\times$ efficiency gain compared to \bwd, as well as more than three orders of magnitude of speedup over $\mathtt{TP}$, respectively. Particularly, on the billion-edge graph Friendster, given additive error $\epsilon=0.02$, \algo finishes an $\epsilon$-approximate PER query using only $1.3$ seconds and \fwd needs $23$ seconds on average, whereas the rest methods are unable to finish processing within one day.

Next, we experiment with edge queries. We compare our proposed \fwd and \algo with the baseline method \bwd and \revise{two dedicated solutions $\mathtt{HAY}$ \cite{hayashi2016efficient} and $\mathtt{MC2}$ \cite{pankdd21} for processing edge-based PER queries}. \revise{Both $\mathtt{HAY}$ and $\mathtt{MC2}$ rely on simulating a huge number of random walks without explicit length constraints.} As can be seen in Fig. \ref{fig:time-edge}, in all cases, \algo significantly improves over the competitors in terms of query efficiency, e.g., often more than 1000-fold improvements over \bwd, \revise{$\mathtt{HAY}$}, and $\mathtt{MC2}$, and up to $132.7\times$ speed up than \fwd. As for \fwd, it achieves the second best performance in most cases except on small graphs including Facebook, DBLP, and Youtube when $\epsilon\le 0.02$, where \bwd has comparable or superior performance to \fwd. Similar observations can be made from Fig. \ref{fig:time-fb}--\ref{fig:time-ytb}. On small graphs, the vectors in \bwd turn into dense ones more easily, causing slow growth in computational cost when further decreasing $\epsilon$. Conversely, for \fwd, the rising in query cost is more pronounced as $\epsilon$ reduces since the number of random walks needed in \fwd (Eq. \eqref{eq:etas-star}) is inversely proportional to $\epsilon^2$.
% The improvement of the query efficiency is more pronounced

Together these results reveal that our proposed \algo obtains high superiority in empirical efficiency over all competitors especially on large graphs with high average degrees, demonstrating the effectiveness of the greedy integration scheme developed in \algo (Section \ref{sec:algo-over}).

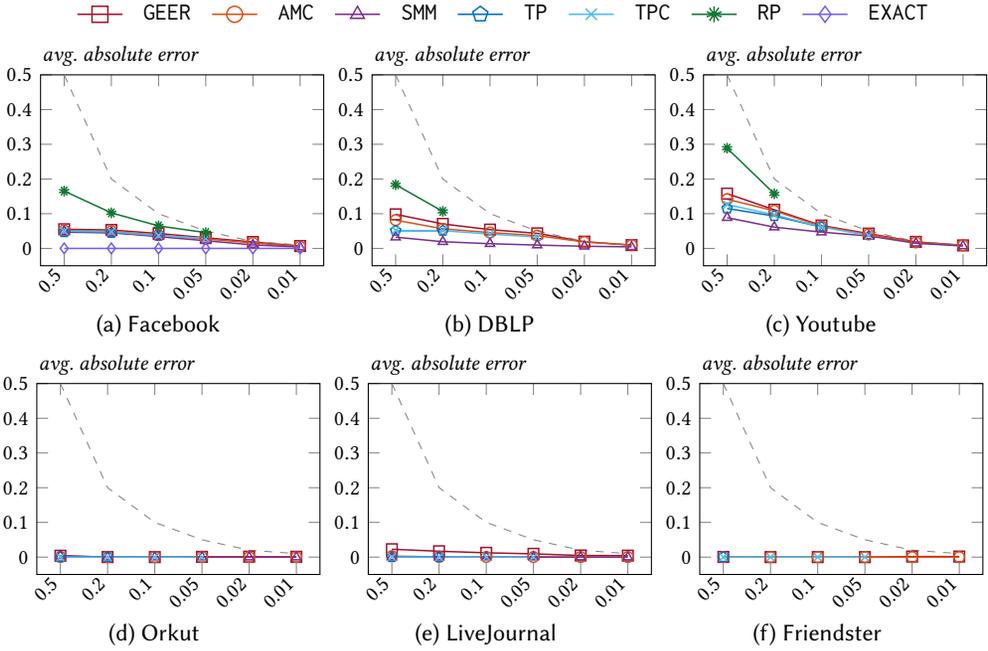
\begin{figure}[!t]
\centering
\begin{small}
\begin{tikzpicture}[every mark/.append style={mark size=3pt}]
    \begin{customlegend}[legend columns=7,
        legend entries={\algo,\fwd,\bwd,$\mathtt{TP}$,$\mathtt{TPC}$,$\mathtt{RP}$,$\mathtt{EXACT}$},
        legend style={at={(0.5,1.05)},anchor=north,draw=none,font=\small,column sep=0.2cm}]
    \addlegendimage{line width=0.2mm,mark=square,color=my_blue}
    \addlegendimage{line width=0.2mm,mark=o,color=my_cyan}
    \addlegendimage{line width=0.2mm,mark=triangle,color=my_violet}
    \addlegendimage{line width=0.2mm,mark=pentagon,color=myblue}
        \addlegendimage{line width=0.2mm,mark=x,color=my_purple}
        \addlegendimage{line width=0.2mm,mark=10-pointed star,color=my_yellow}
        \addlegendimage{line width=0.2mm,mark=diamond,color=my_pp}
    \end{customlegend}
\end{tikzpicture}
\\[-\lineskip]
\vspace{-3mm}

\hspace{-4mm}\subfloat[Facebook]{
\begin{tikzpicture}[scale=1,every mark/.append style={mark size=2pt}]
    \begin{axis}[
        height=\columnwidth/3.4,
        width=\columnwidth/2.6,
        ylabel={\em avg. absolute error},
        xmin=0.5, xmax=6.5,
        ymin=-0.05, ymax=0.5,
        xtick={1,2,3,4,5,6},
        xticklabel style = {font=\footnotesize},
        yticklabel style = {font=\footnotesize},
        xticklabels={0.5,0.2,0.1,0.05,0.02,0.01},
        % ymode=log,
        ytick={0,0.1,0.2,0.3,0.4,0.5},
        % log basis y={10},
        % every axis y label/.style={font=\small,at={(current axis.north)},right=0mm,above=0mm},
        % ylabel near ticks,
        every axis y label/.style={font=\small,at={{(0.28,1.0)}},right=10mm,above=0mm},
        y label style = {font=\footnotesize},
        xlabel near ticks,
        x label style = {font=\tiny},
        legend style={fill=none,font=\tiny,at={(0.02,0.99)},anchor=north west,draw=none},
        x tick label style={rotate=45,anchor=east},
    ]
    \addplot[line width=0.2mm,mark=square,color=my_blue] %GEER 
        plot coordinates {
(1,	0.0552434	)
(2,	0.0530798	)
(3,	0.0429496	)
(4,	0.0307957	)
(5,	0.0183721	)
(6, 0.0068512   )
    };

    \addplot[line width=0.2mm,mark=o,color=my_cyan]  %AMC
        plot coordinates {
(1,	0.0521722	)
(2,	0.0495644	)
(3,	0.0383074	)
(4,	0.0272321	)
(5,	0.0159106	)
(6, 0.0049741   )
    };

    \addplot[line width=0.2mm,mark=pentagon,color=myblue]  %TP
        plot coordinates {
(1,	0.0468985	)
(2,	0.0447126	)
(3,	0.0368985	)
% (4,	1	)
% (5,	1	)
% (6, 1   )
    };
    
    \addplot[line width=0.2mm,mark=triangle,color=my_violet]  %SMM
        plot coordinates {
(1,	0.0491676	)
(2,	0.0456835	)
(3,	0.0338135	)
(4,	0.0221459	)
(5,	0.0100852	)
(6, 0.00386164   )
    };

    \addplot[line width=0.2mm,mark=x,color=my_purple]  %TPC
        plot coordinates {
(1,	0.0501653	)
(2,	0.0476721	)
(3,	0.0372119	)
(4,	0.028957	)
%(5,	1	)
%(6, 1   )
    };
    
    \addplot[line width=0.1mm,color=gray,dashed]  %EXACT
        plot coordinates {
(1,	0.5	)
(2,	0.2	)
(3,	0.1	)
(4,	0.05	)
(5,	0.02	)
(6, 0.01   )
    };

    \addplot[line width=0.2mm,mark=diamond,color=my_pp]  %MPI
        plot coordinates {
(1,	0	)
(2,	0	)
(3,	0	)
(4,	0	)
(5,	0	)
(6, 0   )
    };

    \addplot[line width=0.2mm,mark=10-pointed star,color=my_yellow]  %RP
        plot coordinates {
(1,	0.165254	)
(2,	0.102439	)
(3,	0.064828	)
(4,	0.045570	)
% (5,	0	)
% (6, 0   )
    };

    % \legend{\algo,$\mathtt{TP}$,\fwd}
    \end{axis}
\end{tikzpicture}\hspace{0mm}\label{fig:err-fb}%
}%
\hspace{0mm}\subfloat[DBLP]{
\begin{tikzpicture}[scale=1,every mark/.append style={mark size=2pt}]
    \begin{axis}[
        height=\columnwidth/3.4,
        width=\columnwidth/2.6,
        ylabel={\em avg. absolute error},
        xmin=0.5, xmax=6.5,
        ymin=-0.05, ymax=0.5,
        xtick={1,2,3,4,5,6},
        xticklabel style = {font=\footnotesize},
        yticklabel style = {font=\footnotesize},
        xticklabels={0.5,0.2,0.1,0.05,0.02,0.01},
        % ymode=log,
        ytick={0,0.1,0.2,0.3,0.4,0.5},
        % log basis y={10},
        % log basis y={10},
        % every axis y label/.style={font=\small,at={(current axis.north)},right=0mm,above=0mm},
        % ylabel near ticks,
        every axis y label/.style={font=\small,at={{(0.28,1.0)}},right=10mm,above=0mm},
        y label style = {font=\footnotesize},
        xlabel near ticks,
        x label style = {font=\tiny},
        legend style={fill=none,font=\tiny,at={(0.02,0.99)},anchor=north west,draw=none},
        x tick label style={rotate=45,anchor=east},
    ]
    \addplot[line width=0.2mm,mark=square,color=my_blue]  
        plot coordinates {
(1,	0.0981937	)
(2,	0.0705768	)
(3,	0.0540884	)
(4,	0.043426	)
(5,	0.0194858	)
(6, 0.0095041   )
    };

    \addplot[line width=0.2mm,mark=o,color=my_cyan]  
        plot coordinates {
(1,	0.0814698	)
(2,	0.0563393	)
(3,	0.0456085	)
(4,	0.0367994	)
(5,	0.018436	)
(6, 0.00936595   )
    };
    
    \addplot[line width=0.2mm,mark=pentagon,color=myblue]  
        plot coordinates {
(1,	0.0509323	)
(2,	0.0509323	)
    };
    
    \addplot[line width=0.2mm,mark=triangle,color=my_violet]  %SMM
        plot coordinates {
(1,	0.0324137	)
(2,	0.019271	)
(3,	0.0135058	)
(4,	0.00958136	)
(5,	0.00614012	)
(6, 0.00442718   )
    };
    
    \addplot[line width=0.2mm,mark=x,color=my_purple]  %TPC
        plot coordinates {
(1,	0.0517332	)
(2,	0.051124	)
(3,	0.041377	)
(4,	0.032862	)
    };

    \addplot[line width=0.1mm,color=gray,dashed]  %EXACT
        plot coordinates {
(1,	0.5	)
(2,	0.2	)
(3,	0.1	)
(4,	0.05	)
(5,	0.02	)
(6, 0.01   )
    };

    \addplot[line width=0.2mm,mark=diamond,color=my_pp]  %MPI
        plot coordinates {
    };

    \addplot[line width=0.2mm,mark=10-pointed star,color=my_yellow]  %RP
        plot coordinates {
(1,	0.183488	)
(2,	0.106618	)
    };

    % \legend{\algo,$\mathtt{TP}$,\fwd}
    \end{axis}
\end{tikzpicture}\hspace{0mm}\label{fig:err-dblp}%
}%
\hspace{0mm}\subfloat[Youtube]{
\begin{tikzpicture}[scale=1,every mark/.append style={mark size=2pt}]
    \begin{axis}[
        height=\columnwidth/3.4,
        width=\columnwidth/2.6,
        ylabel={\em avg. absolute error},
        xmin=0.5, xmax=6.5,
        ymin=-0.05, ymax=0.5,
        xtick={1,2,3,4,5,6},
        xticklabel style = {font=\footnotesize},
        yticklabel style = {font=\footnotesize},
        xticklabels={0.5,0.2,0.1,0.05,0.02,0.01},
        % ymode=log,
        ytick={0,0.1,0.2,0.3,0.4,0.5},
        % log basis y={10},
        % ytick={1,10,100,1000,10000,100000,1000000,10000000,100000000},
        % log basis y={10},
        % every axis y label/.style={font=\small,at={(current axis.north)},right=0mm,above=0mm},
        % ylabel near ticks,
        every axis y label/.style={font=\small,at={{(0.28,1.0)}},right=10mm,above=0mm},
        y label style = {font=\footnotesize},
        xlabel near ticks,
        x label style = {font=\tiny},
        legend style={fill=none,font=\tiny,at={(0.02,0.99)},anchor=north west,draw=none},
        x tick label style={rotate=45,anchor=east},
    ]
    \addplot[line width=0.2mm,mark=square,color=my_blue]  
        plot coordinates {
(1,	0.157819	)
(2,	0.111401	)
(3,	0.065968	)
(4,	0.042682	)
(5,	0.01884613	)
(6, 0.008686919   )
    };

    \addplot[line width=0.2mm,mark=o,color=my_cyan]  
        plot coordinates {
(1,	0.141676	)
(2,	0.108399	)
(3,	0.063792	)
(4,	0.0403109	)
(5,	0.01726273	)
(6, 0.008631515   )
    };

    \addplot[line width=0.2mm,mark=pentagon,color=myblue]  
        plot coordinates {
(1,	0.115084	)
(2,	0.0932868	)
(3,	0.061524	)
    };
    
    \addplot[line width=0.2mm,mark=triangle,color=my_violet]  %SMM
        plot coordinates {
(1,	0.08862	)
(2,	0.0612911	)
(3,	0.0468838	)
(4,	0.0358639	)
(5,	0.0148118	)
(6, 0.0069799   )
    };

    \addplot[line width=0.2mm,mark=x,color=my_purple]  %TPC
        plot coordinates {
(1,	0.126355	)
(2,	0.096732	)
(3,	0.061721	)
(4,	0.03729	)
    };

    \addplot[line width=0.1mm,color=gray,dashed]  %EXACT
        plot coordinates {
(1,	0.5	)
(2,	0.2	)
(3,	0.1	)
(4,	0.05	)
(5,	0.02	)
(6, 0.01   )
    };

    \addplot[line width=0.2mm,mark=diamond,color=my_pp]  %MPI
        plot coordinates {
    };

    \addplot[line width=0.2mm,mark=10-pointed star,color=my_yellow]  %RP
        plot coordinates {
(1,	0.288732	)
(2,	0.15740	)
    };

    % \legend{\algo,$\mathtt{TP}$,\fwd}
    \end{axis}
\end{tikzpicture}\hspace{0mm}\label{fig:err-ytb}%
}%
\vskip -3mm

\hspace{-5mm}\subfloat[Orkut]{
\begin{tikzpicture}[scale=1,every mark/.append style={mark size=2pt}]
    \begin{axis}[
        height=\columnwidth/3.4,
        width=\columnwidth/2.6,
        ylabel={\em avg. absolute error},
        xmin=0.5, xmax=6.5,
        ymin=-0.05, ymax=0.5,
        xtick={1,2,3,4,5,6},
        xticklabel style = {font=\footnotesize},
        yticklabel style = {font=\footnotesize},
        xticklabels={0.5,0.2,0.1,0.05,0.02,0.01},
        % ymode=log,
        ytick={0,0.1,0.2,0.3,0.4,0.5},
        % log basis y={10},
        % every axis y label/.style={font=\small,at={(current axis.north)},right=0mm,above=0mm},
        % ylabel near ticks,
        every axis y label/.style={font=\small,at={{(0.28,1.0)}},right=10mm,above=0mm},
        y label style = {font=\footnotesize},
        xlabel near ticks,
        x label style = {font=\tiny},
        legend style={fill=none,font=\tiny,at={(1.07, 0.7)},anchor=north west,draw=none},
        x tick label style={rotate=45,anchor=east},
    ]
    \addplot[line width=0.2mm,mark=square,color=my_blue]  
        plot coordinates {
(1,	0.0039301	)
(2,	0.000152879	)
(3,	4.7967e-05	)
(4,	0.000828325	)
(5,	0.00123099	)
(6, 0.000953153   )
    };

    \addplot[line width=0.2mm,mark=o,color=my_cyan]  
        plot coordinates {
(1,	0.000666337	)
(2,	0.000211622	)
(3,	8.87151e-05	)
(4,	2.96525e-05	)
(5,	9.68774e-06	)
(6, 9.66125e-06   )
    };

    \addplot[line width=0.2mm,mark=pentagon,color=myblue]  
        plot coordinates {
(1,	3.27491e-06	)
(2,	1.95486e-06	)
    };

    \addplot[line width=0.2mm,mark=triangle,color=my_violet]  %SMM
        plot coordinates {
(1,	0.00369024	)
(2,	1.49638e-06	)
(3,	1.2971e-06	)
(4,	1.18183e-06	)
(5,	1.0557e-06	)
(6, 9.69512e-07   )
    };

    \addplot[line width=0.2mm,mark=x,color=my_purple]  %TPC
        plot coordinates {
(1,	4.18933e-06	)
(2,	2.77218e-06	)
(3,	1.97524e-06 )
(4,	1.64653e-06	)
    };

    \addplot[line width=0.1mm,color=gray,dashed]  %EXACT
        plot coordinates {
(1,	0.5	)
(2,	0.2	)
(3,	0.1	)
(4,	0.05	)
(5,	0.02	)
(6, 0.01   )
    };

    \addplot[line width=0.2mm,mark=diamond,color=my_pp]  %MPI
        plot coordinates {
    };

    \addplot[line width=0.2mm,mark=10-pointed star,color=my_yellow]  %RP
        plot coordinates {
    };

    % \legend{\algo,$\mathtt{TP}$,\fwd}
    \end{axis}
\end{tikzpicture}\hspace{0mm}\label{fig:err-ok}%
}%
\hspace{0mm}\subfloat[LiveJournal]{
\begin{tikzpicture}[scale=1,every mark/.append style={mark size=2pt}]
    \begin{axis}[
        height=\columnwidth/3.4,
        width=\columnwidth/2.6,
        ylabel={\em avg. absolute error},
        xmin=0.5, xmax=6.5,
        ymin=-0.05, ymax=0.5,
        xtick={1,2,3,4,5,6},
        xticklabel style = {font=\footnotesize},
        yticklabel style = {font=\footnotesize},
        xticklabels={0.5,0.2,0.1,0.05,0.02,0.01},
        % ymode=log,
        ytick={0,0.1,0.2,0.3,0.4,0.5},
        % log basis y={10},
        % every axis y label/.style={font=\small,at={(current axis.north)},right=0mm,above=0mm},
        % ylabel near ticks,
        every axis y label/.style={font=\small,at={{(0.28,1.0)}},right=10mm,above=0mm},
        y label style = {font=\footnotesize},
        xlabel near ticks,
        x label style = {font=\tiny},
        legend style={fill=none,font=\tiny,at={(0.02,0.99)},anchor=north west,draw=none},
        x tick label style={rotate=45,anchor=east},
    ]
    \addplot[line width=0.2mm,mark=square,color=my_blue]  
        plot coordinates {
(1,	0.0227113	)
(2,	0.0169304	)
(3,	0.0124488	)
(4,	0.00944245	)
(5,	0.00464042	)
(6, 0.00421879   )
    };

    \addplot[line width=0.2mm,mark=o,color=my_cyan]  
        plot coordinates {
(1,	0.00270396	)
(2,	0.000877551	)
(3,	0.000523045	)
(4,	0.000325757	)
(5,	0.000223641	)
(6, 0.000210149   )
    };

    \addplot[line width=0.2mm,mark=pentagon,color=myblue]  
        plot coordinates {
(1,	0.000246237	)
(2,	0.000256092	)
    };

    \addplot[line width=0.2mm,mark=triangle,color=my_violet]  %SMM
        plot coordinates {
(1,	0.000250742	)
(2,	7.9637e-05	)
(3,	5.27111e-05	)
(4,	3.81136e-05	)
(5,	2.70406e-05	)
(6, 2.16343e-05   )
    };

    \addplot[line width=0.2mm,mark=x,color=my_purple]  %TPC
        plot coordinates {
(1,	0.000328472	)
(2,	0.000293954	)
(3,	0.000264165	)
(4,	0.000201866	)
    };

    \addplot[line width=0.1mm,color=gray,dashed]  %EXACT
        plot coordinates {
(1,	0.5	)
(2,	0.2	)
(3,	0.1	)
(4,	0.05	)
(5,	0.02	)
(6, 0.01   )
    };

    \addplot[line width=0.2mm,mark=diamond,color=my_pp]  %MPI
        plot coordinates {
    };

    \addplot[line width=0.2mm,mark=10-pointed star,color=my_yellow]  %RP
        plot coordinates {
    };

    % \legend{\algo,$\mathtt{TP}$,\fwd}
    \end{axis}
\end{tikzpicture}\hspace{0mm}\label{fig:err-lj}%
}%
\hspace{0mm}\subfloat[Friendster]{
\begin{tikzpicture}[scale=1,every mark/.append style={mark size=2pt}]
    \begin{axis}[
        height=\columnwidth/3.4,
        width=\columnwidth/2.6,
        ylabel={\em avg. absolute error},
        xmin=0.5, xmax=6.5,
        ymin=-0.05, ymax=0.5,
        xtick={1,2,3,4,5,6},
        xticklabel style = {font=\footnotesize},
        yticklabel style = {font=\footnotesize},
        xticklabels={0.5,0.2,0.1,0.05,0.02,0.01},
        % ymode=log,
        ytick={0,0.1,0.2,0.3,0.4,0.5},
        % log basis y={10},
        % every axis y label/.style={font=\small,at={(current axis.north)},right=0mm,above=0mm},
        % ylabel near ticks,
        every axis y label/.style={font=\small,at={{(0.28,1.0)}},right=10mm,above=0mm},
        y label style = {font=\footnotesize},
        xlabel near ticks,
        x label style = {font=\tiny},
        legend style={fill=none,font=\tiny,at={(1.07, 0.7)},anchor=north west,draw=none},
        x tick label style={rotate=45,anchor=east},
    ]
    \addplot[line width=0.2mm,mark=square,color=my_blue]  
        plot coordinates {
(1,	0.000499829	)
(2,	7.64405e-05	)
(3,	4.12276e-05	)
(4,	2.79905e-05	)
(5,	0.00164494	)
(6, 0.00164622   )
    };

    \addplot[line width=0.2mm,mark=o,color=my_cyan]  
        plot coordinates {
(1,	0.000446938	)
(2,	0.000149109	)
(3,	2.72941e-05	)
(4,	1.76088e-05	)
(5,	1.91168e-05	)
(6, 1.05283e-05   )
    };

    \addplot[line width=0.2mm,mark=pentagon,color=myblue]  
        plot coordinates {
(1,	3.45388e-06	)
    };
    
    \addplot[line width=0.2mm,mark=triangle,color=my_violet]  
        plot coordinates {
(1,	1	)
(2,	1	)
(3,	1	)
(4,	1	)
(5,	1	)
(6, 1   )
    };

    \addplot[line width=0.2mm,mark=x,color=my_purple]  %TPC
        plot coordinates {
(1,	2.93486e-05	)
(2,	1.45777e-05	)
(3,	6.74292e-06	)
(4,	2.85227e-06	)

    };

    \addplot[line width=0.1mm,color=gray,dashed]  %EXACT
        plot coordinates {
(1,	0.5	)
(2,	0.2	)
(3,	0.1	)
(4,	0.05	)
(5,	0.02	)
(6, 0.01   )
    };

    \addplot[line width=0.2mm,mark=diamond,color=my_pp]  %MPI
        plot coordinates {

    };

    \addplot[line width=0.2mm,mark=10-pointed star,color=my_yellow]  %RP
        plot coordinates {

    };

    % \legend{\algo,$\mathtt{TP}$,\fwd}
    \end{axis}
\end{tikzpicture}\label{fig:err-fri}%
}%
\end{small}
% \vspace{-1mm}
% \vskip -1mm
\caption{Avg. absolute error vs. $\epsilon$ for random queries.} \label{fig:err}
\vspace{-2mm}
\end{figure}

\subsection{Query Accuracy}\label{sec:exp-acc}
Fig. \ref{fig:err} and Fig. \ref{fig:err-edge} report the actual average absolute error by each method on each dataset when varying $\epsilon$ from $0.01$ to $0.5$ for random query sets and edge query sets, respectively. Note that the $x$-axis and $y$-axis represent the given absolute error threshold $\epsilon$ and the actual average absolute error, respectively. The gray dashed line in each figure represents the boundary between successful and failed query results. That is, each data point below the line means its actual absolute error is less than the corresponding desired absolute error threshold $\epsilon$ and indicates it is a successful query result; otherwise, it is a failed query result. As observed in Fig. \ref{fig:err} and Fig. \ref{fig:err-edge}, all tested methods return accurate query results, whose actual absolute errors are less than the given error threshold $\epsilon$. More specifically, as we can see from Fig. \ref{fig:err}, most methods achieve low average absolute errors smaller than $0.1$ for random queries even given a large error threshold $\epsilon=0.5$ and the errors approach $0$ as $\epsilon$ is lowered. \revise{In particular, the random-projection-based method $\mathtt{RP}$ always produces the highest empirical errors on Facebook, DBLP, and Youtube. As for \algo and \fwd, they yield slightly larger errors compared to the competitors \bwd, $\mathtt{TP}$, and $\mathtt{TPC}$ on DBLP and Youtube graphs.} The reason is that some of the excessive efforts made in \bwd, $\mathtt{TP}$, and $\mathtt{TPC}$ (e.g., matrix-vector multiplications or random walks) may lead to a more accurate estimation of ER, while \algo and \fwd eliminate these costs as they are unnecessary. In contrast, on graphs with high average degrees, i.e., Facebook, Orkut, LiveJournal, and Friendster, \algo and \fwd consistently attain insignificant errors (around $10^{-4}$) as the competitors do. This is because ER $r(s,t)$ is inversely proportional to the degrees of nodes $s$ and $t$, and hence the ER values as well as empirical errors in these graphs tend to be of small scale.
In Fig. \ref{fig:err-edge}, observe that \algo, \fwd, \bwd, \revise{$\mathtt{HAY}$,} and $\mathtt{MC2}$ all show similarly high accuracy for edge queries on all datasets. This is expected, since (i) the node pair of an edge $(s,t)$ is likely to have multiple short connections to each other, making it easy to estimate accurately using random walks or graph traversals,
and (ii) the ER value $r(s,t)$ (i.e., the dissimilarity between $s$ and $t$) of node pair $(s,t)$ sharing an edge is usually small.

Together with the observations in Section \ref{sec:exp-time}, we can conclude that \algo and \fwd significantly improve over existing work in the practical efficiency without affecting its result quality.

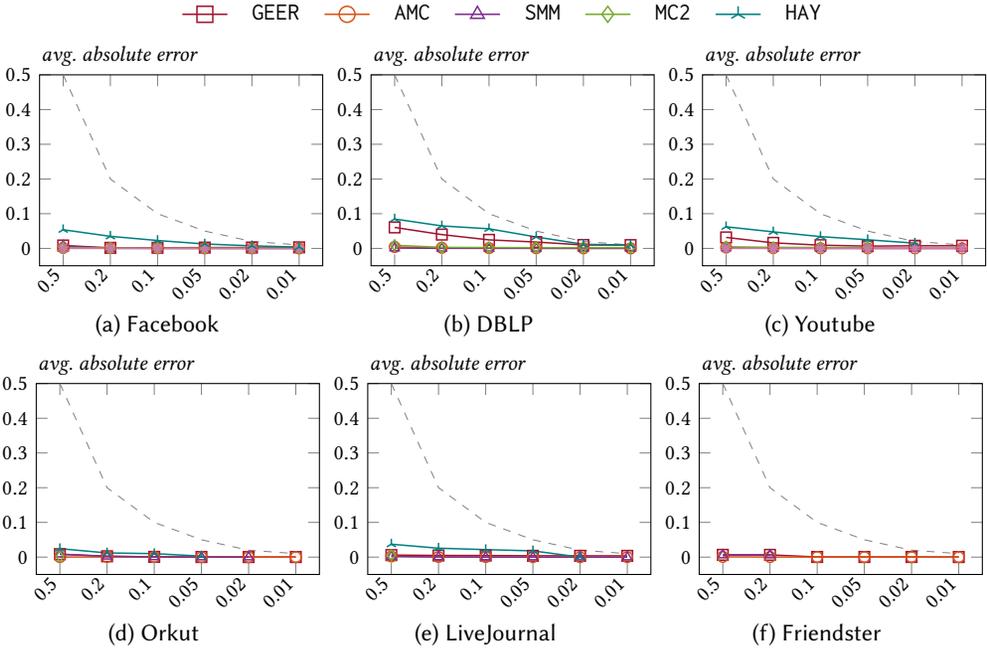
\begin{figure}[!t]
\centering
\begin{small}
\begin{tikzpicture}[every mark/.append style={mark size=3pt}]
    \begin{customlegend}[legend columns=5,
        legend entries={\algo,\fwd,\bwd,$\mathtt{MC2}$,$\mathtt{HAY}$},
        legend style={at={(0.5,1.05)},anchor=north,draw=none,font=\small,column sep=0.25cm}]
    \addlegendimage{line width=0.2mm,mark=square,color=my_blue}
    \addlegendimage{line width=0.2mm,mark=o,color=my_cyan}
    \addlegendimage{line width=0.2mm,mark=triangle,color=my_violet}
    \addlegendimage{line width=0.2mm,mark=diamond,color=my_teal}
    \addlegendimage{line width=0.2mm,mark=Mercedes star,color=my_green}
    % \addlegendimage{line width=0.2mm,mark=pentagon,color=myblue}
    % \addlegendimage{line width=0.2mm,mark=x,color=my_purple}
    \end{customlegend}
\end{tikzpicture}
\\[-\lineskip]
\vspace{-3mm}

\hspace{-4mm}\subfloat[Facebook]{
\begin{tikzpicture}[scale=1,every mark/.append style={mark size=2pt}]
    \begin{axis}[
        height=\columnwidth/3.4,
        width=\columnwidth/2.6,
        ylabel={\em avg. absolute error},
        xmin=0.5, xmax=6.5,
        ymin=-0.05, ymax=0.5,
        xtick={1,2,3,4,5,6},
        xticklabel style = {font=\footnotesize},
        yticklabel style = {font=\footnotesize},
        xticklabels={0.5,0.2,0.1,0.05,0.02,0.01},
        % ymode=log,
        ytick={0,0.1,0.2,0.3,0.4,0.5},
        % log basis y={10},
        % every axis y label/.style={font=\small,at={(current axis.north)},right=0mm,above=0mm},
        % ylabel near ticks,
        every axis y label/.style={font=\small,at={{(0.28,1.0)}},right=10mm,above=0mm},
        y label style = {font=\footnotesize},
        xlabel near ticks,
        x label style = {font=\tiny},
        legend style={fill=none,font=\tiny,at={(0.02,0.99)},anchor=north west,draw=none},
        x tick label style={rotate=45,anchor=east},
    ]
    \addplot[line width=0.2mm,mark=square,color=my_blue] %GEER 
        plot coordinates {
(1,	0.008397	)
(2,	0.00165926	)
(3,	0.00178989	)
(4,	0.00237207	)
(5,	0.00250084	)
(6, 0.00303193   )
    };

    \addplot[line width=0.2mm,mark=o,color=my_cyan]  %AMC
        plot coordinates {
(1,	0.00128875	)
(2,	0.000879996	)
(3,	0.000230558	)
(4,	0.000149246	)
(5,	4.93015e-05	)
(6, 2.27207e-05   )
    };
    
    \addplot[line width=0.2mm,mark=triangle,color=my_violet]  %SMM
        plot coordinates {
(1,	0.0070811	)
(2,	4.28956e-05	)
(3,	3.58937e-05	)
(4,	3.24856e-05	)
(5,	2.91572e-05	)
(6, 2.72486e-05   )
    };

    \addplot[line width=0.2mm,mark=diamond,color=my_teal]  %MC2
        plot coordinates {
(1,	0.00194408	)
(2,	0.000743012	)
(3,	0.000312229	)
(4,	0.000215936	)
(5,	0.000182061	)
(6, 0.000181892   )
    };

    \addplot[line width=0.1mm,color=gray,dashed]  %EXACT
        plot coordinates {
(1,	0.5	)
(2,	0.2	)
(3,	0.1	)
(4,	0.05	)
(5,	0.02	)
(6, 0.01   )
    };

    \addplot[line width=0.2mm,mark=10-pointed star,color=my_red]  %naive SMM
        plot coordinates {
(1,	0	)
(2,	0	)
(3,	0	)
(4,	0	)
(5,	0	)
(6, 0   )
    };

    \addplot[line width=0.2mm,mark=Mercedes star,color=my_green]  %HAY
        plot coordinates {
(1,	0.0537665	)
(2,	0.034801	)
(3,	0.0227952	)
(4,	0.0127903	)
(5,	0.007396	)
(6, 0.0037677   )
    };
    % \legend{\algo,$\mathtt{TP}$,\fwd}
    \end{axis}
\end{tikzpicture}\hspace{0mm}\label{fig:err-fb-edge}%
}%
\hspace{0mm}\subfloat[DBLP]{
\begin{tikzpicture}[scale=1,every mark/.append style={mark size=2pt}]
    \begin{axis}[
        height=\columnwidth/3.4,
        width=\columnwidth/2.6,
        ylabel={\em avg. absolute error},
        xmin=0.5, xmax=6.5,
        ymin=-0.05, ymax=0.5,
        xtick={1,2,3,4,5,6},
        xticklabel style = {font=\footnotesize},
        yticklabel style = {font=\footnotesize},
        xticklabels={0.5,0.2,0.1,0.05,0.02,0.01},
        % ymode=log,
        ytick={0,0.1,0.2,0.3,0.4,0.5},
        % log basis y={10},
        % every axis y label/.style={font=\small,at={(current axis.north)},right=0mm,above=0mm},
        % ylabel near ticks,
        every axis y label/.style={font=\small,at={{(0.28,1.0)}},right=10mm,above=0mm},
        y label style = {font=\footnotesize},
        xlabel near ticks,
        x label style = {font=\tiny},
        legend style={fill=none,font=\tiny,at={(0.02,0.99)},anchor=north west,draw=none},
        x tick label style={rotate=45,anchor=east},
    ]
    \addplot[line width=0.2mm,mark=square,color=my_blue] %GEER 
        plot coordinates {
(1,	0.0605024	)
(2,	0.0398351	)
(3,	0.0245901	)
(4,	0.0183509	)
(5,	0.00987383	)
(6, 0.0095032   )
    };

    \addplot[line width=0.2mm,mark=o,color=my_cyan]  %AMC
        plot coordinates {
(1,	0.00410995	)
(2,	0.00177331	)
(3,	0.000900959	)
(4,	0.00046207	)
(5,	0.000193087	)
(6, 0.000126233   )
    };
    
    \addplot[line width=0.2mm,mark=triangle,color=my_violet]  %SMM
        plot coordinates {
(1,	0.000966007	)
(2,	0.000246597	)
(3,	0.000116973	)
(4,	6.27416e-05	)
(5,	3.17838e-05	)
(6, 2.08197e-05   )
    };

    \addplot[line width=0.2mm,mark=diamond,color=my_teal]  %MC2
        plot coordinates {
(1,	0.00917748	)
(2,	0.00298073	)
(3,	0.00252822	)
(4,	0.00195581	)
(5,	0.00203062	)
(6, 0.00203406   )
    };

    \addplot[line width=0.1mm,color=gray,dashed]  %EXACT
        plot coordinates {
(1,	0.5	)
(2,	0.2	)
(3,	0.1	)
(4,	0.05	)
(5,	0.02	)
(6, 0.01   )
    };

    \addplot[line width=0.2mm,mark=Mercedes star,color=my_green]  %HAY
        plot coordinates {
(1,	0.0846865	)
(2,	0.0648151	)
(3,	0.0565335	)
(4,	0.0329135	)
(5,	0.011891	)
(6, 0.0095872   )
    };
    % \legend{\algo,$\mathtt{TP}$,\fwd}
    \end{axis}
\end{tikzpicture}\hspace{0mm}\label{fig:err-dblp-edge}%
}%
\hspace{0mm}\subfloat[Youtube]{
\begin{tikzpicture}[scale=1,every mark/.append style={mark size=2pt}]
    \begin{axis}[
        height=\columnwidth/3.4,
        width=\columnwidth/2.6,
        ylabel={\em avg. absolute error},
        xmin=0.5, xmax=6.5,
        ymin=-0.05, ymax=0.5,
        xtick={1,2,3,4,5,6},
        xticklabel style = {font=\footnotesize},
        yticklabel style = {font=\footnotesize},
        xticklabels={0.5,0.2,0.1,0.05,0.02,0.01},
        % ymode=log,
        ytick={0,0.1,0.2,0.3,0.4,0.5},
        % log basis y={10},
        % every axis y label/.style={font=\small,at={(current axis.north)},right=0mm,above=0mm},
        % ylabel near ticks,
        every axis y label/.style={font=\small,at={{(0.28,1.0)}},right=10mm,above=0mm},
        y label style = {font=\footnotesize},
        xlabel near ticks,
        x label style = {font=\tiny},
        legend style={fill=none,font=\tiny,at={(0.02,0.99)},anchor=north west,draw=none},
        x tick label style={rotate=45,anchor=east},
    ]
    \addplot[line width=0.2mm,mark=square,color=my_blue] %GEER 
        plot coordinates {
(1,	0.0316945	)
(2,	0.01601095	)
(3,	0.00933864	)
(4,	0.00659208	)
(5,	0.00754798	)
(6, 0.00765881   )
    };

    \addplot[line width=0.2mm,mark=o,color=my_cyan]  %AMC
        plot coordinates {
(1,	0.00254203	)
(2,	0.000636188	)
(3,	0.000336538	)
(4,	0.000238191	)
(5,	9.92478e-05	)
(6, 6.07917e-05   )
    };
    
    \addplot[line width=0.2mm,mark=triangle,color=my_violet]  %SMM
        plot coordinates {
(1,	0.00159251	)
(2,	0.000143567	)
(3,	5.11988e-05	)
(4,	1.88043e-05	)
(5,	5.63501e-06	)
(6, 2.13464e-06   )
    };

    \addplot[line width=0.2mm,mark=diamond,color=my_teal]  %MC2
        plot coordinates {
(1,	0.00459504	)
(2,	0.00304836	)
(3,	0.00258652	)
(4,	0.00234078	)
% (5,	1	)
% (6, 1   )
    };

    \addplot[line width=0.1mm,color=gray,dashed]  %EXACT
        plot coordinates {
(1,	0.5	)
(2,	0.2	)
(3,	0.1	)
(4,	0.05	)
(5,	0.02	)
(6, 0.01   )
    };

    \addplot[line width=0.2mm,mark=10-pointed star,color=my_red]  %naive SMM
        plot coordinates {
(1,	0	)
(2,	0	)
(3,	0	)
(4,	0	)
(5,	0	)
(6, 0   )
    };

    \addplot[line width=0.2mm,mark=Mercedes star,color=my_green]  %HAY
        plot coordinates {
(1,	0.0622898	)
(2,	0.0476646	)
(3,	0.0339993	)
(4,	0.0248955	)
(5,	0.0151608	)
% (6, 0   )
    };
    % \legend{\algo,$\mathtt{TP}$,\fwd}
    \end{axis}
\end{tikzpicture}\hspace{0mm}\label{fig:err-ytb-edge}%
}%
\vskip -3mm

\hspace{-4mm}\subfloat[Orkut]{
\begin{tikzpicture}[scale=1,every mark/.append style={mark size=2pt}]
    \begin{axis}[
        height=\columnwidth/3.4,
        width=\columnwidth/2.6,
        ylabel={\em avg. absolute error},
        xmin=0.5, xmax=6.5,
        ymin=-0.05, ymax=0.5,
        xtick={1,2,3,4,5,6},
        xticklabel style = {font=\footnotesize},
        yticklabel style = {font=\footnotesize},
        xticklabels={0.5,0.2,0.1,0.05,0.02,0.01},
        % ymode=log,
        ytick={0,0.1,0.2,0.3,0.4,0.5},
        % log basis y={10},
        % every axis y label/.style={font=\small,at={(current axis.north)},right=0mm,above=0mm},
        % ylabel near ticks,
        every axis y label/.style={font=\small,at={{(0.28,1.0)}},right=10mm,above=0mm},
        y label style = {font=\footnotesize},
        xlabel near ticks,
        x label style = {font=\tiny},
        legend style={fill=none,font=\tiny,at={(1.07, 0.7)},anchor=north west,draw=none},
        x tick label style={rotate=45,anchor=east},
    ]
    \addplot[line width=0.2mm,mark=square,color=my_blue] %GEER 
        plot coordinates {
(1,	0.0084728	)
(2,	0.00246647	)
(3,	0.000754122	)
(4,	0.000316425	)
(5,	0.00020295	)
(6, 0.000255957   )
    };

    \addplot[line width=0.2mm,mark=o,color=my_cyan]  %AMC
        plot coordinates {
(1,	0.000222513	)
(2,	0.000356938	)
(3,	0.00026818	)
(4,	2.9313e-05	)
(5,	1.53668e-05	)
(6, 6.22219e-06   )
    };
    
    \addplot[line width=0.2mm,mark=triangle,color=my_violet]  %SMM
        plot coordinates {
(1,	0.00830264	)
(2,	0.00230645	)
(3,	0.000506044	)
(4,	0.000101199	)
(5,	7.30835e-09	)
% (6, 0   )
    };

    \addplot[line width=0.2mm,mark=diamond,color=my_teal]  %MC2
        plot coordinates {
(1,	0.00230544	)
% (2,	1	)
% (3,	1	)
% (4,	1	)
% (5,	1	)
% (6, 1   )
    };
    
    \addplot[line width=0.1mm,color=gray,dashed]  %EXACT
        plot coordinates {
(1,	0.5	)
(2,	0.2	)
(3,	0.1	)
(4,	0.05	)
(5,	0.02	)
(6, 0.01   )
    };

    \addplot[line width=0.2mm,mark=Mercedes star,color=my_green]  %HAY
        plot coordinates {
(1,	0.0243344	)
(2,	0.0120167	)
(3,	0.0100269	)
(4,	0.00286349	)
% (5,	0	)
% (6, 0   )
    };
    % \legend{\algo,$\mathtt{TP}$,\fwd}
    \end{axis}
\end{tikzpicture}\hspace{0mm}\label{fig:err-ok-edge}%
}%
\hspace{0mm}\subfloat[LiveJournal]{
\begin{tikzpicture}[scale=1,every mark/.append style={mark size=2pt}]
    \begin{axis}[
        height=\columnwidth/3.4,
        width=\columnwidth/2.6,
        ylabel={\em avg. absolute error},
        xmin=0.5, xmax=6.5,
        ymin=-0.05, ymax=0.5,
        xtick={1,2,3,4,5,6},
        xticklabel style = {font=\footnotesize},
        yticklabel style = {font=\footnotesize},
        xticklabels={0.5,0.2,0.1,0.05,0.02,0.01},
        % ymode=log,
        ytick={0,0.1,0.2,0.3,0.4,0.5},
        % log basis y={10},
        % every axis y label/.style={font=\small,at={(current axis.north)},right=0mm,above=0mm},
        % ylabel near ticks,
        every axis y label/.style={font=\small,at={{(0.28,1.0)}},right=10mm,above=0mm},
        y label style = {font=\footnotesize},
        xlabel near ticks,
        x label style = {font=\tiny},
        legend style={fill=none,font=\tiny,at={(0.02,0.99)},anchor=north west,draw=none},
        x tick label style={rotate=45,anchor=east},
    ]
    \addplot[line width=0.2mm,mark=square,color=my_blue] %GEER 
        plot coordinates {
(1,	0.00584884	)
(2,	0.00478515	)
(3,	0.00451454	)
(4,	0.00382201	)
(5,	0.00366193	)
(6, 0.00360918   )
    };

    \addplot[line width=0.2mm,mark=o,color=my_cyan]  %AMC
        plot coordinates {
(1,	0.0019851	)
(2,	0.000585142	)
(3,	0.00026061	)
(4,	7.1861e-05	)
(5,	3.14095e-05	)
(6, 2.48982e-05   )
    };
    
    \addplot[line width=0.2mm,mark=triangle,color=my_violet]  %SMM
        plot coordinates {
(1,	0.00189284	)
(2,	0.000532775	)
(3,	0.000145172	)
(4,	1.11194e-05	)
(5,	6.26802e-06	)
(6, 4.43345e-06   )
    };

    \addplot[line width=0.2mm,mark=diamond,color=my_teal]  %MC2
        plot coordinates {
(1,	0.00514261	)
% (2,	1	)
% (3,	1	)
% (4,	1	)
% (5,	1	)
% (6, 1   )
    };

    \addplot[line width=0.1mm,color=gray,dashed]  %EXACT
        plot coordinates {
(1,	0.5	)
(2,	0.2	)
(3,	0.1	)
(4,	0.05	)
(5,	0.02	)
(6, 0.01   )
    };

%     \addplot[line width=0.2mm,mark=10-pointed star,color=my_red]  %naive SMM
%         plot coordinates {
% (1,	0	)
% (2,	0	)
% (3,	0	)
% (4,	0	)
% (5,	0	)
% (6, 0   )
%     };

    \addplot[line width=0.2mm,mark=Mercedes star,color=my_green]  %HAY
        plot coordinates {
(1,	0.0372022	)
(2,	0.0256308	)
(3,	0.0215973	)
(4,	0.0180124	)
(5,	0	)
% (6, 0   )
    };
    % \legend{\algo,$\mathtt{TP}$,\fwd}
    \end{axis}
\end{tikzpicture}\hspace{0mm}\label{fig:err-lj-edge}%
}%
\hspace{0mm}\subfloat[Friendster]{
\begin{tikzpicture}[scale=1,every mark/.append style={mark size=2pt}]
    \begin{axis}[
        height=\columnwidth/3.4,
        width=\columnwidth/2.6,
        ylabel={\em avg. absolute error},
        xmin=0.5, xmax=6.5,
        ymin=-0.05, ymax=0.5,
        xtick={1,2,3,4,5,6},
        xticklabel style = {font=\footnotesize},
        yticklabel style = {font=\footnotesize},
        xticklabels={0.5,0.2,0.1,0.05,0.02,0.01},
        % ymode=log,
        ytick={0,0.1,0.2,0.3,0.4,0.5},
        % log basis y={10},
        % every axis y label/.style={font=\small,at={(current axis.north)},right=0mm,above=0mm},
        % ylabel near ticks,
        every axis y label/.style={font=\small,at={{(0.28,1.0)}},right=10mm,above=0mm},
        y label style = {font=\footnotesize},
        xlabel near ticks,
        x label style = {font=\tiny},
        legend style={fill=none,font=\tiny,at={(1.07, 0.7)},anchor=north west,draw=none},
        x tick label style={rotate=45,anchor=east},
    ]
    \addplot[line width=0.2mm,mark=square,color=my_blue] %GEER 
        plot coordinates {
(1,	0.0065482	)
(2,	0.00646336	)
(3,	0.000436166	)
(4,	0.00040828	)
(5,	0.000415959	)
(6, 0.00024731   )
    };

    \addplot[line width=0.2mm,mark=o,color=my_cyan]  %AMC
        plot coordinates {
(1,	0.00076831	)
(2,	0.000291943	)
(3,	0.00014945	)
(4,	9.92335e-05	)
(5,	1.24475e-05	)
(6, 1.18031e-05   )
    };
    
    \addplot[line width=0.2mm,mark=triangle,color=my_violet]  %SMM
        plot coordinates {
(1,	0.00592854	)
(2,	0.00592752	)
% (3,	1	)
% (4,	1	)
% (5,	1	)
% (6, 1   )
    };

    \addplot[line width=0.2mm,mark=diamond,color=my_teal]  %MC2
        plot coordinates {
(1,	1	)
(2,	1	)
(3,	1	)
(4,	1	)
(5,	1	)
(6, 1   )
    };

    \addplot[line width=0.1mm,color=gray,dashed]  %EXACT
        plot coordinates {
(1,	0.5	)
(2,	0.2	)
(3,	0.1	)
(4,	0.05	)
(5,	0.02	)
(6, 0.01   )
    };

    \addplot[line width=0.2mm,mark=Mercedes star,color=my_green]  %HAY
        plot coordinates {
    };
    % \legend{\algo,$\mathtt{TP}$,\fwd}
    \end{axis}
\end{tikzpicture}\label{fig:err-fri-edge}%
}%
\end{small}
\vspace{-1mm}
\vskip -1mm
\caption{Avg. absolute error vs. $\epsilon$ for edge queries.} \label{fig:err-edge}
%\vspace{-2mm}
\end{figure}

\subsection{Parameter Analysis}\label{sec:exp-param}
\revise{
This section first studies the effect of the input parameter $\tau$  (i.e., the maximum number of batches of random walks) in \fwd and \algo. Next, we vary the intermediate parameter $\ell_b$ in \algo so as to validate the effectiveness of the greedy strategy for combining \fwd and \bwd (Section \ref{sec:greedy-combi}). Lastly, we test and verify the efficiency benefits brought by our refined maximum length $\ell$ (Theorem \ref{lem:ell}) in \bwd when comparing it with Peng et al.'s $\ell$ defined in  Eq. \eqref{eq:old-ell}.

\header
{\bf Varying $\tau$ in \texttt{AMC} and \texttt{GEER}.}
}
Fig. \ref{fig:param-1} and Fig. \ref{fig:param-2} plot the average running times of \fwd and \algo when varying $\tau$ from 1 to 8 on DBLP, Youtube, and Orkut datasets, when $\epsilon=0.2$ and $\epsilon=0.02$, respectively. 
% We can make the following observations. 
In Fig. \ref{fig:param-1}, the running times of both \fwd and \algo first decrease notably and then remain stable or go up slowly as $\tau$ is increased from 1 to 8 on DBLP and Youtube. For example, on Youtube, \fwd and \algo achieve the best performance when $\tau=5$ and $\tau=2$, which are $12\times$ and $1.6\times$ faster compared to their costs when $\tau=1$ respectively. The result shows that a reasonable $\tau$ is conducive to averting substantial overheads incurred by random walks and facilitating early termination in \fwd, which validates the effectiveness of our adaptive sampling scheme described in Section \ref{sec:amc}. We can also make similar observations from Fig. \ref{fig:time-param-dblp-2} and \ref{fig:time-param-youtube-2}, where the query performance of \algo declines to a low level after $\tau$ is over $3$ and the running time of \fwd is drastically reduced when increasing $\tau$ from $1$ to $8$.  This is expected, since a smaller error $\epsilon=0.02$ will lead to immense random walks in \fwd and hence more batches of random walks are needed to determine the termination point. On the Orkut dataset, the running time of \algo grows with $\tau$ as observed in Fig. \ref{fig:time-param-orkut-1} and Fig. \ref{fig:time-param-orkut-2}, whereas the query efficiency of \fwd drops with the increase of $\tau$ when $\epsilon=0.2$ and rises if $\epsilon=0.02$. The reason is that Orkut has a high average degree, and thus by Eq. \eqref{eq:time-amc} the sampling costs in \fwd and \algo are not significant. In such cases, multiple batches of sampling instead bring on additional costs. But when $\epsilon$ is small (e.g., 0.02), \fwd needs numerous random walks and the adaptive sampling scheme still works. To sum up, choosing $\tau=5$ in \fwd and \algo is able to obtain appealing performance in most cases.
% , which is able to , as displayed in Fig. \ref{fig:param-1} and Fig. \ref{fig:param-2}.

\begin{figure}[!t]
\centering
\begin{small}
\begin{tikzpicture}[every mark/.append style={mark size=3pt}]
    \begin{customlegend}[legend columns=2,
        legend entries={\algo,\fwd},
        legend style={at={(0.5,1.05)},anchor=north,draw=none,font=\small,column sep=0.25cm}]
    \addlegendimage{line width=0.2mm,mark=square,color=my_blue}
    \addlegendimage{line width=0.2mm,mark=o,color=my_cyan}
    \end{customlegend}
\end{tikzpicture}
\\[-\lineskip]
\vspace{-3mm}

\hspace{-4mm}\subfloat[DBLP]{
\begin{tikzpicture}[scale=1,every mark/.append style={mark size=2pt}]
    \begin{axis}[
        height=\columnwidth/3.4,
        width=\columnwidth/2.6,
        ylabel={\em running time} (ms),
        xmin=0.5, xmax=8.5,
        ymin=10, ymax=10000,
        xtick={1,2,3,4,5,6,7,8},
        xticklabel style = {font=\footnotesize},
        yticklabel style = {font=\footnotesize},
        xticklabels={1,2,3,4,5,6,7,8},
        ymode=log,
        log basis y={10},
        % every axis y label/.style={font=\small,at={(current axis.north)},right=10mm,above=0mm},
        every axis y label/.style={font=\small,at={{(0.28,1.0)}},right=10mm,above=0mm},
        % ylabel near ticks,
        % y label style = {font=\small, rotate=-90,anchor=south},
        y label style = {font=\footnotesize},
        xlabel near ticks,
        x label style = {font=\tiny},
        legend style={fill=none,font=\tiny,at={(0.6,0.99)},anchor=north west,draw=none},
    ]
    \addplot[line width=0.2mm,mark=square,color=my_blue]  
        plot coordinates {
(1,	166.546	)
(2,	109.803	)
(3,	86.9406	)
(4,	77.432	)
(5,	93.4225	)
(6, 93.5719 )
(7, 96.0409 )
(8, 99.3455 )
    };
    \addplot[line width=0.2mm,mark=o,color=my_cyan]  
        plot coordinates {
(1,	2365	)
(2,	1468.08	)
(3,	870.508	)
(4,	630.548	)
(5,	542.775	)
(6, 564.243 )
(7, 592.755 )
(8, 645.431 )
    };
    
    % \legend{\algo,\fwd}
    \end{axis}
\end{tikzpicture}\hspace{0mm}\label{fig:time-param-dblp-1}%
}
\hspace{0mm}\subfloat[Youtube]{
\begin{tikzpicture}[scale=1,every mark/.append style={mark size=2pt}]
    \begin{axis}[
        height=\columnwidth/3.4,
        width=\columnwidth/2.6,
        ylabel={\em running time} (ms),
        xmin=0.5, xmax=8.5,
        ymin=10, ymax=1000,
        xtick={1,2,3,4,5,6,7,8},
        xticklabel style = {font=\footnotesize},
        yticklabel style = {font=\footnotesize},
        xticklabels={1,2,3,4,5,6,7,8},
        ymode=log,
        log basis y={10},
        % every axis y label/.style={font=\small,at={(current axis.north)},right=10mm,above=0mm},
        every axis y label/.style={font=\small,at={{(0.28,1.0)}},right=10mm,above=0mm},
        % ylabel near ticks,
        % y label style = {font=\small, rotate=-90,anchor=south},
        y label style = {font=\footnotesize},
        xlabel near ticks,
        x label style = {font=\tiny},
        legend style={fill=none,font=\tiny,at={(0.6,0.99)},anchor=north west,draw=none},
    ]
    \addplot[line width=0.2mm,mark=square,color=my_blue]  
        plot coordinates {
(1,	17.8531	)
(2,	11.0977	)
(3,	16.4163	)
(4,	18.9163	)
(5,	20.8352	)
(6, 21.894 )
(7, 23.6264 )
(8, 24.2496 )
    };
    \addplot[line width=0.2mm,mark=o,color=my_cyan]  
        plot coordinates {
(1,	981.452	)
(2,	524.499	)
(3,	273.617	)
(4,	135.982	)
(5,	82.0696	)
(6, 99.3981 )
(7, 105.9715 )
(8, 116.3292 )
    };
    
    % \legend{\algo,\fwd}
    \end{axis}
\end{tikzpicture}\hspace{0mm}\label{fig:time-param-youtube-1}%
}
\hspace{0mm}\subfloat[Orkut]{
\begin{tikzpicture}[scale=1,every mark/.append style={mark size=2pt}]
    \begin{axis}[
        height=\columnwidth/3.4,
        width=\columnwidth/2.6,
        ylabel={\em running time} (ms),
        xmin=0.5, xmax=8.5,
        ymin=1, ymax=100,
        xtick={1,2,3,4,5,6,7,8},
        xticklabel style = {font=\footnotesize},
        yticklabel style = {font=\footnotesize},
        xticklabels={1,2,3,4,5,6,7,8},
        ymode=log,
        log basis y={10},
        % every axis y label/.style={font=\small,at={(current axis.north)},right=0mm,above=0mm},
        every axis y label/.style={font=\small,at={{(0.28,1.0)}},right=10mm,above=0mm},
        % ylabel near ticks,
        % y label style = {font=\footnotesize},
        y label style = {font=\footnotesize},
        xlabel near ticks,
        x label style = {font=\tiny},
        legend style={fill=none,font=\tiny,at={(0.6,0.35)},anchor=north west,draw=none},
    ]
    \addplot[line width=0.2mm,mark=square,color=my_blue]  
        plot coordinates {
(1,	3.29701	)
(2,	4.69091	)
(3,	5.36482	)
(4,	5.97706	)
(5,	6.23853	)
(6, 6.2691 )
(7, 6.9538 )
(8, 9.02345 )
    };
    \addplot[line width=0.2mm,mark=o,color=my_cyan]  
        plot coordinates {
(1,	44.286	)
(2,	45.8319	)
(3,	61.0354	)
(4,	65.544	)
(5,	74.0523	)
(6, 72.2784 )
(7, 78.0541 )
(8, 78.5258 )
    };
    
    % \legend{\algo,\fwd}
    \end{axis}
\end{tikzpicture}\label{fig:time-param-orkut-1}%
}%
\end{small}
\vspace{-2mm}
\caption{Varying $\tau$ when $\epsilon=0.2$.} \label{fig:param-1}
%\vspace{-2mm}
\end{figure}
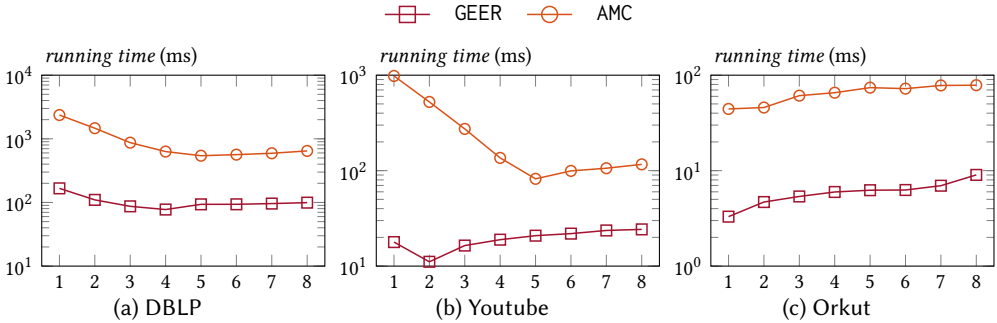

\begin{figure}[!t]
\centering
\begin{small}
\begin{tikzpicture}[every mark/.append style={mark size=3pt}]
    \begin{customlegend}[legend columns=2,
        legend entries={\algo,\fwd},
        legend style={at={(0.5,1.05)},anchor=north,draw=none,font=\small,column sep=0.25cm}]
    \addlegendimage{line width=0.2mm,mark=square,color=my_blue}
    \addlegendimage{line width=0.2mm,mark=o,color=my_cyan}
    \end{customlegend}
\end{tikzpicture}
\\[-\lineskip]
\vspace{-3mm}

\hspace{-4mm}
\subfloat[DBLP]{
\begin{tikzpicture}[scale=1,every mark/.append style={mark size=2pt}]
    \begin{axis}[
        height=\columnwidth/3.4,
        width=\columnwidth/2.6,
        ylabel={\em running time} (ms),
        xmin=0.5, xmax=8.5,
        ymin=5000, ymax=1000000,
        xtick={1,2,3,4,5,6,7,8},
        xticklabel style = {font=\footnotesize},
        yticklabel style = {font=\footnotesize},
        xticklabels={1,2,3,4,5,6,7,8},
        ymode=log,
        log basis y={10},
        % every axis y label/.style={font=\small,at={(current axis.north)},right=10mm,above=0mm},
        every axis y label/.style={font=\small,at={{(0.28,1.0)}},right=10mm,above=0mm},
        % ylabel near ticks,
        % y label style = {font=\small, rotate=-90,anchor=south},
        y label style = {font=\footnotesize},
        xlabel near ticks,
        x label style = {font=\tiny},
        legend style={fill=none,font=\tiny,at={(0.6,0.99)},anchor=north west,draw=none},
    ]
    \addplot[line width=0.2mm,mark=square,color=my_blue]  
        plot coordinates {
(1,	11581.3	)
(2,	7620.17	)
(3,	6345.95	)
(4,	5671.06	)
(5,	5770.15	)
(6, 6065.77 )
(7, 6186.22 )
(8, 6202.86 )
    };
    \addplot[line width=0.2mm,mark=o,color=my_cyan]  
        plot coordinates {
(1,	842764	)
(2,	414459	)
(3,	215712	)
(4,	125308	)
(5,	80798	)
(6, 59628.7 )
(7, 48757.9 )
(8, 45855.7 )
    };
    % \legend{\algo,\fwd}
    \end{axis}
\end{tikzpicture}\hspace{0mm}\label{fig:time-param-dblp-2}%
}
\subfloat[Youtube]{
\begin{tikzpicture}[scale=1,every mark/.append style={mark size=2pt}]
    \begin{axis}[
        height=\columnwidth/3.4,
        width=\columnwidth/2.6,
        ylabel={\em running time} (ms),
        xmin=0.5, xmax=8.5,
        ymin=1000, ymax=1000000,
        xtick={1,2,3,4,5,6,7,8},
        xticklabel style = {font=\footnotesize},
        yticklabel style = {font=\footnotesize},
        xticklabels={1,2,3,4,5,6,7,8},
        ymode=log,
        log basis y={10},
        % every axis y label/.style={font=\small,at={(current axis.north)},right=10mm,above=0mm},
        every axis y label/.style={font=\small,at={{(0.28,1.0)}},right=10mm,above=0mm},
        % ylabel near ticks,
        % y label style = {font=\small, rotate=-90,anchor=south},
        y label style = {font=\footnotesize},
        xlabel near ticks,
        x label style = {font=\tiny},
        legend style={fill=none,font=\tiny,at={(0.6,0.99)},anchor=north west,draw=none},
    ]
    \addplot[line width=0.2mm,mark=square,color=my_blue]  
        plot coordinates {
(1,	3595.63	)
(2,	2424.26	)
(3,	1779.06	)
(4,	1473.39	)
(5,	1663.92	)
(6,	1772.35	)
(7, 1856.07 )
(8, 1888.36 )
    };
    \addplot[line width=0.2mm,mark=o,color=my_cyan]  
        plot coordinates {
(1,	328779	)
(2,	192735	)
(3,	102353	)
(4,	53476.4	)
(5,	27233.4	)
(6, 13906.1 )
(7, 12550.3 )
(8, 11957.2 )
    };
    % \legend{\algo,\fwd}
    \end{axis}
\end{tikzpicture}\hspace{0mm}\label{fig:time-param-youtube-2}%
}
\subfloat[Orkut]{
\begin{tikzpicture}[scale=1,every mark/.append style={mark size=2pt}]
    \begin{axis}[
        height=\columnwidth/3.4,
        width=\columnwidth/2.6,
        ylabel={\em running time} (ms),
        xmin=0.5, xmax=8.5,
        ymin=1000, ymax=100000,
        xtick={1,2,3,4,5,6,7,8},
        xticklabel style = {font=\footnotesize},
        yticklabel style = {font=\footnotesize},
        xticklabels={1,2,3,4,5,6,7,8},
        ymode=log,
        log basis y={10},
        % every axis y label/.style={font=\small,at={(current axis.north)},right=0mm,above=0mm},
        every axis y label/.style={font=\small,at={{(0.28,1.0)}},right=10mm,above=0mm},
        % ylabel near ticks,
        y label style = {font=\footnotesize},
        xlabel near ticks,
        x label style = {font=\tiny},
        legend style={fill=none,font=\tiny,at={(0.6,0.99)},anchor=north west,draw=none},
    ]
    \addplot[line width=0.2mm,mark=square,color=my_blue]  
        plot coordinates {
(1,	1702.66	)
(2,	2831.43	)
(3,	2745.95	)
(4,	3266.83	)
(5,	3648.09	)
(6,	3980.04	)
(7,	4025.54	)
(8, 4461.56  )
    };
    \addplot[line width=0.2mm,mark=o,color=my_cyan]  
        plot coordinates {
(1,	34805.2	)
(2,	21537.4	)
(3,	18651.9	)
(4,	17517.6	)
(5,	16826.7	)
(6, 15398.9 )
(7, 14015.4 )
(8, 14022.7 )
    };
    % \legend{\algo,\fwd}
    \end{axis}
\end{tikzpicture}\label{fig:time-param-orkut-2}%
}%
\end{small}
\vspace{-3mm}
\caption{Varying $\tau$ when $\epsilon=0.02$.} \label{fig:param-2}
\vspace{-2mm}
\end{figure}
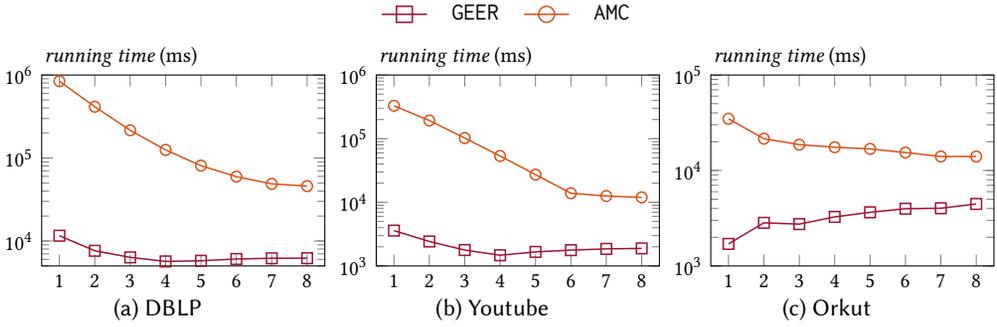

\revise{
\header
{\bf Varying $\ell_b$ in \texttt{GEER}.} First, we denote by $\ell^{\ast}_b$ the value of $\ell_b$ determined by Eq. \eqref{eq:smm-amc} (Line 9 in Algorithm \ref{alg:bere}). Then, we remove the requirement for Eq. \eqref{eq:smm-amc} at Line 9 in Algorithm \ref{alg:bere} and vary $\ell_b$ in the range $\{\ell^{\ast}_b-6, \ell^{\ast}_b-4, \ell^{\ast}_b-2, \ell^{\ast}_b, \ell^{\ast}_b+2, \ell^{\ast}_b+4, \ell^{\ast}_b+6\}$. Fig. \ref{fig:param-ellb} displays the average query time of \algo on four datasets (i.e., Facebook, DBLP, LiveJournal, and Orkut) when $\epsilon$ is varied from $0.01$ to $0.2$ and $\ell_b$ is varied in the above-said range. At a glance, it can be seen that \algo exhibits the best performance when $\ell_b$ is set to $\ell^{\ast}_b$ or nearly $\ell^{\ast}_b$ on all four datasets. More specifically, the running time of \algo first reduces considerably when increasing $\ell_b$ from $\ell^{\ast}_b-6$ to $\ell^{\ast}_b$ and then goes up significantly as $\ell_b$ keeps increasing. For instance, on Orkut, \algo with $\ell_b=\ell^{\ast}_b\pm 6$ is orders of magnitude slower than \algo with $\ell_b=\ell^{\ast}_b$. 
The reasons for the remarkable changes are as follows. When $\ell_b$ is less than $\ell^{\ast}_b$, the intermediate result $\vec{\mathbf{s}}^\ast$ and $\vec{\mathbf{t}}^\ast$ generated by \bwd are large, and hence, \algo degrades to \fwd. On the other hand, when $\ell_b$ is much larger than $\ell^{\ast}_b$, \algo conducts vast matrix-vector multiplications that render \algo even more costly than \fwd. The result shows that our greedy strategy for determining $\ell_b$ strikes a good balance between \bwd and \fwd, and enables \algo get the lowest (or nearly lowest) computation cost.

\header
{\bf Our $\ell$ vs. Peng et al.'s $\ell$.} In the last set of experiments, we evaluate the efficiency performance of \bwd (Algorithm \ref{alg:smm}) when its input parameter $\ell_b$ is set to our refined $\ell$ (Eq. \eqref{eq:ell}) and Peng et al.'s $\ell$ (Eq. \eqref{eq:old-ell}), respectively. Fig. \ref{fig:time-ell-rand} reports the average query time of \bwd with two types of $\ell$ on five datasets including Facebook (FB), DBLP (DB), Youtube (YT), Orkut (OK), and LiveJournal (LJ). When $\epsilon=0.5$ (Fig. \ref{fig:time-ell-rand}(a)), \bwd with our $\ell$ is $3.3\times$ and $6.7\times$ faster than \bwd with Peng et al.'s $\ell$ on high-degree graphs FB and OK, respectively. On graphs with low average degrees (i.e., DB, YT, and LJ), our $\ell$ yields up to $2.1\times$ performance improvements over Peng et al.'s $\ell$ in \bwd. This phenomenon is due to that our $\ell$ in Eq. \eqref{eq:ell} is inversely correlated with the degrees of query nodes, meaning that our $\ell$ will be much smaller than Peng et al.'s $\ell$ on graphs with large average node degrees. Additionally, when $\epsilon=0.05$, our $\ell$ still obtains $2.4\times$ and $3.8\times$ speedups on FB and OK, respectively, and $1.2$-$2\times$ improvements on the remaining datasets. In sum, our refined $\ell$ can bring considerable efficiency enhancements to \bwd compared to Peng et al.'s $\ell$, especially on graphs with high average degrees.
}

\begin{figure}[!t]
\centering
\begin{small}
\begin{tikzpicture}[every mark/.append style={mark size=3pt}]
    \begin{customlegend}[legend columns=4,
        legend entries={Facebook,DBLP,LiveJournal,Orkut},
        legend style={at={(0.5,1.05)},anchor=north,draw=none,font=\small,column sep=0.2cm}]
    \addlegendimage{line width=0.2mm,mark=square,color=my_teal}
    \addlegendimage{line width=0.2mm,mark=o,color=myblue}
    \addlegendimage{line width=0.2mm,mark=triangle,color=my_cyan}
    \addlegendimage{line width=0.2mm,mark=diamond,color=my_green}
    \end{customlegend}
\end{tikzpicture}
\\[-\lineskip]
\vspace{-3mm}

\hspace{-3mm}\subfloat[$\epsilon=0.2$]{
\begin{tikzpicture}[scale=1,every mark/.append style={mark size=2pt}]
    \begin{axis}[
        height=\columnwidth/3.4,
        width=\columnwidth/2.6,
        ylabel={\em running time} (ms),
        xlabel={$\ell^{\ast}_b \pm x$},
        xmin=0.5, xmax=7.5,
        ymin=1, ymax=1000000,
        xtick={1,2,3,4,5,6,7},
        xticklabel style = {font=\tiny},
        yticklabel style = {font=\tiny},
        xticklabels={-6,-4,-2,0,+2,+4,+6},
        ytick={1,10,100,1000,10000,100000,1000000},
        ymode=log,
        log basis y={10},
        % every axis y label/.style={font=\small,at={(current axis.north)},right=10mm,above=0mm},
        every axis y label/.style={font=\small,at={{(0.28,1.0)}},right=10mm,above=0mm},
        % ylabel near ticks,
        % y label style = {font=\small, rotate=-90,anchor=south},
        y label style = {font=\footnotesize},
        xlabel near ticks,
        x label style = {font=\small},
        legend style={fill=none,font=\tiny,at={(0.6,0.99)},anchor=north west,draw=none},
    ]
    \addplot[line width=0.2mm,mark=square,color=my_teal]  %Facebook
        plot coordinates {
(1,	9.71792	)
(2,	9.66599	)
(3,	5.58007	)
(4,	2.868676	)
(5, 10.9929 )
(6, 49.0429 )
(7, 98.2548 )
    };    

    \addplot[line width=0.2mm,mark=o,color=myblue]  %DBLP
        plot coordinates {
(1,	317.097	)
(2,	293.582	)
(3,	44.0066	)
(4,	24.5735	)
(5, 516.738 )
(6, 3098.19 )
(7, 6435.86 )
    };

    \addplot[line width=0.2mm,mark=triangle,color=my_cyan]  %LiverJournal
        plot coordinates {
(1,	519.227	)
(2,	490.295	)
(3,	104.207	)
(4,	83.31	)
(5, 3321.55 )
(6, 34878.8 )
(7, 83242.2 )
    };

    \addplot[line width=0.2mm,mark=diamond,color=my_green]  %Orkut
        plot coordinates {
(1,	61.4098	)
(2,	51.061	)
(3,	21.0263	)
(4,	6.25676	)
(5, 4510.69 )
(6, 207809 )
(7, 0.88943e+06 )
    };
    
    % \legend{\algo,\fwd}
    \end{axis}
\end{tikzpicture}\hspace{0mm}\label{fig:time-param-ell-eps1}%
}
\hspace{0mm}\subfloat[$\epsilon=0.05$]{
\begin{tikzpicture}[scale=1,every mark/.append style={mark size=2pt}]
    \begin{axis}[
        height=\columnwidth/3.4,
        width=\columnwidth/2.6,
        ylabel={\em running time} (ms),
        xlabel={$\ell^{\ast}_b\pm x$},
        xmin=0.5, xmax=7.5,
        ymin=10, ymax=1000000,
        xtick={1,2,3,4,5,6,7},
        xticklabel style = {font=\tiny},
        yticklabel style = {font=\tiny},
        xticklabels={-6,-4,-2,0,+2,+4,+6},
        ytick={10,100,1000,10000,100000,1000000},
        ymode=log,
        log basis y={10},
        % every axis y label/.style={font=\small,at={(current axis.north)},right=10mm,above=0mm},
        every axis y label/.style={font=\small,at={{(0.28,1.0)}},right=10mm,above=0mm},
        % ylabel near ticks,
        % y label style = {font=\small, rotate=-90,anchor=south},
        y label style = {font=\footnotesize},
        xlabel near ticks,
        x label style = {font=\small},
        legend style={fill=none,font=\tiny,at={(0.6,0.99)},anchor=north west,draw=none},
    ]
    \addplot[line width=0.2mm,mark=square,color=my_teal]  %Facebook
        plot coordinates {
(1,	252.181	)
(2,	97.3654	)
(3,	39.3479	)
(4,	14.1852	)
(5, 31.0329 )
(6, 74.5287 )
(7, 124.093 )
    };    

    \addplot[line width=0.2mm,mark=o,color=myblue]  %DBLP
        plot coordinates {
(1,	15021.6	)
(2,	9269.97	)
(3,	1746.33	)
(4,	992.556	)
(5, 2807.88 )
(6, 5590.9 )
(7, 9280.85 )
    };

    \addplot[line width=0.2mm,mark=triangle,color=my_cyan]  %LiverJournal
        plot coordinates {
(1,	18782	)
(2,	13306.9	)
(3,	2470.92	)
(4,	997.63	)
(5, 10434.8 )
(6, 56770.3 )
(7, 108147 )
    };

    \addplot[line width=0.2mm,mark=diamond,color=my_green]  %Orkut
        plot coordinates {
(1,	2521.79	)
(2,	2480.66	)
(3,	2254.79	)
(4,	220.8493	)
(5, 21465 )
(6, 498653 )
(7, 0 )
    };

    % \legend{\algo,\fwd}
    \end{axis}
\end{tikzpicture}\hspace{0mm}\label{fig:time-param-ell-eps2}%
}
\hspace{0mm}\subfloat[$\epsilon=0.01$]{
\begin{tikzpicture}[scale=1,every mark/.append style={mark size=2pt}]
    \begin{axis}[
        height=\columnwidth/3.4,
        width=\columnwidth/2.6,
        ylabel={\em running time} (ms),
        xlabel={${\ell^{\ast}_b} \pm x$},
        xmin=0.5, xmax=7.5,
        ymin=100, ymax=1000000,
        xtick={1,2,3,4,5,6,7},
        xticklabel style = {font=\tiny},
        yticklabel style = {font=\tiny},
        xticklabels={-6,-4,-2,0,+2,+4,+6},
        ytick={10,100,1000,10000,100000,1000000},
        ymode=log,
        log basis y={10},
        % every axis y label/.style={font=\small,at={(current axis.north)},right=0mm,above=0mm},
        every axis y label/.style={font=\small,at={{(0.28,1.0)}},right=10mm,above=0mm},
        % ylabel near ticks,
        % y label style = {font=\footnotesize},
        y label style = {font=\footnotesize},
        xlabel near ticks,
        x label style = {font=\small},
        legend style={fill=none,font=\tiny,at={(0.6,0.35)},anchor=north west,draw=none},
    ]

    \addplot[line width=0.2mm,mark=square,color=my_teal]  %Facebook
        plot coordinates {
(1,	1814.72	)
(2,	1649.33	)
(3,	171.543	)
(4,	131.512	)
(5, 171.429 )
(6, 318.192 )
(7, 468.922 )
    };    

    \addplot[line width=0.2mm,mark=o,color=myblue]  %DBLP
        plot coordinates {
(1,	93234.1	)
(2,	24456.3	)
(3,	11485.8	)
(4,	9235	)
(5, 10664.1 )
(6, 14825 )
(7, 18136.9 )
    };

    \addplot[line width=0.2mm,mark=triangle,color=my_cyan]  %LiverJournal
        plot coordinates {
(1,	747869	)
(2,	182423	)
(3,	23087.9	)
(4,	27367.6)
(5, 60958 )
(6, 99902.6 )
(7, 149280 )
    };
    
    \addplot[line width=0.2mm,mark=diamond,color=my_green]  %Orkut
        plot coordinates {
(1,	156869	)
(2,	103736	)
(3,	47607	)
(4,	8991.46	)
(5, 44713.8 )
(6, 284130 )
(7, 0.88842e+06 )
    };

    % \legend{\algo,\fwd}
    \end{axis}
\end{tikzpicture}\label{fig:time-param-ell-eps3}%
}%
\end{small}
\vspace{-2mm}
\caption{Varying $\ell_b$ in \algo} \label{fig:param-ellb}
%\vspace{-2mm}
\end{figure}
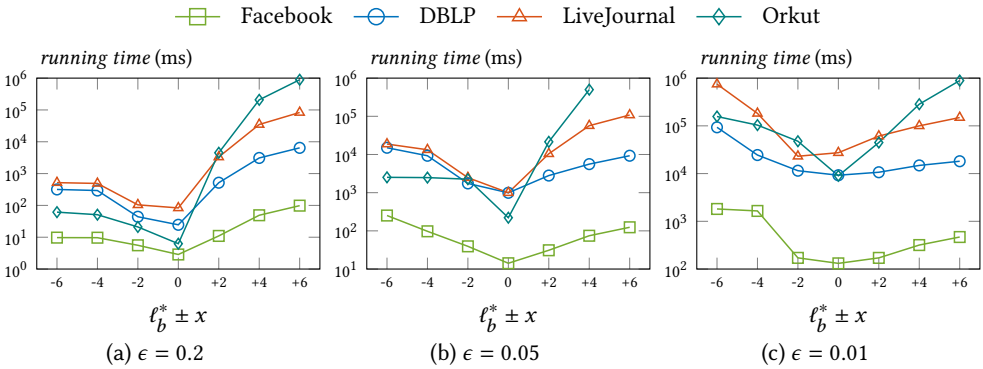

\begin{figure}[!t]
\centering
\begin{small}
\begin{tikzpicture}
    \begin{customlegend}[legend columns=2,
        legend entries={Our $\ell$,Peng et al.'s $\ell$},
        legend columns=-1,
        area legend,
        legend style={at={(0.55,1.15)},anchor=north,draw=none,font=\scriptsize,column sep=0.3cm}]
        % \addlegendimage{,pattern=horizontal lines}
        \addlegendimage{black, fill=yellow, postaction={ pattern=north east lines }}
        \addlegendimage{black, fill=my_purple, postaction={pattern=grid}}
    \end{customlegend}
\end{tikzpicture}
\\[-\lineskip]
\vspace{-3mm}

\hspace{-5mm}\subfloat[{\em $\epsilon=0.5$}]{
\begin{tikzpicture}[scale=1]
\begin{axis}[
    height=\columnwidth/3.4,
    width=\columnwidth/2.4,
    ybar=0.1pt,
    bar width=0.2cm,
    enlarge x limits=true,
    ylabel={\em running time} (ms),
    symbolic x coords={2,3,4,5,6},
    xticklabels={FB,DB,YT,OK,LJ},
    xtick=data,
    ymin=100,
    ymax=20000000,
    ytick={100,1000,10000,100000,1000000,10000000},
    ymode=log,
    y label style = {font=\footnotesize},
    xlabel near ticks,
    x label style = {font=\small},
    yticklabel style = {font=\tiny},
    log basis y={10},
    every axis y label/.style={font=\small,at={{(0.25,1.0)}},right=10mm,above=0mm},
        label style={font=\scriptsize},
        tick label style={font=\scriptsize},
        % x tick label style={rotate=0,anchor=east}
    ]
\addplot [black, fill=yellow, postaction={ pattern=north east lines }] coordinates {(2,392.83) (3,40157.3) (4,26122.3) (5,2.02927e+06) (6,1.0636e+06)};
\addplot [black, fill=my_purple, postaction={pattern=grid}] coordinates {(2,1294.969) (3,61762.9) (4,53301.2) (5,1.35493e+07) (6,2208993)};
%\legend{\pmsc, \msc}
\end{axis}
\end{tikzpicture}\hspace{4mm}%
}%
\subfloat[{$\epsilon=0.05$}]{
\begin{tikzpicture}[scale=1]
\begin{axis}[
    height=\columnwidth/3.4,
    width=\columnwidth/2.4,
    ybar=0.1pt,
    bar width=0.2cm,
    enlarge x limits=true,
    ylabel={\em running time} (ms),
    symbolic x coords={2,3,4,5,6},
    xticklabels={FB,DB,YT,OK,LJ},
    xtick=data,
    yticklabel style = {font=\tiny},
    ymin=100,
    ymax=100000000,
    ytick={100,1000,10000,100000,1000000,10000000.100000000},
    ymode=log,
    log basis y={10},
    y label style = {font=\footnotesize},
    xlabel near ticks,
    x label style = {font=\small},
    % every axis y label/.style={at={(current axis.north west)},right=4mm,above=0mm},
    every axis y label/.style={font=\small,at={{(0.25,1.0)}},right=10mm,above=0mm},
        label style={font=\scriptsize},
        tick label style={font=\scriptsize},
        % x tick label style={rotate=0,anchor=east}
    ]
\addplot [black, fill=yellow, postaction={ pattern=north east lines }] coordinates {(2,1055.03) (3,81107.8) (4,47314.1) (5,9.94211e+06) (6,2.21778e+06) };
\addplot [black, fill=my_purple, postaction={pattern=grid}] coordinates {(2,2477.85) (3,94367.4) (4,92530.1) (5,3.76447e+07) (6,4.35172e+06)};
%\legend{\pmsc, \msc}
\end{axis}
\end{tikzpicture}\hspace{0mm}%
}%
\vspace{-2mm}
\end{small}
\caption{Our $\ell$ vs. Peng et al.'s $\ell$ in \bwd} \label{fig:time-ell-rand}
%\vspace{-2mm}
\end{figure}
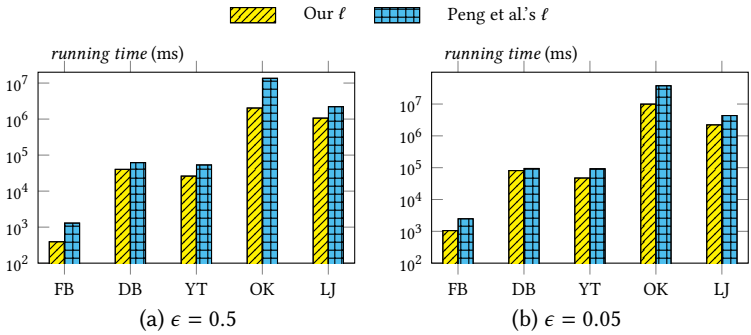

%% file: tex/relatedwork.tex
\section{Additional Related Work}\label{sec:related}
\revise{In this section, we review a number of studies germane to our work.}
\subsection{Effective Resistance}
Apart from the methods discussed in Section \ref{sec:mcpt} for $\epsilon$-approximate PER queries, there exist several techniques for all pairwise ER computation or estimating ER values for all edges in the input graph, as reviewed in the sequel. Fouss et al. \cite{fouss2007random} proposed to calculate the ER value for any node pair in the input graph $G$ by first computing the Moore-Penrose pseudoinverse of the Laplacian matrix $\DM-\AM$ based on a sparse Cholesky
factorization of $\DM-\AM$.
\citet{mavroforakis2015spanning} utilized the random projection and SDD solvers to approximate the ER values of all edges.
In Ref. \cite{hayashi2016efficient}, Hayashi et al. presented a method for computing approximate ER values of all edges in the input graph by sampling spanning trees. This method is still state-of-the-art for this problem. Jambulapati and Sidford then \cite{jambulapati2018efficient} introduced an algorithm that returns $\epsilon$-approximates ER values for all possible node pairs in the graph using $O(n^2/\epsilon)$ time. Their method is based on computing the sketches of the Laplacian matrix and its pseudoinverse. After that, \cite{goranci2018dynamic} studied how to dynamically maintain all-pairs approximate ER values in separable graphs. In a recent work \cite{broberg2020prediction}, based on prior work on graph sparsifiers \cite{abraham2016fully}, the authors presented an algorithm for the offline fully dynamic all-pairs ER problem, and further extended this idea to create data structures to maintain the approximate all-pairs ER values in fully dynamic graphs in a prediction assisted setting. \cite{aybat2017decentralized} developed an efficient linearly convergent distributed algorithm for ER computation. The proposed algorithm is further applied to the consensus problem, which aims to compute the average of node values in a distributed manner. Since this paper focuses on handling $\epsilon$-approximate PER queries, the problems studied in the aforementioned work are beyond the scope of this paper. 
Given a graph $G$, a {\em spanning tree} $T(G)$ is a subgraph of $G$, which includes all the nodes of $G$ with a minimum possible number of edges \cite{cheriton1976finding}. Corollary 4.2 in \cite{lovasz1993random} proves that $r(s,t)=\frac{|T(G^{\prime})|}{|T(G)|}$, where $G^{\prime}$ be a subgraph of $G$ obtained by identifying $s$ and $t$.
Besides the methods remarked in Section \ref{sec:mcpt}, \citet{pankdd21} also introduced an algorithm $\mathtt{ST}$ for estimating ER by employing the local algorithm presented in \cite{lyons2018sharp} to estimate $|T(G^{\prime})|$ and $|T(G)|$. We do not compare with $\mathtt{ST}$ since it returns $r(s,t)$ satisfying the relative approximation guarantee rather than the $\epsilon$-approximation. As pointed out in \cite{pankdd21}, $\mathtt{ST}$ does not perform well in practice.

\subsection{Commute Time}
By definition, ER is closely pertinent to {\em commute time} \cite{motwani1995randomized}; hence we review that are relevant to estimating commute times in graphs in the sequel. Sarkar and Moore \cite{sarkar2007tractable} proposed $\mathtt{GRANCH}$, which computes the {\em truncated commute times} of all nodes to a target node by decomposing the input graph into overlapping neighborhoods for each node. In the subsequent work \cite{sarkar2008fast}, they combined sampling techniques with the branch and bound pruning scheme to obtain upper and lower bounds on commute times from a node, without iterating over the entire graph. Ref. \cite{bonchi2012fast} described an algorithm estimating column-wise and pair-wise commute times based on a variation of the conjugate gradient algorithm. In \cite{cohen2016faster}, the authors showed how to build a $O\left(\frac{n\log{n}}{\epsilon^2}\right)$-size data structure in time $O\left(nm^{3/4}+n^{2/3}m\right)$, which can output a $1+\epsilon$ multiplicative approximation to the commute time of any given pair of nodes. \cite{khoa2019incremental} presented a method that incrementally estimates the commute time to facilitate online applications using properties of random walks and hitting time on graphs. 

\subsection{Personalized PageRank}
% \renchi{
Another line of research that is germane to this work is {\em personalized PageRank} (PPR) \cite{jeh2003scaling,haveliwala2002topic} and its variant {\em heat kernel PageRank} (HKPR) \cite{chung2007heat,kloster2014heat}. In the past years, PPR and HKPR have been extensively studied in the literature, as surveyed in \cite{park2019survey}. Amid them, a number of recent studies \cite{lofgren2016personalized,lofgren2015bidirectional,wang2017fora,wang2016hubppr,wu2021unifying,yang2019efficient,banerjee2015fast,shi2019realtime,wei2018topppr,wang2021approx,wang2019efficient} propose to combine random walks with restart \cite{tong2006fast,fogaras2005towards} with  push-operation-based deterministic graph traversal \cite{andersen2006local,andersen2007local} for improved efficiency in processing single-source or single-pair PPR/HKPR queries.
At first glance, it seems that we can
simply adapt and extend these techniques to $\epsilon$-approximate PER computation. Unfortunately, contrary to PPR and HKPR, ER is inherently more sophisticated than PPR/HKPR as they differ significantly in the ways of defining random walks. In particular, the random walk model used in PPR/HKPR is {\em random walk with restart} \cite{tong2006fast}, which would stop at each visited node on the walk with a certain probability according to geometric or Poisson distributions, whereas ER relies on simple random walks of various fixed lengths (from $1$ to $\infty$). With a stopping probability, the random walks with restart and push-operation-based graph traversal in PPR/HKPR computation are in turn more likely to terminate after a few steps. As for ER, a linchpin to facilitate ER computation is a tight lower bound for the maximum random walk length. Consequently, it is more challenging to enable the combination of graph traversal and random walks for enhanced efficiency in ER computation. 

In addition to the above similarity/dissimilarity measures, there also exist a plethora of other random walk-based measures \cite{yang2021fast,shi2020realtime,yang2022efficient} akin to ER, as surveyed in \cite{lovasz1993random,yen2008family}.

%% file: tex/conclusion.tex
\section{Conclusion}\label{sec:conclude}
This paper presents two novel probabilistic algorithms, \fwd and \algo, for answering $\epsilon$-approximate PER queries efficiently. \fwd improves over the state-of-the-art approach by an optimized bound for the maximum random walk length and an adaptive sampling strategy. Particularly, \fwd provides a superior time complexity while retaining the accuracy assurance. For the purpose of practical efficiency, we further develop \algo, which is built on the idea of integrating deterministic graph traversal into \fwd. GEER offers significantly enhanced empirical efficiency over \fwd without degrading its theoretical guarantees. Such empirical efficiency is unmatched by any existing solutions for $\epsilon$-approximate PER computation. Our experimental results show that \algo consistently outperforms the state of the art in terms of efficiency, and is often orders of magnitude faster.

%% file: tex/appendix.tex
\section{Proofs}\label{sec:proof}

\subsection{\bf Proof of Theorem \ref{lem:ell}}
We first need the following lemma.
\begin{lemma}\label{lem:pist}
Given an undirected graph $G$, for any two nodes $u,v\in V$ and any integer $i\ge 1$, the matrix $\mathbf{P}$ can be decomposed as
% \begin{small}
\begin{align}
{p_i(u,v)}&=\sum_{k=1}^{n}{\vec{\mathbf{f}}_k(u)\cdot\vec{\mathbf{f}}_k(v)\cdot\vec{\boldsymbol{\pi}}(v)\cdot\lambda^i_k},\label{eq:fkpi} \\
\text{and}\quad \tfrac{1}{\vec{\boldsymbol{\pi}}(v)}&=\sum_{k=1}^{n}{\vec{\mathbf{f}}_k^2(v)},\label{eq:f-pi}
\end{align}
% \end{small}
where $1 = \lambda_1 \geq \lambda_2 \geq \dotsb \geq \lambda_n \geq -1$ are the eigenvalues of $\mathbf{P}$ and $\vec{\mathbf{f}}_1,\vec{\mathbf{f}}_2,\dotsc,\vec{\mathbf{f}}_n$ are the corresponding eigenvectors. The eigenvector $\vec{\mathbf{f}}_1$ is taken to be $\vec{\mathbf{1}}$.
\end{lemma}
\begin{proof}
% Let $\lambda_1,\lambda_2,\dotsc,\lambda_n$ be the ordered eigenvalues of $\mathbf{P}$ and $\vec{\mathbf{f}}_1,\vec{\mathbf{f}}_2,\dotsc,\vec{\mathbf{f}}_n$ be the corresponding eigenvectors.  
By definition, we have $\PM\vec{\mathbf{f}}_k=\lambda_k\vec{\mathbf{f}}_k\ \forall{1\le k\le n}$. Then, it is easy to derive that
% \begin{equation*}
$\PM^2\vec{\mathbf{f}}_k=\PM\lambda_k\vec{\mathbf{f}}_k=\lambda_k\PM\vec{\mathbf{f}}_k=\lambda^2_k\vec{\mathbf{f}}_k.$
% \end{equation*}
By repeating the above process, we obtain $\PM^i\vec{\mathbf{f}}_k=\lambda^i_k\vec{\mathbf{f}}_k$ $\forall{i}\ge 1$ and $\forall{1\le k\le n}$. In particular, since random walk matrix $\PM$ is row-stochastic, its largest eigenvalue $\lambda_1=1$ and the corresponding eigenvector $\vec{\mathbf{f}}_1$ can be taken to be the constant vector $\vec{\mathbf{1}}$ \cite{haveliwala2003second}.

According to \cite{wilmer2009markov}, $\vec{\mathbf{e}}_v=\sum_{k=1}^{n}{\vec{\mathbf{f}}_k(v)\cdot \vec{\boldsymbol{\pi}}(v)\cdot \vec{\mathbf{f}}_k}$.
Hence, we get
$\PM^i(u,v)=(\PM^i\vec{\mathbf{e}}_v)(u)=\left(\sum_{k=1}^{n}{\vec{\mathbf{f}}_k(v)\cdot \vec{\boldsymbol{\pi}}(v)\cdot \PM^i\vec{\mathbf{f}}_k}\right)(u).$
% \end{align*}
% \end{small}
Since $\PM^i\vec{\mathbf{f}}_k=\lambda^i_k\vec{\mathbf{f}}_k$, substituting in the above equality gives us
% \begin{small}
\begin{equation*}
\PM^i(u,v)=\left(\sum_{k=1}^{n}{\vec{\mathbf{f}}_k(v)\cdot \vec{\boldsymbol{\pi}}(v)\cdot \lambda^i_k\vec{\mathbf{f}}_k}\right)(u)
=\vec{\boldsymbol{\pi}}(v)\sum_{k=1}^{n}{\vec{\mathbf{f}}_k(u)\cdot\vec{\mathbf{f}}_k(v)\cdot \lambda^i_k},
\end{equation*}
% \end{small}
which is exactly Eq. \eqref{eq:fkpi}. In addition, 
% \begin{small}
\begin{equation*}
 1=\vec{\mathbf{e}}_v\cdot\vec{\mathbf{e}}_v^{\top}=\sum_{k=1}^{n}{\vec{\mathbf{f}}_k(v)\cdot \vec{\boldsymbol{\pi}}(v)\cdot \vec{\mathbf{f}}_k\vec{\mathbf{e}}_v^{\top}} =\sum_{k=1}^{n}{\vec{\mathbf{f}}^2_k(v)\cdot\vec{\boldsymbol{\pi}}(v)},
\end{equation*}
% \end{small}
which proves Eq. \eqref{eq:f-pi}. Therefore, the lemma is proved.
\end{proof}

Now, we are ready to prove Theorem \ref{lem:ell}.
\begin{proof}[Proof of Theorem \ref{lem:ell}]
According to Eq. \eqref{eq:f-pi} in Lemma \ref{lem:pist}, the following equations hold
% \begin{small}
\begin{equation}\label{eq:fk-fk}
% \begin{split}
 \sum_{k=1}^{n}{\vec{\mathbf{f}}_k^2(s)}=\frac{2m}{d(s)} \text{ and } \sum_{k=1}^{n}{\vec{\mathbf{f}}_k^2(t)}=\frac{2m}{d(t)}.
% \end{split}
\end{equation}
% \end{small}
Note that $1 = \lambda_1 \geq \lambda_2 \geq \dotsb \geq \lambda_n \geq -1$. Thus, by Eq.~\eqref{eq:fkpi}, we get
% \begin{small}
\begin{align}
\frac{d(t)}{2m}\sum_{k=1}^{n}{\vec{\mathbf{f}}_k(s)\cdot\vec{\mathbf{f}}_k(t)} 
&\geq \frac{d(t)}{2m}\sum_{k=1}^{n}{\vec{\mathbf{f}}_k(s)\cdot\vec{\mathbf{f}}_k(t)\cdot \lambda^2_k} = p_2(s,t)\ge 0 \label{eq:fkst}.
\end{align}
% \end{small}
Using Eq. \eqref{eq:fk-fk} and Eq. \eqref{eq:fkst} together with $\vec{\mathbf{f}}_1=\vec{\mathbf{1}}$, we have that
\begin{small}
\begin{align}
& \frac{1}{2m}\sum_{k=2}^{n}{\left(\vec{\mathbf{f}}_k(s)-\vec{\mathbf{f}}_k(t)\right)^2} = \frac{1}{2m}\sum_{k=1}^{n}{\left(\vec{\mathbf{f}}_k(s)-\vec{\mathbf{f}}_k(t)\right)^2}   = \frac{1}{2m}\sum_{k=1}^{n}{\left(\vec{\mathbf{f}}_k^2(s)+\vec{\mathbf{f}}_k^2(t)-2\vec{\mathbf{f}}_k(s)\cdot\vec{\mathbf{f}}_k(t) \right)} \le \frac{1}{{d(s)}}+\frac{1}{{d(t)}}\label{eq:fk}.
\end{align}
\end{small}
In addition, by Eq. \eqref{eq:fkpi} in Lemma \ref{lem:pist} and the fact $\vec{\mathbf{f}}_1=\vec{\mathbf{1}}$,
% \begin{small}
\begin{equation}\label{eq:pflambda}
\frac{p_i(s,s)}{d(s)}+\frac{p_i(t,t)}{d(t)}-\frac{2 p_i(s,t)}{d(t)}= \frac{1}{2m}\sum_{k=2}^{n}{(\vec{\mathbf{f}}_k(s)-\vec{\mathbf{f}}_k(t))^2\cdot\lambda^i_k}.
\end{equation}
% \end{small}
With Eq. \eqref{eq:pflambda} and Eq. \eqref{eq:fk}, we obtain
% \begin{small}
\begin{align*}
\left|r(s,t)-{r}_{\ell}(s,t)\right|& = \left|\sum_{i=\ell+1}^{\infty}{\frac{p_i(s,s)}{d(s)}+\frac{p_i(t,t)}{d(t)}-\frac{2 p_i(s,t)}{d(t)}}\right| = \left| \frac{1}{2m}\sum_{k=2}^{n}{(\vec{\mathbf{f}}_k(s)-\vec{\mathbf{f}}_k(t))^2}\sum_{i=\ell+1}^{\infty}{\lambda^i_k} \right|\\
& \leq \frac{1}{2m}\sum_{k=2}^{n}{(\vec{\mathbf{f}}_k(s)-\vec{\mathbf{f}}_k(t))^2}\sum_{i=\ell+1}^{\infty}{\lambda^i}
% & = \left| \frac{1}{2m}\sum_{k=2}^{n}{(\vec{\mathbf{f}}_k(s)-\vec{\mathbf{f}}_k(t))^2}\sum_{i=\ell+1}^{\infty}{\lambda^i_k} \right|\\
% & = \left| \frac{1}{2m}\sum_{k=2}^{n}{(\vec{\mathbf{f}}_k(s)-\vec{\mathbf{f}}_k(t))^2}\cdot\frac{\lambda^{\ell+1}_k}{1-\lambda_k} \right|\\
 = \frac{\lambda^{\ell+1}}{1-\lambda}\cdot \frac{1}{2m}\sum_{k=2}^{n}{(\vec{\mathbf{f}}_k(s)-\vec{\mathbf{f}}_k(t))^2}\\
& \le \frac{\lambda^{\ell+1}}{1-\lambda}\cdot \left( \frac{1}{d(s)}+\frac{1}{d(t)} \right).
\end{align*}
% \end{small}
Plugging Eq. \eqref{eq:ell} into the above inequality yields $|r(s,t)-r_{\ell}(s,t)|\le \tfrac{\epsilon}{2}$, which proves the theorem.
\end{proof}

\subsection{\bf Proof of Lemma \ref{lem:sum-bound}}
\begin{proof}
Consider a length-$\ell_f$ random walk $W$ from $u$ which contains $\ell_f$ visited nodes $\{w_1,w_2,\dotsc,w_{\ell_f}\}$. Note for any $1\le i\le \ell_f$, $w_i$ and $w_{i+1}$ are two adjacent nodes. This implies that for any node $v\in V$, it appears at most $\left\lceil \tfrac{\ell_f}{2}\right\rceil$ times. Let $W_{1}$ contain $\left\lceil \tfrac{\ell_f}{2}\right\rceil$ nodes of $W$ with the largest values with respect to $\vec{\mathbf{x}}(\cdot)$ and $W_{2}$ contain the remaining $\left\lfloor \tfrac{\ell_f}{2}\right\rfloor$ nodes. Hence,
% \begin{small}
\begin{align*}
 \sum_{w \in W}{\vec{\mathbf{x}}(w)}& =  \sum_{w \in W_{1}}{\vec{\mathbf{x}}(w)}  + \sum_{w \in W_{2}}{\vec{\mathbf{x}}(w)}\textstyle \le \left\lceil \tfrac{\ell_f}{2}\right\rceil\cdot\max_1(\vec{\mathbf{x}}) + \left\lfloor \tfrac{\ell_f}{2}\right\rfloor\cdot\max_2(\vec{\mathbf{x}}).
\end{align*}
% \end{small}
Note that $\sum_{w \in W}{\vec{\mathbf{x}}(w)}\ge \ell_f\cdot \min{(\vec{\mathbf{x}})}$ and the lemma follows.
\end{proof}

\subsection{\bf Proof of Theorem \ref{lem:fwd}}
\begin{proof}
	Letting $Z_{k}^\prime := Z_{k}+\frac{\psi}{2}$, we immediately have $0\leq Z_{k}^\prime\leq \psi$. Then, the empirical mean $Z^\prime=Z+\tfrac{\psi}{2}$ and empirical variance remains the same $\hat{\sigma}^2$. Let $T$ be the stopping time. If Algorithm \ref{alg:fwd} stops at the $i$-th iteration such that $i<\tau$, we know $f(\eta,\hat{\sigma}^2,\psi,{\delta}/{\tau})\leq \tfrac{\epsilon}{2}$. Thus, given $i< \tau$, using the Bernstein inequality (Lemma \ref{lem:bernstein}), we obtain
	\begin{align*}
		\mathbb{P}\big[|Z-\mathbb{E}[Z]|> \tfrac{\epsilon}{2}\land T=i\big] & \leq \mathbb{P}\big[|Z-\mathbb{E}[Z]|\geq f(\eta,\hat{\sigma}^2,\psi,\tfrac{\delta}{\tau})\land T=i\big]\\
		&=\mathbb{P}\big[|Z'-\mathbb{E}[Z']|\geq f(\eta,\hat{\sigma}^2,\psi,\tfrac{\delta}{\tau})\land T=i\big]\leq \tfrac{\delta}{\tau}.
	\end{align*}
	In addition, consider that $T=\tau$. According to the Hoeffding's inequality (Lemma \ref{lem:hoeffding}), we have
	\begin{equation*}
		\mathbb{P}\big[|Z-\mathbb{E}[Z]|> \tfrac{\epsilon}{2}\land T=\tau\big]
		\leq 2\exp{\left(-\tfrac{\eta\epsilon^2}{2\psi^2}\right)}\leq 2\exp{\left(-\tfrac{\eta^\ast\epsilon^2}{2\psi^2}\right)}=\tfrac{\delta}{\tau}.
	\end{equation*}
	Putting it together yields
% 	\begin{equation*}
		$\mathbb{P}\big[|Z-\mathbb{E}[Z]|> \tfrac{\epsilon}{2}\big]\leq (\tau-1)\cdot \tfrac{\delta}{\tau} + \tfrac{\delta}{\tau}=\delta$,
% 	\end{equation*}
	which immediately concludes that $\left|r_f(s,t)- q(s,t)\right| \le \frac{\epsilon}{2}$ holds with a probability of at least $1-\delta$, since $r_f(s,t)=Z$ and $q(s,t)=\mathbb{E}[Z]$.
	
	Moreover, when setting $\ell_f=\ell$ as in Eq.~\eqref{eq:ell}, $\vec{\mathbf{s}}=\vec{\mathbf{e}}_s$ and $\vec{\mathbf{t}}=\vec{\mathbf{e}}_t$, by Theorem~\ref{lem:ell}, we know that
% 	\begin{equation*}
		$|r(s,t)-r_{\ell}(s,t)|\le \frac{\epsilon}{2}$.
% 	\end{equation*} 
	Meanwhile, 
	\begin{align*}
		q(s,t)
		&=\sum_{i=1}^{\ell}\left({\frac{p_i(s,s)}{d(s)}+\frac{p_i(t,t)}{d(t)}-\frac{  p_i(s,t)}{d(t)}}-\frac{ p_i(t,s)}{d(s)}\right)=r_\ell(s,t)-\mathbbm{1}_{s\neq t}\cdot(\frac{1}{d(s)}+\frac{1}{d(t)}),
	\end{align*}
	Therefore, with a probability of at least $1-\delta$, we have
	\begin{small}
	\begin{align*}
		\textstyle \left|r_f(s,t)+\mathbbm{1}_{s\neq t}\cdot(\frac{1}{d(s)}+\frac{1}{d(t)})-r(s,t)\right| & \leq \left|r_f(s,t)+\mathbbm{1}_{s\neq t}\cdot(\frac{1}{d(s)}+\frac{1}{d(t)})-r_\ell(s,t)\right| + \left|r_\ell(s,t)-r(s,t)\right|\\
		&\leq \frac{\epsilon}{2}+\frac{\epsilon}{2}={\epsilon}.
	\end{align*}	
	\end{small}
	This completes the proof.
\end{proof}